\documentclass[reqno]{amsart}
\usepackage{hyperref}
\usepackage{graphicx}
\usepackage{curves}
\usepackage{float}
\usepackage{comment}
\usepackage{caption}
\usepackage{subcaption}
\usepackage{dirtytalk}
\usepackage{amsthm}
\usepackage{mathrsfs}
\usepackage{amsbsy}
\usepackage{mathtools}
\usepackage{cool}
\usepackage{stmaryrd}
\usepackage{relsize}
\theoremstyle{plain}
\usepackage{tikz-cd}
\usetikzlibrary{arrows,shapes}
\usetikzlibrary{decorations.text}
\usepackage{pgfplots}
\pgfplotsset{compat=1.14}
\usepackage[percent]{overpic}
\usepackage{soul}

\let\oldtocsection=\tocsection

\let\oldtocsubsection=\tocsubsection

\renewcommand{\tocsection}[2]{\hspace{0em}\oldtocsection{#1}{#2}}
\renewcommand{\tocsubsection}[2]{\hspace{1.75em}\oldtocsubsection{#1}{#2}}

\newcommand*{\mailto}[1]{\href{mailto:#1}{\nolinkurl{#1}}}

\newcommand{\yud}{\mathrel{\substack{\Ydown\\[-.2em]\Yup}}}

\newcommand\restr[2]{{
  \left.\kern-\nulldelimiterspace 
  #1 
  \vphantom{\big|} 
  \right|_{#2} 
  }}

\newtheorem{theorem}{Theorem}[section]
\newtheorem{proposition}[theorem]{Proposition}
\newtheorem{definition}[theorem]{Definition}
\newtheorem{lemma}[theorem]{Lemma}
\newtheorem{corollary}[theorem]{Corollary}
\newtheorem{remark}[theorem]{Remark}
\newtheorem{assumption}[theorem]{Assumption}

\newtheorem{example}[theorem]{Example}

\newtheorem*{fga}{Finite Gap Ansatz}

\newcommand{\R}{{\mathbb R}}
\newcommand{\N}{{\mathbb N}}

\newcommand{\C}{{\mathbb C}}

\newcommand{\nn}{\nonumber}
\newcommand{\be}{\begin{equation}}
\newcommand{\ee}{\end{equation}}
\newcommand{\bea}{\begin{eqnarray}}
\newcommand{\eea}{\end{eqnarray}}
\newcommand{\btheo}{\begin{theorem}}
\newcommand{\etheo}{\end{theorem}}

\newcommand{\re}{\Re}
\newcommand{\im}{\Im}
\newcommand{\asympt}{\mathcal{O}}

\newcommand{\esf}{\mathsf{E}}
\newcommand{\msf}{\mathsf{M}}
\newcommand{\nsf}{\mathsf{N}}

\newcommand{\CO}{{\mathcal O}}

\newcommand{\clos}{\textsf{clos}}
\newcommand{\ph}{\mathrm{ph}}
\newcommand{\thsf}{\mathsf{\theta}}
\newcommand{\omsf}{\mathsf{\omega}}

\newcommand{\lam}{\lambda}

\newcommand{\Om}{\Omega}
\newcommand{\z}{\zeta}

\newcommand{\ro}{\rho}
\newcommand{\3}{\hspace{2pt}{\rm mod}3}

\numberwithin{equation}{section}

\title[WKB for a Dirac Operator with a rapidly oscillating Potential]
{Semiclassical WKB Problem for the non-self-adjoint Dirac operator with 
an analytic rapidly oscillating potential} 

\author[S.Fujii\'e]{Setsuro Fujii\'e$^{\ast}$}
\address{$^{\ast}$ Department of Mathematical Sciences, Ritsumeikan University, Japan}
\email{\mailto{fujiie@fc.ritsumei.ac.jp}}
\urladdr{\url{https://research-db.ritsumei.ac.jp/rithp/k03/resid/S000793?lang=en}}

\author[N. Hatzizisis]{Nicholas Hatzizisis$^{\dag}$}
\address{$^{\dag}$ Department of Mathematics \& Applied Mathematics, University of Crete, 
and 
Institute of Applied and Computational Mathematics, FORTH, Greece}
\email{\mailto{nhatzitz@gmail.com}}
\urladdr{\url{http://www.nikoshatzizisis.wordpress.com/home/}}

\author[S. Kamvissis]{Spyridon Kamvissis$^{\ddag}$}
\address{$^{\ddag}$ Department of Mathematics \& Applied Mathematics, 
University of Crete, and Institute
of Applied and Computational Mathematics, FORTH, Greece}
\email{\mailto{spyros@tem.uoc.gr}}
\urladdr{\url{http://www.tem.uoc.gr/~spyros/}}

\begin{document}

\abstract
In this paper we examine the semiclassical behavior  
of the scattering data of a 
non-self-adjoint Dirac operator with a rapidly oscillating
potential that is complex analytic in some neighborhood of the real line.
Some of our results are rigorous and quite general. 
On the other hand, complete  and concrete understanding requires the  investigation of
the WKB geometry of specific examples. 
For such detailed computations we use a particular  example that has been investigated numerically
more than 20 years ago by Bronski and Miller and rely heavily on their numerical computations.
Mostly employing the exact WKB method, we provide the complete rigorous uniform 
semiclassical analysis of the Bohr-Sommerfeld condition for the location of the 
eigenvalues across unions of analytic arcs as well as the associated norming constants.  
For the reflection coefficient as well as the eigenvalues near 0 in the spectral plane, we 
employ  instead an  older theory that  has
been developed in great detail by Olver.  Our analysis is motivated
by the need to understand  the semiclassical behaviour of the focusing cubic NLS equation 
with initial data $A\exp\{iS/\epsilon\}$,  in view of the well-known fact discovered 
by Zakharov and Shabat that the spectral analysis of the Dirac operator enables  
the solution of the NLS equation via inverse scattering theory.
\endabstract

\maketitle

\tableofcontents

\section{Introduction}
\label{intro}

\subsection{Motivation}
\label{intro-moto}

Consider the \textit{semiclassical limit} ($\epsilon\downarrow0$) of the solution to the initial 
value problem of the one-dimensional \textit{nonlinear Schr\"odinger equation} with cubic 
nonlinearity 
\begin{equation}
\label{ivp-nls}
\left\{
\begin{array}{l}
i\epsilon\partial_t\psi+
\frac{\epsilon^2}{2}\partial_x^2\psi+
|\psi|^2\psi=
0, \quad(x,t)\in\R\times\R\\
\psi(x,0)=A(x)\exp\{iS(x)/\epsilon\},\quad x\in\R
\end{array}
\right.
\end{equation}
where $A$ and $S$ are real valued integrable functions defined on the real line.

According to the seminal discovery in \cite{zs},  this initial value problem can be studied
via the so-called \textit{inverse scattering method}.  In order for one to do so,  the first  
ingredient is the $direct$ spectral and scattering analysis of the associated \textit{Dirac} 
(or Zakharov-Shabat) operator $\mathfrak{D}_\epsilon$ given by
\begin{equation*}
\mathfrak{D}_\epsilon=
\begin{bmatrix}
-\frac{\epsilon}{i}\frac{d}{dx} & -iA(x)\exp\{iS(x)/\epsilon\}\\
-iA(x)\exp\{-iS(x)/\epsilon\} & \frac{\epsilon}{i}\frac{d}{dx}
\end{bmatrix}.
\end{equation*}

Our main goal in this paper is to provide a rigorous investigation of the semiclassical 
behavior of the \textit{scattering data} (reflection coefficient,  eigenvalues and their 
associated norming constants) for the operator $\mathfrak{D}_\epsilon $ and apply 
our conclusions to the semiclassical investigation of the solutions to the focusing 
NLS equation.

We are interested in $A(x)$ and $S(x)$ that are {\it real analytic}  functions 
which can be extended holomorphically to at least a region in the whole complex plane
\footnote{For the \textit{inverse scattering method} to  be applicable we will also need that  $S'$
is integrable for real $x$.}. Although a good deal of our analysis will be fairly general, we will eventually 
focus on a very particular choice of $A(x)$ and $S(x)$ so that we have a concrete  configuration of the 
geometry of turning points and Stokes lines. Our  model will be the  $A(x)=S(x)=\sech(2x)$ 
since this has been the first case considered in detail: a very careful numerical analysis of 
the eigenvalue problem was done by Bronski \cite{bron} in $1996$ and a first investigation of 
the formal WKB theory, involving  numerical computations of turning points and Stokes lines, 
was conducted by Miller \cite{mil} in $2001$.  Another model has been studied numerically 
recently in \cite{kcv}; there,  the authors use $A(x)= 30/(1+x^4)$ and $S(x)=15x^2$.

If the parameter $\epsilon$ is not too small, 
then of course the solution is described by  a \say{nonlinear} superposition
of \say{breathers}  (corresponding to purely imaginary eigenvalues of $\mathfrak{D}_\epsilon$), 
traveling solitons  (corresponding to all other  eigenvalues of $\mathfrak{D}_\epsilon$)
and a \say{background radiation contribution}  (corresponding to the continuous spectrum of 
$\mathfrak{D}_\epsilon$). This is in a sense the content of the inverse scattering method.

On the other hand, as $\epsilon\downarrow0$, the solution attains a special form
which in particular regions attains a highly oscillatory behavior.
Often in such problems one expects the validity of the  so-called \textit{finite gap ansatz} 
for the  semiclassical  asymptotics. 
\begin{fga}
Let $(x_0, t_0)$ be a generic  point with  $x_0 \in \mathbb{R}$ and $t_0 > 0$.
The solution $\psi$ of (\ref{ivp-nls}) is asymptotically ($\epsilon\downarrow0$)
described  (locally) as a slowly modulated $G$ phase wavetrain.  More precisely, 
setting $x=x_0+\epsilon  \hat{x}$ and $t=t_0+\epsilon  \hat{t}$ (so that $x_0, t_0$ are 
``slow" variables while $\hat{x}, \hat{t}$ are ``fast" variables), there exist parameters
(all of which depend on the slow variables $x_0$, $t_0$ but not on $\hat{x}, \hat{t}$)
\begin{itemize}
\item
$a, U_0, k_0, w_0$ in $\C$ and
\item
${\bf U}=[U_1, \dots,U_G]^T$, 
${\bf k}=[k_1, \dots, k_G]^T$, 
${\bf w}=[w_1, \dots, w_G]^T$, 
${\bf Y}=[Y_1, \dots, Y_G]^T$, 
${\bf Z}=[Z_1, \dots, Z_G]^T$ 
in $\C^G$
\end{itemize}
such that as $\epsilon\downarrow0$, 
\be\nn
\psi(x,t)= \psi(x_0+ \epsilon  \hat{x}, t_0+\epsilon\hat{t}\hspace{2pt}) 
\ee
has the following leading order asymptotics 
\begin{multline}
\label{asymptotics}
\psi(x,t) \sim
a(x_0, t_0) 
\exp\Big\{
i\Big(\tfrac{U_0(x_0, t_0)}{\epsilon}+k_0(x_0, t_0)\hat{x}-w_0(x_0, t_0)\hat{t}\Big)
\Big\}\\
\cdot
\frac{\Theta\bigg({\bf Y}(x_0, t_0)+i\Big(\frac{{\bf U}(x_0, t_0)}{\epsilon} +
{\bf k}(x_0, t_0) \hat{x}-{\bf w}(x_0, t_0)\hat{t}
\Big)\bigg)}
{
\Theta\bigg({\bf Z}(x_0, t_0)+i\Big(\frac{{\bf U}(x_0, t_0)}{\epsilon} +
{\bf k}(x_0, t_0) \hat{x}-{\bf w}(x_0, t_0)\hat{t}
\Big)\bigg)}.
\end{multline}
\end{fga}

All parameters here are  defined in terms of an underlying
\textit{Riemann surface} $X$ which depends solely on $x_0, t_0$.   The moduli of 
$X$ vary slowly with $x, t$, i.e. they depend on  $x_0, t_0$ but not on
$\epsilon, \hat{x}, \hat{t}$.  $\Theta$ is the G-dimensional \textit{Jacobi theta 
function} associated with $X$.  The genus of $X$ can vary with $x_0, t_0$.  
In fact, the $(x,t)$-plane is  divided into open regions in each of which 
G is constant. 

On the boundaries of such regions (sometimes called ``caustics"; 
they are unions of analytic arcs),  some degeneracies appear in the mathematical 
analysis (we may have ``pinching" of the surfaces $X$ for example).

The above formula (\ref{asymptotics}) gives asymptotics which is  pointwise in $x,t$. 
We refer to \cite{kmm}  and \cite{kr}, for the justification of the above ``finite gap" 
formula in the case of special ``bell-like"  initial data $A(x)$ [and $S(x)=0$], as well as 
the  exact formulae for the parameters and the definition of the theta functions. The formula 
is actually uniformly valid in compact $(x,t)$-sets not containing points on the caustics.
The mathematical theory leading to the asymptotic formula is the well-known non-linear 
steepest descent method. Near the caustics,  boundary layers separating regions of different  
genus appear. For an analysis of the somewhat more delicate behavior (especially for higher 
order terms in $\epsilon$) near the first caustic one can consult \cite{bt}.
  
The following two pictures (prepared for us by \textit{Nikos Efremidis})
exemplify this behavior in two separate cases. In both cases $A(x)=2\sech x$; but
in the first picture $S$ is identically 0 while in the second one $S(x)=\sech x$. Here
$\epsilon =0.1$ and the time is actually rescaled (divided by 0.1).
We only show plots for $|\psi|^2$.

One observes the following:
\begin{itemize}
\item[(i)]
There is no qualitative difference between the two cases, at least for not very large times.
\item[(ii)]
The behavior of the solution is quite different in  three distinct regions in the $(x, t)$-plane.
\item[(iii)]
In the first region (for smaller times t) things are fairly ``smooth". There are no oscillations.  
In the intermediate region fast oscillations appear. The curve separating these two regions 
appears to be fairly continuous.  There is also a third region where the nature of the fast  
oscillations changes.  Again the intermediate and third regions are separated by a seemingly  
continuous curve.
\end{itemize}

\begin{figure}[htbp]
\begin{center}
\includegraphics[width=100mm]{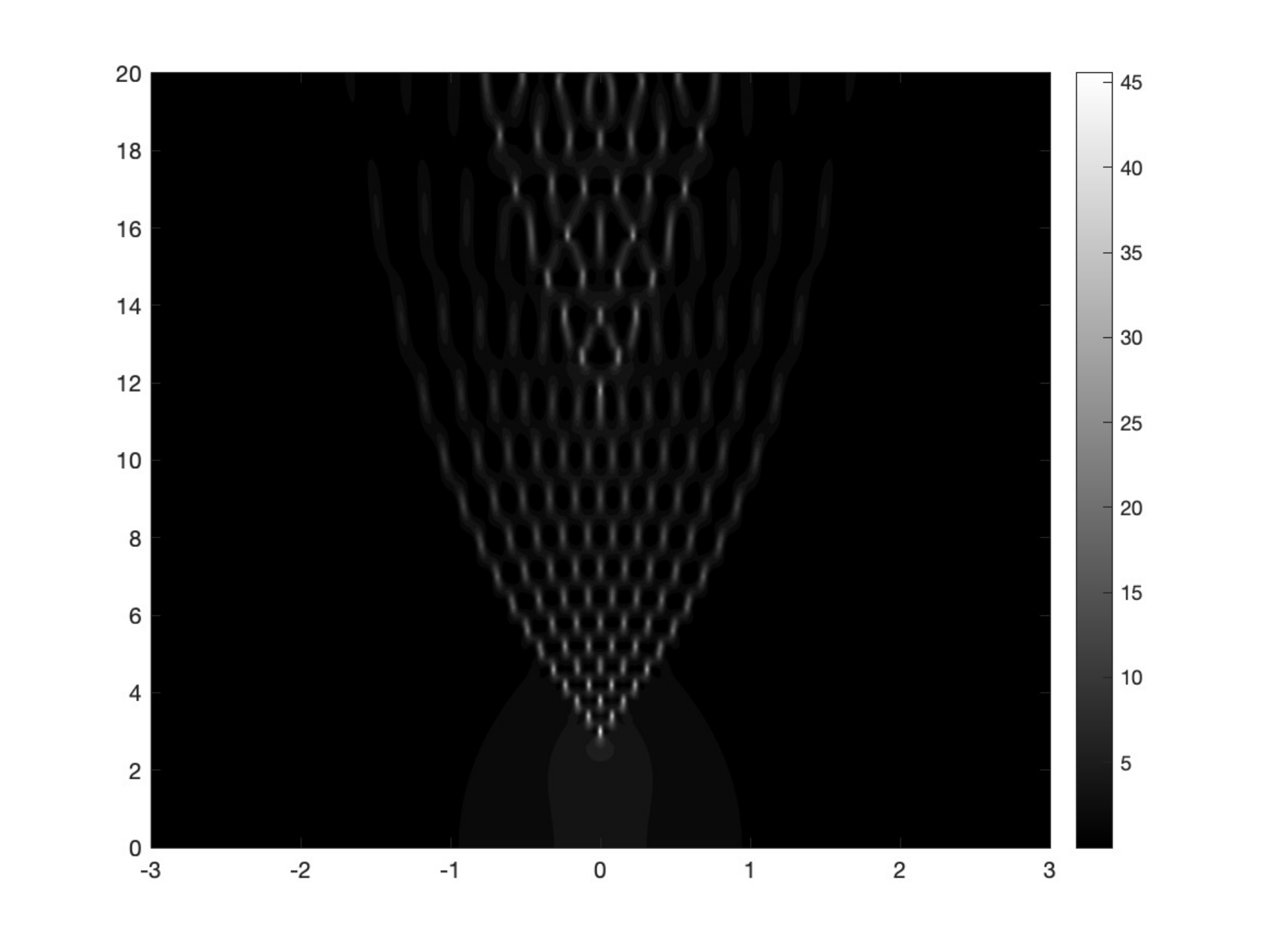}
\end{center}
\caption{
A solution of the initial value problem (\ref{ivp-nls}) for $\epsilon=0.1$ with initial 
data $\psi(x,0)=2\sech x$ [i.e.  $A(x)=2\sech x$ and $S(x)=0$].  In this figure 
$|\psi|^2(x,t)$ is being plotted.  The horizontal axis represents $x$ while $t$ runs 
on the vertical axis (as shown on the left).  The bar on the right of the plot shows 
a graduated scale for the values of $|\psi|^2$ ranging from deep black 
(where $|\psi|^2$ is small) to bright white (where $|\psi|^2$ is big).}
\label{efrem-zero-phase}
\end{figure}

\begin{figure}[htbp]
\begin{center}
\includegraphics[width=100mm]{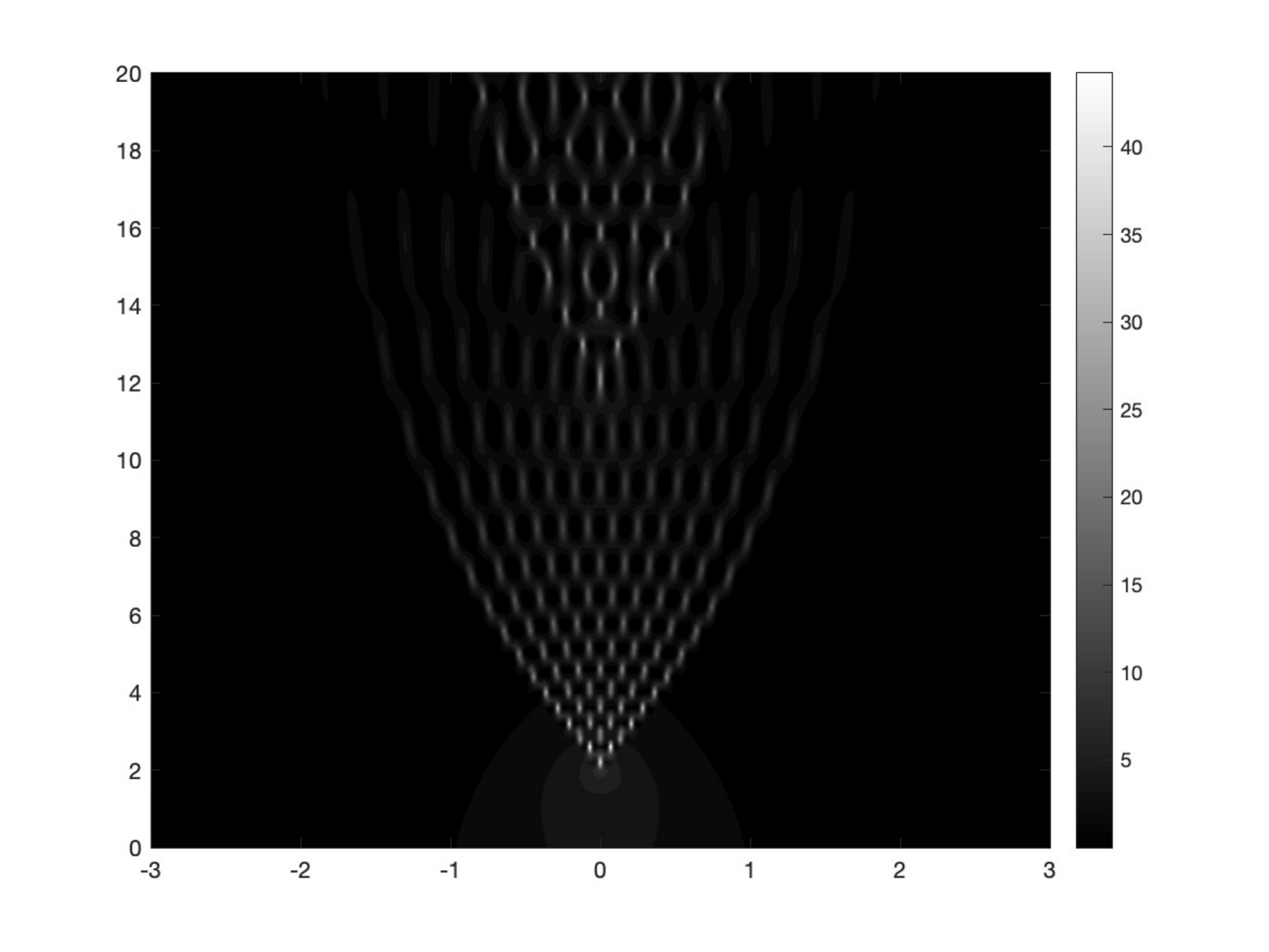}
\end{center}
\caption{
A solution of the initial value problem (\ref{ivp-nls}) for $\epsilon=0.1$ with initial data 
$\psi(x,0)=2(\sech x)\exp\big\{\tfrac{\sech x}{0.1}\big\}$ 
[i.e.  $A(x)=2\sech x$ and $S(x)=\sech x$].  
The figure  shows the plot of $|\psi|^2(x,t)$.  On the horizontal axis lies 
$x$ while $t$ runs on the vertical axis (as shown on the left).  
Also on the right of the plot there is a bar showing a graduated scale
for the values of $|\psi|^2$ ranging from deep black (where $|\psi|^2$ is small) 
to bright white (where $|\psi|^2$ is big).}
\label{efrem-non-zero-phase}
\end{figure}

In \cite{kmm} we have been able to identify the first region in Figure \ref{efrem-zero-phase} 
with the genus 0 region. The intermediate region is the genus 2 region  and then comes the  
genus 4 region.  The separating curves are the ``caustics".  In principle, there can be a very 
large number of such caustics.  But it is also worth pointing out that the $2\sech x$ data,  
being a nonlinear superposition of many (of order $\epsilon^{-1}$) breathers, is actually  
periodic in time with period $\asympt(\epsilon^{-1})$.  Of course, this is not seen in our 
picture,  since what happens at such large times is not shown.

If we were to run our numerics in Figure \ref{efrem-non-zero-phase} for larger times,  
we would also observe many separated traveling solitons at different non-zero speeds,  
since (as we shall see later) the initial data give rise to many eigenvalues with non-zero 
real part.

The question of the semiclassical approximation of the scattering data has a 
deeper significance in view of the instabilities of the problem which apper in 
many levels.  In fact even in the non-semiclassical regime, the {\it focusing} NLS 
is the main model for the so-called ``modulational instability" (see \cite{bf} and \cite{bm}),  
although for positive fixed $ \epsilon $ the initial value problem is well-posed.
Semiclassically the  instabilities become more pronounced.  One way to see this, 
is related to the underlying ellipticity of the formal semiclassical limit.  To be more 
specific,  consider the well-known \textit{Madelung transformation} (see \cite{m})
\begin{equation*}
\begin{cases}
\ro= |\psi|^2\\
\mu= \epsilon\im (\psi^*\psi_x)
\end{cases}
\end{equation*}
where $\psi^*$ denotes the complex conjugate of $\psi$.
Then the initial value problem (\ref{ivp-nls}) becomes
\begin{equation*}
\begin{cases}
\ro_t  +\mu_x = 0\\
\mu_t + \Big(\frac{\mu^2}{\ro} + \frac{\ro^2}{2}\Big)_x = 
\frac{\epsilon^2}{4} \partial_x[\ro (\log \ro)_{xx}]\\
\ro(x,0)=A^2(x)\\
 \mu(x,0)=A^2(x)S'(x).
\end{cases}
\end{equation*}
The formal limit, as $\epsilon\downarrow0$,  is
\begin{equation*}
\begin{cases}
\ro_t  +\mu_x = 0\\
\mu_t + \Big(\frac{\mu^2}{\ro} + \frac{\ro^2}{2}\Big)_x = 0\\
\ro(x,0)=A^2(x)\\
 \mu(x,0)=A^2(x)S'(x).
\end{cases}
\end{equation*}
This is an initial value problem for an elliptic system of equations; and so one expects 
that small perturbations of the initial data (independent of $\epsilon$) can lead to large 
changes in the solution at any given time.

Instabilities also appear, independently, at the spectral analysis of the related 
non-self-adjoint Dirac operator: indeed, small changes of the potential are expected 
to result in relatively large changes in the discrete spectrum;  this is not true when the 
phase $S=0$ but it is true otherwise. We refer to \cite{bron2} for a numerical investigation 
of this fact in the  special case $A(x)=S(x)=sech(2x)$ and to \cite{den} for the discussion of the more general 
phenomenon in the theory of pseudo-differential operators. Instabilities are evident also 
at the related \textit{equilibrium measure problem} (it is a ``max-min" problem; see \cite{kr}),  
the related \textit{Whitham equations} (they are also elliptic) and even in the numerical 
studies of the problem.  What was initially called the modulational instability is just the tip of the iceberg!
  
One might even question whether the whole project is worthwhile studying in detail,  
at least as a valid physical model, in view of all these instabilities.  To this,  a first response is 
that the semiclassical analysis turns out in practice to be relevant even for not so small 
values of the semi-classical parameter $\epsilon$ (for example in applications to nonlinear 
optics).  From  a mathematical point of view,  it is 
a fascinating  instance of a difficult unstable 
problem that can be approximated by solvable  stable problems; this approximation is very 
singular and non-trivial and its study involves important connections to different areas of 
mathematics (namely \textit{PDE theory},  \textit{spectral theory of non-self-adjoint operators},  
\textit{WKB analysis},  \textit{potential theory} (cf. \cite{kr}),  \textit{Riemann surface theory},   
\textit{hydrodynamic instability theory}).

Going back to the pictures presented above, in the first (genus 0) region there are no 
fast oscillations and actually strong semiclassical limits exist for both $\rho$ and $\mu$.  
They actually satisfy the formally limiting PDE.  In the other (higher genus) regions, violent 
oscillations of frequency order $\epsilon^{-1}$ appear and a strong pointwise limit does not 
exist.  But to the extent that we can perform a full asymptotic analysis, both at the direct  
scatterring and the inverse scattering stage,  we can show that there is at least a weak limit 
and we are able to provide complete asymptotic formulae (as above).

The semiclassical analysis  of the NLS solution $\psi$ in \cite{kmm} undertook  the 
asymptotic ($\epsilon\downarrow0$) analysis of the inverse scattering problem, to 
which Zakharov and Shabat have reduced the solution of the equation; in particular,  
we worked on the formulation of that problem as a \textit{Riemann-Hilbert factorisation 
problem} and we applied and extended ideas and calculations going back to the seminal 
work of \textit{Deift, Venakides  and Zhou} (\cite{dz}, \cite{DVZ}).  But no careful 
semiclassical analysis of the $direct$ scattering problem had been undertaken until 
recently.  Instead, an ad hoc approximation of the eigenvalues by their 
\textit{Bohr-Sommerfeld approximants} was used as a starting point (and the reflection  
coefficient was set identically to 0, in the same formal spirit).

In \cite{fujii+kamvi} the rigorous semiclassical analysis of the scattering data of the 
related Dirac (or Zakharov-Shabat) operator was completed in the case where $A$ 
is real analytic,  integrable,  positive, symmetric,  with only one local maximum (and where 
for simplicity the second derivative of $A$ is non-zero). We applied the so-called 
\textit{exact WKB theory}, which will also be applied in the present work.

In \cite{h+k} the analyticity assumption was replaced by a mild smoothness assumption 
and a different method was employed,  going back to \textit{Langer} and \textit{Olver} 
\cite{olver1975}. In a sequel \cite{h+k2},  the general case with several local maxima and 
minima was also completed, under the assumption that potentials are smooth and positive. 
In all the above cases the initial phase $S(x)$ is identically zero. Of course, even with an 
initial phase zero, the Madelung system above shows that a non-trivial phase will appear  
immediately for any small time $t$. But it is also reasonable to postulate a non-zero initial 
phase especially in cases where the discrete spectrum is no more imaginary and the 
instability due to the non-self-adjointness of the Dirac operator is instrumental.

The exact WKB method was first developed for the \textit{Schr\"odinger operator},  
but here we apply it to the Dirac operator that is associated to the 
\textit{focusing NLS equation}. The method goes back to works of \textit{Ecalle} 
\cite{e} and \textit{Voros} \cite{v} but here we argue along the lines of the papers of 
\textit{G\'erard-Grigis} \cite{gg} and \textit{Fujii\'e-Lasser-N\'ed\'elec} \cite{fln}.  
Rather than relying on the usual formal WKB method which relies on asymptotic 
series that are in general divergent,  we use a ``resummation" of the series and in 
fact construct ``exact solutions" in terms of convergent series,  thus resolving a 
problem of ``asymptotics beyond all orders".

For the study of the eigenvalues with small imaginary part,  the exact WKB method seems 
to break down and Olver's  \cite{olver1975} is not directly applicable since we don't know 
a priori that the eigenvalues are purely imaginary (as in \cite{h+k}). Fortunately the somewhat  
forgotten paper \cite{olver1978} which generalizes his  earlier work is proved to be useful 
here.

\subsection{Organization of the paper}
 
The plan of this paper is the following. In the next section we  investigate the eigenvalues 
of the Dirac operator that in the semiclassical limit lie away from the real axis. 
In \S \ref{gen-strat},  we present the main ideas that  will be used in the following 
sections, including the definition of \textit{turning points}, \textit{progressive paths}, 
\textit{Stokes lines}, \textit{admissible contours} and \textit{asymptotic spectral arcs}. 
In \S\ref{exact-wkb}, we introduce the theory of the exact (resummed and  converging) 
WKB solutions and state a basic theorem about convergence and semiclassical asymptotics. 
Then in \S \ref{exact-wkb-connection-1-tp} we prove a theorem that shows how different 
solutions are connected in a neighborhood of a simple turning point. We use it in 
\S \ref{bs-at-two-turn-pts},    to give a rigorous justification of the Bohr-Sommerfeld 
asymptotic conditions for the location of the eigenvalues that lie away from the bifurcation 
point (see Figure \ref{cut-the-L}).  In \S\ref{near-the-bf} we refine the exact WKB analysis in 
a small neighborhood of a bifurcation point.  In the last section of this paragraph, 
namely \S \ref{geometry}, we apply all the results above to the particular case  
$A(x)=S(x)=\sech(2x)$, where we make heavy use of the numerical results and figures 
of \cite{mil}.

We are not able at this point to extend the exact WKB method  for the eigenvalues near $0$.
In paragraph \ref{evs-near-zero}, we use instead Olver's theory as it was extended in 
\cite{olver1978}. We start in \S \ref{preparations-near-0} with some preparatory material and 
move to \S \ref{sols-1-turn-pt} where we encounter the behavior of solutions near a simple
turning point (compare this with \S \ref{exact-wkb-connection-1-tp}). Next, comes 
\S \ref{asymptt-estim-for-error} where we present asymptotics for the \textit{error-terms} in our 
solutions.  The connection of the above solutions -in the way of Olver- is achieved in 
\S \ref{olver-connection}.  Using the asymptotics for the error-terms, we arrive 
at the asymptotic form of these connection formulas in \S \ref{connection-asymptotics}. 
We then assemble all the previous results, and in \S \ref{application-of-connection} we provide 
a Bohr-Sommerfeld quantization condition for the location of the eigenvalues that lie near zero. 

In paragraph \ref{norming-constants} we are interested in the behavior of the corresponding 
norming constants to the eigenvalues discussed previously. Then in paragraph \ref{reflection}, 
we use Olver's method as exploited in \cite{h+k}, to study the behavior of the reflection coefficient; 
we prove that it is  small away from the point 0 (in the spectral real line) in \S \ref{away-0-refle}. 
Also, in \S \ref{near-0-refle}, we give asymptotics for the reflection coefficient nearer to zero. 

Finally, in paragraph \ref{application-nls},  we present the application of the WKB results to 
the focusing NLS problem and explain how the analysis of \cite{kmm}, \cite{kr} extends to  our present case
 in view of these results.  

Since the approximation technique we are using for the near-zero eigenvalues relies on 
modified Bessel functions, we have included in the appendix a paragraph that contains all 
the material that we shall need.

\subsection{Notation}
\label{notation}

Before we start our main exposition, we specify  some notation used throughout 
our work.

\begin{itemize}
\item
The symbol $\N_0$ is used for the set $\N\cup\{0\}$.
\item
The Riemann sphere is denoted by $\overline{\C}=\C\cup\{\infty\}$.
\item
The (topological) closure of a set $S\subseteq\overline{C}$ shall be denoted by 
$\clos(S)\subseteq\overline{C}$.
\item
We use the notation $\R^\pm=\{\,x\in\R\mid\pm x>0\,\}$ for the sets of 
positive/negative numbers.
\item
Complex conjugation is denoted with a star superscript, \say{$*$}; 
i.e. $z^*$ is the complex conjugate of $z$.
\item
The transpose of a matrix $M$ is denoted by $M^T$.
\item
For the upper half-plane we write $\mathbb{H}^+=\{\,z\in\C\mid\Im z>0\,\}$.
\item
For a set $\Sigma\subseteq\R$, when we write $i\Sigma$ we mean 
the set $\{i\kappa\mid\kappa\in\Sigma\}\subset\C$. Also, for 
$z_1, z_2\in\C$ the set $[z_1,z_2]\subset\C$ denotes the closed line segment 
starting at $z_1$ and ending at $z_2$.
\item
The set of holomorphic functions from a domain $A$ to a domain $B$ shall be denoted 
by $\mathscr{H}(A;B)$. For the holomorphic functions from $A$ to $\C$, we simply write 
$\mathscr{H}(A)$.
\item
The symbol $\sigma(a\rightsquigarrow b)$ is used to denote a
curve in $\C$,  parametrized by $\sigma$,  starting at $a$ and ending
at $b$.
\item
Let ${\mathcal W}[\textbf{f},\textbf{g}]$ be defined as the determinant
of the $2\times2$ matrix $[\textbf{f}\hspace{3pt}\textbf{g}]$ consisting of
the two column vectors $\textbf{f}$, $\textbf{g}\in\C^2$.
\item
The wronskian of the pair $\{f,g\}$ of functions is represented by $W[f,g]$.
\item
For a continuous function $f$ in $\R$, the notation $\int f(t)dt$ denotes any of its 
primitives. 
\item
If $z\in\C$, we denote by $\text{ph}z\in(-\pi,\pi]$ its phase.
\item
The abbreviations LHS and RHS stand for left-hand side and right-hand side 
respectively.
\item
EV and EF stand for eigenvalues and eigenfunctions respectively.
\item
The acronym MBF means modified Bessel function.
\end{itemize}

\section{Eigenvalues Away From The Real Axis}
\label{evs}

\subsection{General strategy}
\label{gen-strat}
This paragraph  presents a general scheme (see \cite{mil}, and also \cite{fr} 
for the case of Schr\"odinger operators) for finding 
the semiclassical EVs of the Dirac (or Zakharov-Shabat) operator
\begin{equation}
\label{dirac}
\mathfrak{D}_\epsilon=
\begin{bmatrix}
-\frac{\epsilon}{i}\frac{d}{dx} & \omega(x,\epsilon)\\
-\omega^{*}(x,\epsilon) & \frac{\epsilon}{i}\frac{d}{dx}
\end{bmatrix}
\end{equation}
where $\omega$ is the complex potential function
\be\nn
\omega(x,\epsilon)=-iA(x)\exp\Big\{i\frac{S(x)}{\epsilon}\Big\}
\ee
for some real-valued, analytic  functions $A(x)$ and $S(x)$ 
defined on the real line 
[$\omega^{*}(x,\epsilon)$ represents the complex conjugate of 
$\omega(x,\epsilon)$] such that $A$, $S$ and $S'$ are integrable. As previously 
stated in the introduction, we shall be interested in the spectral (generalized eigenvalue) 
problem
\begin{equation}
\label{ev-problem}
\mathfrak{D}_\epsilon\textbf{u}=\lambda\textbf{u} 
\end{equation} 
where $\textbf{u}=[u_1\hspace{3pt}u_2]^T$ is a function from $\R$ to $\C^2$ 
and $\lambda\in\mathbb{C}$ plays the role of the spectral parameter.  

The operator 
$\mathfrak{D}_\epsilon$ is non-self-adjoint and the eigenvalues are distributed 
in the complex plane. But in the semiclassical limit,  their location is restricted 
close to the \textit{numerical range} $\mathcal R$ of the semiclassical symbol.

We shall see [see the reasoning that follows, eventually leading to \eqref{dirac-reduced}] 
that the eigenvalues of $\mathfrak{D}_\epsilon$ coincide with those of the operator
\be
\label{inter-operator}
\begin{bmatrix}
-\frac\epsilon i\frac d{dx}-\frac 12 S'(x) & -iA(x)\\
-iA(x) & \frac\epsilon i\frac d{dx}-\frac 12 S'(x)
\end{bmatrix}.
\ee
We have the following definition.
\begin{definition}
The \textbf{numerical range} $\mathcal R$ of the symbol
$$
\begin{bmatrix}
-\xi-\frac 12 S'(x) & -iA(x)\\
-iA(x) & \xi-\frac 12 S'(x)
\end{bmatrix}
$$
of the matrix/operator in (\ref{inter-operator}) is defined to be the union of the 
set of values of the eigenvalues $-\frac 12S'(x)\pm\sqrt{\xi-A(x)^2}$ of this matrix 
when $(x,\xi)$ varies in the whole phase space $\R^2$. 
\end{definition}
We see that
\begin{align}
\nn
\mathcal R
&=
\R\cup {\mathcal R}_0\quad\text{where}\quad\\ 
\label{r0}
{\mathcal R}_0
&=
-\frac 12\Big[\inf_{x\in\R} S'(x), \sup_{x\in\R}S'(x)\Big]\times
\Big[-\sup_{x\in\R}|A(x)|, \sup_{x\in\R}|A(x)|\Big].
\end{align}
The following result concerning the semiclassical eigenvalues is well-known; see e.g. \cite{den}. 
\begin{proposition}
For any compact set $K\subset\R^2$ with $K\cap {\mathcal R}_0=\emptyset$, 
there exists $\epsilon_0>0$ such that there is no eigenvalue in $K$ for 
$0<\epsilon<\epsilon_0$.
\end{proposition}

It is clear that (\ref{ev-problem}) can be written equivalently as a first order 
system of differential equations,  namely 
\begin{equation}
\label{miller-ev}
\frac{\epsilon}{i}\frac{d}{dx}\textbf{u}(x,\lambda,\epsilon)=
K(x,\lambda,\epsilon) \textbf{u}(x,\lambda,\epsilon) 
\end{equation} 
where 
\begin{equation}
\label{matrixN} 
K(x,\lambda,\epsilon)= 
\begin{bmatrix} 
-\lambda & \omega(x,\epsilon)\\ 
\omega^{*}(x,\epsilon) & \lambda \end{bmatrix}.  
\end{equation} 

Applying first the transformation 
\begin{equation*} 
\textbf{u}(x,\lambda,\epsilon)= 
\begin{bmatrix} 
\exp\big\{i\frac{S(x)}{2\epsilon}\big\} & 0\\ 
0 & \exp\big\{-i\frac{S(x)}{2\epsilon}\big\}
\end{bmatrix} 
\tilde{\textbf{v}}(x,\lambda,\epsilon) 
\end{equation*} 
where $\tilde{\textbf{v}}=[\tilde{v}_1\hspace{3pt}\tilde{v}_2]^T$,  takes the system in 
(\ref{miller-ev}) to a new form 
\begin{equation}
\label{dirac-reduced}
\frac{\epsilon}{i}\frac{d}{dx}\tilde{\textbf{v}}(x,\lambda,\epsilon)=
\begin{bmatrix} 
-\lambda-\frac{1}{2}S'(x) & -iA(x)\\ 
iA(x) & \lambda+\frac{1}{2}S'(x)
\end{bmatrix} 
\tilde{\textbf{v}}(x,\lambda,\epsilon) 
\end{equation} 
where prime denotes differentiation with respect to $x$.  Next, we apply the mapping
\begin{equation*}
\tilde{\textbf{v}}(x,\lambda,\epsilon)=
\begin{bmatrix}
1 & 1\\
-1 & 1
\end{bmatrix}
\textbf{v}(x,\lambda,\epsilon)
\end{equation*}
where
$\textbf{v}=[v_1\hspace{3pt}v_2]^T$, to finally express the initial system as 
\begin{equation}
\label{miller-final}
\frac{\epsilon}{i}\frac{d}{dx}\textbf{v}(x,\lambda,\epsilon)
=M(x,\lambda)\textbf{v}(x,\lambda,\epsilon)
\end{equation}
where
\be
\label{mu-matrix}
M(x,\lambda)=
\begin{bmatrix}
0 & g_+(x,\lambda)\\
-g_-(x,\lambda) & 0
\end{bmatrix}
\ee
\begin{equation}
\label{gigi}
\text{and}\quad g_{\pm}(x,\lambda)=
\mp[\lambda+\frac{1}{2}S'(x)\pm iA(x)].
\end{equation}

The important function to study  is
\begin{align}
\label{turn-pt}
V_0(x,\lambda)
&=
\det M(x,\lambda)\nn\\
&=
g_{-}(x,\lambda)g_{+}(x,\lambda)\nn\\
&=
-[\lambda+\tfrac{1}{2}S'(x)]^2-A^{2}(x).
\end{align}
The zeros of this function play an important role. So we state the following definition.
\begin{definition}
\label{definition-turn-pt}
For a fixed value of $\lambda\in\mathbb{C}$,  the zeros of $V_0(\cdot,\lambda)$
in $\mathbb{C}$ are called  \textbf{turning points} of (\ref{miller-final}).
\end{definition}

In our application, the potential function 
$V_0(\cdot,\lambda)$ in (\ref{turn-pt}), i.e.
\be\nn
V_0(x,\lambda)=
-[\lambda+\tfrac{1}{2}S'(x)]^2-A^{2}(x)
\ee
is complex-valued for $\lambda\in\C\setminus\R$ and the turning points are in 
general complex. Assuming that the functions $A(x)$ and $S(x)$ extend analytically 
to some fixed complex neighborhood $\Omega_0$ of the real axis, we may consider the asymptotics 
of solutions of (\ref{miller-ev}) [eventually of (\ref{miller-final})] on some contour 
(in the complex $x$-plane), other than the real $x$-axis, connecting $-\infty$ to $+\infty$ 
and forcing that contour to pass through at least one pair of complex turning points.  

Let us define
$$
z(x,\lambda,\alpha)=i\int_{\gamma(\alpha\rightsquigarrow x)}
\sqrt{-V_0(t,\lambda)}dt
$$
for a fixed point $\alpha\in \Omega_0$.
\begin{definition}
\label{progress-path}
A path $\gamma$ in $\Omega_0$ along which $\re z$
is 
strictly monotone will be called a \textbf{progressive path}. 
\end{definition}
\begin{remark}
This notion does not depend on the branch of the square root nor the base point $\alpha$.
This is equivalent to  $\gamma$ been  {\it transversal} to the Stokes curves  defined later in Definition \ref{stokes-curves-gg}.
When the branch of the square root and the orientation of the path are specified, we also say 
a path is $+$-progressive (resp. $-$-progressive) if $\re z$
is strictly increasing (resp. decreasing). 
\end{remark}

It is important to notice that, in our setting where $S'$ and $A$ tend to 0 at infinity, 
$\sqrt{-V_0(t,\lambda)}$ tends to $\pm \lambda$.
Hence if $\im \lambda$ is non-zero, then the real axis is progressive near infinity. 

The level curves of $\Re[z(x,\lambda,\alpha)]$ in the complex $x$-plane are very important 
for what follows. Indeed, $(\pm)$-progressive curves are always transversal to those.  For that, 
they deserve a definition.
\begin{definition}
\label{stokes-curves-gg}
The level curves of $\Re[z(x,\lambda,\alpha)]$ in the complex
$x$-plane are called the \textbf{Stokes lines} (or \textbf{Stokes curves}) of 
the system (\ref{miller-final}).  
\end{definition}

The geometric configuration of the Stokes curves is crucial in the investigation of 
the domain of validity of the asymptotic expansion of the WKB solutions. Notice that 
from a simple turning point [i.e.  a simple zero of the function $V_0(\cdot,\lambda)$ 
in (\ref{turn-pt})], exactly three Stokes lines emanate.  At such a point,  the angles 
between two Stokes curves are all $2\pi/3$.

Fix $\lambda\in\C\setminus \R$. With the above intuition in mind,  we have the 
following definition.
\begin{definition}
\label{x-appropriate-admissible-curve}
Let  $\{x_-(\lambda)$,  $x_+(\lambda)\}$ be a pair of simple complex roots of 
$V_0(x,\lambda)=0$.  A contour $C=C^-\cup C^0\cup C^+$ in $\Omega_0$
shall be called \textbf{admissible} if the following four conditions hold true:
\begin{itemize}
\item[(i)] 
$V_0$ is holomorphic in $x$, in the region of the
complex $x$-plane enclosed by $C$ and the real $x$-axis (this is needed to ensure
that the approximate eigenfunctions can be continued back to the real $x$-axis).
\item[(ii)]
The contour $C^-$ is a progressive path from $x_-(\lambda)$ to $-\infty$ which 
coincides with the real axis for  $-\re x\gg 1$.
\item[(iii)]
The contour $C^+$ is a progressive path from $x_+(\lambda)$ to $+\infty$ which 
coincides with the real axis for  $\re x\gg 1$.
\item[(iv)]
The contour $C^0$ is a path from $x_-(\lambda)$ to $x_+(\lambda)$ along which 
$\re z(x,\lambda, x_-(\lambda))$ vanishes identically; in other words, $C^0$ is a 
Stokes line.
\end{itemize}
\end{definition}

\begin{remark}
The term \say{admissible} also appears in \cite{mil} but here, we impose 
a stronger requirement on $C^-$ and $C^+$. This stronger definition is 
directly applicable to the exact WKB method of the next section, and was 
first used for the study of quantum resonances of the Schr\"odinger operator 
in \cite{fr}.
\end{remark}


Item (iv) in the list above can be reinterpreted as a differential equation for the path 
$C^0$ in the complex $x$-plane.  Indeed, if $x=u+iv$, where $u,v\in\mathbb{R}$, then 
a field of curves is defined by the differential relation
\begin{equation}
\label{de-real-path}
\Re\bigg\{i\sqrt{-V_0(u+iv,\lambda)}(du+idv)\bigg\}=0.
\end{equation}

The idea guiding the definition above is that the set of $\lambda$ for which
such admissible contours exist  is (hopefully) the locus of points in the complex plane that 
attracts the actual eigenvalues as $\epsilon \to 0$. This cannot be true in general,
as we shall see soon, but an appropriately extented generalisation of this statement 
probably is. See the last remark \ref{finite-gap-remark} of this section.

Suppose now that for a pair of complex turning points $\{x_-(\lambda),x_+(\lambda)\}$ 
we have the condition
\begin{equation}
\label{lambda-relation}
\Re\bigg\{i
\int_{C^0(x_-(\lambda)\rightsquigarrow x_+(\lambda))}
\sqrt{-V_0(t,\lambda)}dt
\bigg\}
=0
\end{equation}
and let  $\lambda$ vary. Given a pair of turning points depending on $\lambda$ that
are distinct throughout a domain in the complex $\lambda$-plane,
relation (\ref{lambda-relation}) itself determines a curve in the complex
$\lambda$-plane. If $\lambda$ is on one of these curves, the standard
WKB procedure can be expected to apply to determine whether $\lambda$ is
in fact an $o(1)$ distance away from an eigenvalue (of course, this is 
subject to the curve avoiding any singularities 
and the existence of appropriate progressive paths $C^-$, $C^+$).

Therefore, as $\epsilon\downarrow0$, we expect that the discrete eigenvalues 
of (\ref{dirac}) will accumulate on the union of curves in the complex 
$\lambda$-plane described by formula (\ref{lambda-relation}),
with the union being taken over pairs of complex turning points.
Hence we are led to the following.
\begin{definition}
\label{asympt-spec-arc}
The curves in the complex $\lambda$-plane consisting of $\lambda$-points 
that give rise to admissible contours $C$ (on the $x$-plane) will be called 
\textbf{asymptotic spectral arcs}.
\end{definition}

\begin{definition}
\label{asymptotic-spectrum}
The set that is defined as the union of all the asymptotic spectral arcs in the complex 
$\lambda$-plane will be called the \textbf{asymptotic spectrum} of our Dirac operator. 	
\end{definition}

\begin{remark}
It follows from the arguments above that this asymptotic spectrum is a union  
of analytic arcs; if for each $\lambda$ there is only a finite number of turning points, 
then the asymptotic spectrum is a finite union of asymptotic spectral arcs. 
\end{remark}

Applying the exact WKB theory stated below in a neighborhood of the contour $C$ 
(admissible contour), we shall rigorously obtain the eigenvalue approximation by 
$\lambda$'s satisfying
\begin{equation}
\label{bsqc}
\frac{1}{\pi\epsilon}
\Im\bigg\{i
\int_{C^0(x_-(\lambda)\rightsquigarrow x_+(\lambda))}
\sqrt{-V_0(t,\lambda)}dt
\bigg\}
-\frac{1}{2}
\in\mathbb{Z}
\end{equation}
which can be interpreted as a \textit{Bohr-Sommerfeld quantization rule}.
Such $\lambda$'s are on the asymptotic spectral arc corresponding to the
condition (\ref{lambda-relation}) for the turning points $x_-(\lambda)$ and 
$x_+(\lambda)$.

We remark that, in general, there may well be eigenvalues which are not associated 
with such an asymptotic spectral arc as defined so far. This is illustrated by the 
example below.

\begin{example}
Consider  the simplest case where $S(x)$ is linear, i.e.
$S'(x)=c\in \R$ is constant.  In this case, the rectangular region ${\mathcal R}_0$ 
corresponds to the segment 
$${\mathcal R}_0=\Big\{-\frac{c}{2}\Big\}\times
\Big[-\sup_{x\in\R}|A|, \sup_{x\in\R}|A|\Big].$$

Additionally,  if $A(x)$ is ``bell-shaped", i.e.  if it decays at $\pm\infty$ and 
$xA'(x)<0$ in $\R^\pm$, then for any $\lambda$ lying in either of the 
segments $[-\frac c2,-\frac c2+i\sup A]$ or $[-\frac c2,-\frac c2-i\sup A]$, there 
exist exactly two real, simple turning points $x_-(\lambda)$,  $x_+(\lambda)$ 
connected by a Stokes line $C^0=[x_-(\lambda), x_+(\lambda)]\subset \R$, and 
$C_-=(-\infty, x_-(\lambda))$,  $C_+=(x_+(\lambda), +\infty)$ are progressive.  
Hence $C=\R$ is an admissible contour and the vertical segments 
$[-\frac c2,-\frac c2+i\sup A]$ and $[-\frac c2,-\frac c2-i\sup A]$ are asymptotic 
spectral arcs. 

Now, if we keep the linearity of $S$ but consider $A(x)$ to be double humped, 
say $A(x)=[(x^2-1)^2+1]^{-1}$, the complex intervals $[-\frac c2, \frac{-c+i}2]$,   
$[-\frac c2, \frac{-c-i}2]$ are asymptotic spectral arcs just as above. In particular, 
when $\lambda=\frac{-c\pm i}2$, there appears a double turning point at the origin 
$x=0$ on the Stokes line $C^0=[-\sqrt 2,\sqrt 2]$.  However, for 
$\lambda\neq\frac{-c\pm i}2$ in $[\frac{-c+i}2, -\frac c2+i]$ or $[\frac{-c-i}2, -\frac c2-i]$, 
$C^0$ splits into two Stokes lines, and such $\lambda$ does not belong to the 
asymptotic spectral arcs according to the above definition (although there do exist 
eigenvalues near such $\lambda$, see \cite{h+w}). 
\end{example}

But there is no reason why we should not allow two or more Stokes lines apearing in an admissible
path. We thus propose the following general definition. We note 
however that for the specific example we eventually focus on, $A(x)=S(x)=\sech(2x)$, 
there is at most one Stokes line appearing.

\begin{definition}
\label{general-admissible-curve}
Consider $N\in\N$. In general, we call a contour $C$ \textbf{admissible} if the following 
conditions hold.
\begin{itemize}
\item
$V_0$ is holomorphic in $x$, in the region of the
complex $x$-plane enclosed by $C$ and the real $x$-axis.
\item
$C$ coincides with $\R$ outside a compact set.
\item
$C$ contains  a finite number of pairs of  simple turning points 
$x_-^i(\lambda)$, $x_+^i(\lambda)$, where $i=1, ... , N$, with
$$
V_0(x^i_\pm(\lambda),\lambda)=0,
$$
\item
$C$ consists of a progressive curve $C^-$ connecting $- \infty$ and $x_-^1(\lambda)$,
a progressive curve $C^+$ connecting  $x_+^N(\lambda)$ and $+ \infty$,  progressive 
curves connecting $x^i_+(\lambda)$ and $x^{i+1}_-(\lambda)$  ($i=1,..., N-1$) and 
Stokes lines connecting $x^i_-(\lambda)$ and $x^i_+(\lambda)$, $i=1, ... , N$.
\end{itemize}
\end{definition}

Even this  more general definition cannot be completely comprehensive as a 
characterization of eigenvalues, since it does not take into account double 
truning points, which can exist, at least non-generically, as we will eventually see.

\subsection{The exact WKB method}
\label{exact-wkb}
In this and next subsection \ref{exact-wkb-connection-1-tp}, we review the local 
theory of the exact WKB method applied to our Zakharov-Shabat operator in a 
fixed open set in $\C$.

Let us fix $\lambda$, let $\Omega$ be a simply connected subdomain of the 
$x$-complex plane free from turning points, and take a fixed point 
$\alpha\in\Omega$. We define a \textit{phase map}
\begin{align}
z(\cdot,\lambda,\alpha):\Omega
&\rightarrow
\mathbb{C}\quad\text{with}\nn\\
\label{phase}
z(x,\lambda,\alpha)=i\int_{\gamma(\alpha\rightsquigarrow x)}
&
\sqrt{-V_0(t,\lambda)}dt
\end{align}
where of course the path of integration is irrelevant and we choose the 
branch of $\sqrt{-V_0(x,\lambda)}$ which makes it positive for large $x>0$.  
Observe that $z$ is uniquely defined in $\Omega$ and it maps 
$\Omega$ conformally to $z(\Omega,\lambda,\alpha)$.

We look for solutions of (\ref{miller-final}) having the form
\begin{equation}
\textbf{v}(x,\lambda,\epsilon,\alpha)=
\exp\bigg\{\pm\frac{z(x,\lambda,\alpha)}{\epsilon}\bigg\}
\tilde{\textbf{w}}^\pm[z(x,\lambda,\alpha)]
\end{equation}
for
$\tilde{\textbf{w}}^\pm=[\tilde{w}^{\pm}_1\hspace{3pt}\tilde{w}^{\pm}_2]^T$.
Substituting this in (\ref{miller-final}), we get the system
\begin{equation}\label{ansatz-1-sys}
\frac{\epsilon}{i}\frac{d}{dz}\tilde{\textbf{w}}^{\pm}(z)
=
\begin{bmatrix}
\pm i & H(z)^{-2}\\
-H(z)^{2} & \pm i
\end{bmatrix}
\tilde{\textbf{w}}^{\pm}(z)
\end{equation}
where [upon recalling (\ref{gigi})] the function $H$ is given by
\begin{equation*}
H(z)=\bigg[\frac{g_-(x,\lambda)}{g_+(x,\lambda)}\bigg]^{1/4}.
\end{equation*}
Finally,  we set
\begin{equation*}
P_{\pm}(z)=
\begin{bmatrix}
H(z)^{-1} & H(z)^{-1} \\
iH(z) & -iH(z)
\end{bmatrix}
\begin{bmatrix}
0 & 1\\
1 & 0
\end{bmatrix}^{\frac{1\pm1}{2}}
\end{equation*}
and apply the transformation
\be\nn
\tilde{\textbf{w}}^{\pm}(z)=P_{\pm}(z)\textbf{w}^{\pm}(z)
\ee
for $\textbf{w}^\pm=[w^{\pm}_{\rm even}\hspace{3pt}w^{\pm}_{\rm odd}]^T$.  
Then the system in (\ref{ansatz-1-sys}) is written as
\begin{equation}\label{ansatz-final-sys}
\frac{\epsilon}{i}\frac{d}{dz}\textbf{w}^{\pm}(z)
=
\begin{bmatrix}
0 & \mathcal{H}(z)\\
\mathcal{H}(z) & \mp\frac{2}{\epsilon}
\end{bmatrix}
\textbf{w}^{\pm}(z)
\end{equation}
in which we consider $\mathcal{H}$ to be defined by
\begin{equation}
\mathcal{H}(z)=\frac{\dot{H}(z)}{H(z)}
\end{equation}
and where the dot denotes differentiation with respect to $z$. Easily, one can
arrive at
\begin{equation*}
\mathcal{H}(z)=\frac{d}{dz}\log[H(z)]=
\frac{1}{4}
\frac{A(x)S''(x)-2\lambda A'(x)-A'(x)S'(x)}
{\big\{[\lambda+\tfrac{1}{2}S'(x)]^2+A^{2}(x)\big\}^{3/2}}.
\end{equation*}

In the usual WKB theory,  the vector valued \textit{symbols} $\textbf{w}^{\pm}$ are 
constructed as a power series in the parameter $\epsilon$,  which is in general divergent.
Here we would like to introduce the reader to the so-called {\it exact} WKB method (along 
the lines of G\'erard-Grigis \cite{gg} and Fujii\'e-Lasser-N\'ed\'elec \cite{fln}).  This method 
consists in the resummation of this divergent series in the following way.  By taking a point 
$x_0\in\Omega$,  setting 
\be\nn
z_0\equiv z_0(\lambda)=z(x_0,\lambda,\alpha)
\ee  
and postulating the \textit{series ansatz}
\begin{equation}
\label{formal-series}
\textbf{w}^{\pm}=
\sum_{n=0}^{\infty}\textbf{w}^{\pm}_{n}=
\sum_{n=0}^{\infty}
\begin{bmatrix}
w^{\pm}_{2n}\\
w^{\pm}_{2n-1}
\end{bmatrix}\equiv
\begin{bmatrix}
w^{\pm}_{\rm even}\\
w^{\pm}_{\rm odd}
\end{bmatrix}
\end{equation}
we see from (\ref{ansatz-final-sys}) that the functions $w^{\pm}_{n}, ~-1\leq n\in\mathbb{Z},$ 
can  be defined inductively by
\begin{equation*}
w_{-1}^\pm\equiv 0,\quad w_0^\pm\equiv 1
\end{equation*}
and for $1\leq n\in\mathbb{Z}$
\begin{equation}\label{recurrence}
\left\{
\begin{array}{rl}
\frac{d}{dz}w_{2n}^\pm(z) &= \mathcal{H}(z)w_{2n-1}^\pm(z) \\
(\frac{d}{dz}\pm\frac{2}{\epsilon})w_{2n-1}^\pm(z) &= \mathcal{H}(z)w_{2n-2}^\pm(z) \\
w_n^\pm(z_0) & =0.
\end{array}
\right.
\end{equation}

These recurrence equations uniquely determine (at least in a neighborhood
of $x_0$) the sequence of scalar functions $\{w_n^\pm(z,\epsilon,z_0)\}_{n=-1}^\infty$ 
and hence the sequence of vector-valued functions 
$\{\textbf{w}_n^\pm(z,\epsilon,z_0)\}_{n=0}^\infty$.
We can write the relations (\ref{recurrence}) in an integral form if we introduce
the following integral operators $\mathcal{I}_\pm$ and $\mathcal{J}$ taking
functions from $\mathscr{H}(z(\Omega,\lambda,\alpha))$ to functions in
$\mathscr{H}(z(\Omega,\lambda,\alpha))$ by
\begin{equation}
\label{I}
\mathcal{I}_\pm[f](z)=
\int_{\Gamma} e^{\pm2(\zeta-z)/\epsilon}\mathcal{H}(\zeta)f(\zeta)d\zeta
\end{equation}
\begin{equation}
\label{J}
\mathcal{J}[f](z)=
\int_{\Gamma}\mathcal{H}(\zeta)f(\zeta)d\zeta
\end{equation}
where $\Gamma\equiv\Gamma(z_0\rightsquigarrow z)$ is the image by $z$
of a path $\gamma\equiv\gamma(x_0\rightsquigarrow x)$ in $\Omega$
connecting $x_0$, $x$; in other words $\Gamma=z(\gamma)$ (for notation 
cf.  \S \ref{notation}). Then (\ref{recurrence}) for $n\geq1$ becomes
\begin{equation}
\label{integral-recurrence}
\left\{
\begin{array}{rl}
w_{2n}^\pm &= \mathcal{J}[w_{2n-1}^\pm] \\
w_{2n-1}^\pm &= \mathcal{I}_{\pm}[w_{2n-2}^\pm].
\end{array}
\right.
\end{equation}

Going back to the initial system (\ref{miller-ev}),
we have constructed (formal) WKB solutions of the form
\begin{multline}
\label{wkb-sols}
\textbf{u}^{\pm}(x,\lambda,\epsilon,\alpha,x_0)=\\
\exp\Big\{\pm\frac{z(x,\lambda,\alpha)}{\epsilon}\Big\}
Q(x,\lambda,\epsilon)
\begin{bmatrix}
0 & 1\\
1 & 0
\end{bmatrix}^{\frac{1\pm1}{2}}
\textbf{w}^{\pm}[z(x,\lambda,\alpha),\epsilon,z_0]
\end{multline}
where $Q(\cdot,\lambda,\epsilon)$ is a matrix-valued function with value
\be\nn
Q(x,\lambda,\epsilon)=
\begin{bmatrix}
\exp\{-iS(x)/(2\epsilon)\} & 0\\
0 & \exp\big\{iS(x)/(2\epsilon)\}
\end{bmatrix}
\begin{bmatrix}
1 & 1\\
-1 & 1
\end{bmatrix}
\begin{bmatrix}
H(z)^{-1} & H(z)^{-1} \\
iH(z) & -iH(z)
\end{bmatrix}.
\ee

Here,  we added the superscript $\pm$ to the solution
$\textbf{u}=[u_1\hspace{3pt}u_2]^T$ -consequently arriving at the
notation $\textbf{u}^{\pm}=[u^{\pm}_1\hspace{3pt}u^{\pm}_2]$- 
to distinguish between the two different cases  in the previous
discussion; also we  keep in mind that the WKB solutions
in (\ref{wkb-sols}) depend on a base point $\alpha$ for the phase
map (\ref{phase}) and a base point $x_0$ for the symbols
$\textbf{w}^{\pm}$. 

Recall Definition \ref{progress-path} and write
\be\nn
\Omega_\pm=
\{\, x\in\Omega\mid
\exists
\hspace{3pt}(\pm)\text{-progressive path}\hspace{3pt}
\gamma(x_0\rightsquigarrow x)\subset\Omega\,\}.
\ee
We have  the following basic theorem; for the proof,  we refer to \cite{fln}.
\begin{theorem}
\label{formal-to-rigorous}
The exact WKB solutions in (\ref{wkb-sols}) satisfy the following properties.
\begin{enumerate}
\item
The formal series in (\ref{formal-series})  is absolutely convergent in a
neighborhood of $x_0$.\\
\item
For each $N\in\mathbb{N}$, we have
\begin{align*}
\mathbf{w}^{\pm}-\sum_{n=0}^{N-1}\mathbf{w}^{\pm}_{n} 
& =
\mathcal{O}(\epsilon^N)
\quad\text{as}\quad\epsilon\downarrow0\\
w^\pm_{\rm even}-\sum_{n=0}^{N-1}w_{2n}^\pm 
& =
\mathcal{O}(\epsilon^N)
\quad\text{as}\quad\epsilon\downarrow0\\
w^\pm_{\rm odd}-\sum_{n=0}^{N-1}w_{2n-1}^\pm 
& =
\mathcal{O}(\epsilon^N)
\quad\text{as}\quad\epsilon\downarrow0
\end{align*}
uniformly in any compact subset of $\Omega_\pm$.  In particular,
there we have
\begin{align*}
w^\pm_{\rm even} 
& =
1+\mathcal{O}(\epsilon)
\quad\text{as}\quad\epsilon\downarrow0\\
\quad w^\pm_{\rm odd} 
& =
\mathcal{O}(\epsilon)
\quad\text{as}\quad\epsilon\downarrow0.
\end{align*}
\item
Any two exact WKB solutions with different base points
for the symbol satisfy
\begin{align}
\label{wronsky+-}
{\mathcal W}[{\bf u}^+(x,\lambda,\epsilon,\alpha,x_0),
{\bf u}^-(x,\lambda,\epsilon,\alpha,x_1)]
&=
4i w_{\rm even}^+(z_1,\epsilon,z_0)\\
\label{wronsky++}
{\mathcal W}[{\bf u}^+(x,\lambda,\epsilon,\alpha,x_0),
{\bf u}^+(x,\lambda,\epsilon,\alpha,x_1)]
&=
-4i e^{2z_1/\epsilon}w_{\rm odd}^+(z_1,\epsilon,z_0)
\end{align}
where $x_0$,  $x_1$ are two different points in $\Omega$ and
$z_j=z(x_j,\lambda,\alpha)$ for $j=0,1$.
\end{enumerate}
\end{theorem}

\subsection{Connection around a simple turning point}
\label{exact-wkb-connection-1-tp}

We continue in this section by presenting a local connection formula for exact WKB 
solutions near a simple turning point. Let $\alpha$ be a simple turning point, i.e.
\begin{align*}
\det M(\alpha)&=
g_+(\alpha)g_-(\alpha)=0\\
(\det M)'(\alpha)&=
g_+(\alpha)g'_-(\alpha)+g'_+(\alpha)g_-(\alpha)\neq 0
\end{align*}
which means that $\alpha$ is a simple zero of one and only one of either $g_+(x)$ or 
$g_-(x)$. Three Stokes lines $l_0$, $l_1$ and $l_2$, numbered in the anti-clockwise 
sense, emanate from $\alpha$ and  devide a neighborhood $\omega$ of $\alpha$ into 
three regions $\omega_0$ bounded by $l_1$ and $l_2$,  $\omega_1$ bounded by $l_2$ 
and $l_0$ and $\omega_2$ bounded by $l_0$ and $l_1$. In each of these regions, 
we are going to define exact WKB solutions ${\bf u}_0$, ${\bf u}_1$ and ${\bf u}_2$ 
respectively in such a way that its asymptotic behavior is known near $\alpha$ by 
Theorem \ref{formal-to-rigorous}. 

We have first to specify the branch of the multi-valued function
$z(x,\lambda,\alpha)$ defined by \eqref{phase}.
Let us take a base point $x_j$  in $\omega_j$ $(j=0,1,2)$.
We put a branch cut on the Stokes line $l_1$, and take the branch of the function 
$\sqrt{\det M}$ so that
$\re \,z(x,\lambda,\alpha)$ increases from $x_0$ towards the turning point $\alpha$. 
Then consequently $\re \,z(x,\lambda,\alpha)$ decreases from $x_1$ towards  $\alpha$ 
and increases from $x_2$ towards  $\alpha$.

We define three exact WKB solutions:
\begin{align}
{\bf u}_0(x,\epsilon)={\bf u}^+(x,\lambda,\epsilon,\alpha,x_0), \\
{\bf u}_1(x,\epsilon)={\bf u}^-(x,\lambda,\epsilon,\alpha,x_1), \\
{\bf u}_2(x,\epsilon)={\bf u}^+(x,\lambda,\epsilon,\alpha,x_2).
\end{align}
These solutions are constructed near $x_0$, $x_1$, $x_2$ respectively but can be 
extended analytically to $\omega$. They cannot be linearly independent since 
the vector space of the solutions is of dimension two. More precisely, we have the 
following linear relation.
\begin{proposition}
\label{dep}
Let ${\bf u}_0$, ${\bf u}_1$, ${\bf u}_2$ be three solutions to the differential equation 
\eqref{ev-problem} in a region $\omega$. Then the following identity holds:
\be
\label{Widentity}
{\mathcal W}[{\bf u}_1,{\bf u}_2]{\bf u}_0+
{\mathcal W}[{\bf u}_2,{\bf u}_0]{\bf u}_1+
{\mathcal W}[{\bf u}_0,{\bf u}_1]{\bf u}_2=0.
\ee
\end{proposition}
\begin{proof}
If the three determinants ${\mathcal W}[{\bf u}_1,{\bf u}_2]$, ${\mathcal W}[{\bf u}_2,{\bf u}_0]$ 
and ${\mathcal W}[{\bf u}_0,{\bf u}_1]$ are all 0, the identity holds obviously.
Suppose that at least one of those does not vanish, say ${\mathcal W}[{\bf u}_1,{\bf u}_2]\ne 0$. 
Then by taking the determinant of the LHS of \eqref{Widentity} with ${\bf u}_1$ we find 
$$
{\mathcal W}[{\bf u}_1,{\bf u}_2]{\mathcal W}[{\bf u}_0,{\bf u}_1]+
{\mathcal W}[{\bf u}_0,{\bf u}_1]{\mathcal W}[{\bf u}_2,{\bf u}_1]
$$
which is obviously zero. In the same way, we obtain that the determinant of the LHS of 
\eqref{Widentity} with ${\bf u}_2$ is zero. Since the pair of solutions $\{{\bf u}_1, {\bf u}_2\}$ 
makes a basis of the solution space, this means that the LHS of \eqref{Widentity} should 
be zero. 
\end{proof}

Let us come back to our exact WKB solutions ${\bf u}_0$, ${\bf u}_1$ and ${\bf u}_2$ 
defined above and compute the asymptotic behavior of the determinants appearing in 
this proposition.

\begin{proposition}
\label{3.567}
As $\epsilon\downarrow0$, the following asymptotics hold true
\begin{align}
\label{01}
&{\mathcal W}[{\bf u}_0,{\bf u}_1]=2i(1+\CO(\epsilon)),
\\
\label{12}
&{\mathcal W}[{\bf u}_1,{\bf u}_2]=-2i(1+\CO (\epsilon)),
\\
\label{20}
&{\mathcal W}[{\bf u}_2,{\bf u}_0]=\mp 2(1+\CO(\epsilon)),
\end{align}
where in (\ref{20}) the upper sign (minus) is to be chosen in the case of $\alpha$ being 
a zero of  $g_+$ and the lower sign (plus) in the case that $\alpha$ is a zero of $g_-$. In 
particular, any two of the solutions ${\bf u}_0$, ${\bf u}_1$ and ${\bf u}_2$ are linearly 
independent for sufficiently small $\epsilon$.
\end{proposition}
\begin{proof}
Using Theorem \ref{formal-to-rigorous} (iii), we find
$$
{\mathcal W}[{\bf u}_1,{\bf u}_2]=-2iw_{\rm even}^+(z_1,\epsilon,z_2)
$$
where $z_j=z(x_j,\lambda,\alpha)$.
Since there exists a $(+)$-progressive path from $x_2$ to $x_1$, one has 
$w_{\rm even}^+(z_1,\epsilon,z_2)=1+\CO (\epsilon)$ by Theorem \ref{formal-to-rigorous} (ii), 
and hence we obtain \eqref{12}. Similarly we have
$$
{\mathcal W}[{\bf u}_0,{\bf u}_1]=2iw_{\rm even}^+(z_1,\epsilon,z_0)=2i(1+\CO(\epsilon)).
$$

For the computation of the determinant  ${\mathcal W}[{\bf u}_2,{\bf u}_0]$, we have to be 
careful since there is a branch cut between $x_2$ and $x_0$. Before applying the 
\say{wronskian formula}, we rewrite ${\bf u}_2$ on the Riemann sheet continued from $x_0$ 
passing across the branch. Let $x$ be a point near $x_2$ and $\hat x$ the same point as $x$ 
but on that Riemann sheet. Then we have
\begin{align*}
&z(x,\lambda,\alpha)=-z(\hat x,\lambda,\alpha),\\
&H(x)=\mp iH(\hat x),\,\text{ if }\alpha \text{ is a zero of }g_\pm\quad\text{and}\\
&w^\pm (z,\epsilon,z_2)=w^\mp (z,\epsilon,z_2).
\end{align*}
Observe that $\arg (x-\alpha)=\arg(\hat x-\alpha)+2\pi$ and hence the first identity is obvious 
since the integrand of $z$ contains a multi-valued function of type $\sqrt{x-\alpha}$. 
The second assertion holds due to the fact that $H$ is the product of $g_+^{-1/4}$ and 
$g_-^{1/4}$. To verify the third identity, let $w_n^\pm(z)$, $n\ge 1$, be a family of solutions 
to the transport equations \eqref{recurrence} satisfying the initial condition $w_n^\pm (z_2)=0$ 
at $z_2=z(x_2)$, and set $f_n^\pm(\hat z)=w_n^\pm(z)$, $n\ge 1$. 
Inspecting \eqref{recurrence} and using $d/dz=-d/d\hat z$, it is easy to check that
\begin{align*}
\left (\frac d{d\hat z}\mp\frac 2\epsilon\right )f_{2n+1}^\pm(\hat z)=
\frac{dH(\hat z)/d\hat z}{H(\hat z)}f_{2n}^\pm (\hat z), \\
\frac d{d\hat z}f_{2n+2}^\pm(\hat z)=
\frac{dH(\hat z)/d\hat z}{H(\hat z)}f_{2n+1}^\pm (\hat z)
\end{align*}
with $f_n^\pm (\hat z_2)=0$ at $\hat z_2=z(\hat x_2)$. By uniqueness of solutions, it follows 
that $w_n^\pm(z)=w_n^\mp(\hat z)$ with $w_n^\mp(\hat z_2)=0$, which, in view of 
\eqref{formal-series}, implies that 
$$
{\bf w}^\pm (z,\epsilon,z_2)=
\sum_{n=0}^\infty\left (
\begin{array}{c}
w_{2n}^\pm(z,z_2) \\
w_{2n+1}^\pm(z,z_2)
\end{array}
\right )
=\sum_{n=0}^\infty\left (
\begin{array}{c}
w_{2n}^\mp(\hat z,\hat z_2) \\
w_{2n+1}^\mp(\hat z,\hat z_2)
\end{array}
\right )
={\bf w}^\mp(\hat z,\epsilon,\hat z_2).
$$
Thus, with this point $\hat x_2$ as base point of the symbol, the solution ${\bf u}_2(x,\epsilon)$ 
is written as
$$
{\bf u}_2(x,\epsilon)=\mp i {\bf u}^-(x,\lambda,\epsilon,\alpha, \hat x_2),\,
\text{ if }\alpha \text{ is a zero of }g_\pm.
$$

We now apply Theorem \ref{formal-to-rigorous} (iii) to compute the desired determinant, 
using the above expression. Since there exists a progressive curve from $x_0$ to $\hat x_2$, 
we have
$$
{\mathcal W}[{\bf u}_0,{\bf u}_2]=\pm i(\det Q)w_{\rm even}^+(\hat z_2,\epsilon,z_0)=
\pm 2(1+\CO(\epsilon))
$$
when $\alpha$ is a zero of $g_\pm$.
\end{proof}

\subsection{Eigenvalues near a generic point on the asymptotic spectral arc}
\label{bs-at-two-turn-pts}

In this section, we return to the eigenvalue problem \eqref{ev-problem} for the 
Zakharov-Shabat operator $\mathfrak{D}_\epsilon$. 
We derive the quantization condition for the eigenvalues in a small 
neighborhood of a fixed $\lambda_0$, independent of $\epsilon$, 
belonging to an asymptotic spectral arc under the following generic condition.

\vspace{0.5cm}

\begin{minipage}{0.1\textwidth}
{\bf (H1):}
\end{minipage}
\hspace{0.01\textwidth}
\begin{minipage}{0.8\textwidth}
$C^0$ contains no other turning point than $x_-(\lambda)$ and $x_+(\lambda)$.
\end{minipage}

\vspace{0.5cm}
 
Let $C(\lambda_0)$ be the admissible contour 
$C(\lambda_0)=C^-(\lambda_0)\cup C^0(\lambda_0)\cup C^+(\lambda_0)$ 
corresponding to $\lambda_0$ and denote by 
$\alpha_0=\alpha(\lambda_0)=x_-(\lambda_0)$,  
$\beta_0=\beta(\lambda_0)=x_+(\lambda_0)$ 
the simple turning points at the extremities of $C^0(\lambda_0)$.  For $\lambda$ in 
a small enough neighborhood of $\lambda_0$, the simple turning points 
$\alpha=\alpha(\lambda)=x_-(\lambda)$ and $\beta=\beta(\lambda)=x_+(\lambda)$ 
are still defined and analytic.

From point $\alpha$ there are  three Stokes lines emanating,  namely 
$l_0^{\alpha}$, $l_1^{\alpha}$, $l_2^{\alpha}$ (in anti-clockwise order); similarly from 
$\beta$ there emerge three Stokes lines which we denote by 
$l_0^\beta$, $l_1^\beta$, $l_2^\beta$ (in clockwise order).  For the configuration we refer 
to Figure \ref{stokes1}.  Suppose $\lambda=\lambda_0$ and that $l_0^{\alpha}$ and 
$l_0^{\beta}$ coincide with $C^0(\lambda)$. The five Stokes lines 
$C^0(\lambda)$,  $l_1^{\alpha}$, $l_2^{\alpha}$,  $l_1^{\beta}$ and 
$l_2^{\beta}$ divide a neighborhood $\omega$ of the admissible contour into 
four regions which we denote as follows
\begin{itemize}
\item
$\omega_{\rm left}^{\alpha}$ is bounded by $l_1^{\alpha}$,  $l_2^{\alpha}$
\item
$\omega_{\rm right}^{\beta}$ is bounded by $l_1^{\beta}$,  $l_2^{\beta}$
\item
$\omega_{\rm up}^{\alpha,\beta}$ is bounded by $C^0(\lambda)$, 
$l_1^{\alpha}$,  $l_1^{\beta}$ and finally
\item
$\omega_{\rm down}^{\alpha,\beta}$ is bounded by $C^0(\lambda)$, 
$l_2^{\alpha}$,  $l_2^{\beta}$. 
\end{itemize}
Let us take a point in each of these regions: $x_0^{\alpha}$ in the region 
$\omega_{\rm left}^{\alpha}$,  $x_0^{\beta}$ in the region $\omega_{\rm right}^{\beta}$,  
$x_1$ in the region $\omega_{\rm down}^{\alpha,\beta}$ and $x_2$ in the 
region $\omega_{\rm up}^{\alpha,\beta}$.

\begin{figure}[htbp]
\centering
\includegraphics[bb=-100 0 710 240, width=13cm]{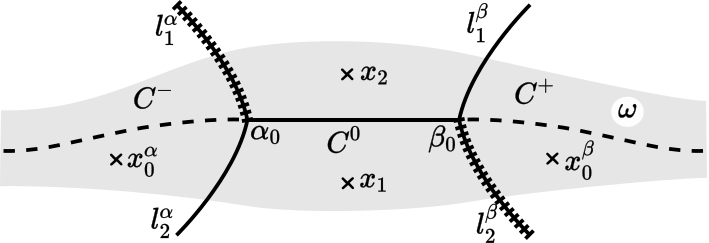}
\caption{Stokes lines near a simple turning point.}
\label{stokes1}
\end{figure}

As in \S \ref{exact-wkb-connection-1-tp}, we have to consider branch cuts; we choose 
one on the Stokes line $l_1^{\alpha}$ and one on $l_2^{\beta}$. Assuming that $\re z(x)$ 
increases towards $\alpha$ in $\omega_{\rm left}^{\alpha}$,  we define exact WKB 
solutions as follows.
\begin{align*}
{\bf u}_0^\alpha(x,\epsilon)&={\bf u}^+(x,\lambda,\epsilon,\alpha,x_0^\alpha), \quad
{\bf u}_0^\beta(x,\epsilon)={\bf u}^-(x,\lambda,\epsilon,\beta,x_0^\beta)\\
{\bf u}_1^\alpha(x,\epsilon)&={\bf u}^-(x,\lambda,\epsilon,\alpha,x_1), \quad 
{\bf u}_1^\beta(x,\epsilon)={\bf u}^-(x,\lambda,\epsilon,\beta,x_1)\\
{\bf u}_2^\alpha(x,\epsilon)&={\bf u}^+(x,\lambda,\epsilon,\alpha,x_2), \quad
{\bf u}_2^\beta(x,\epsilon)={\bf u}^+(x,\lambda,\epsilon,\beta,x_2)
\end{align*}
Remark that all these exact WKB solutions are chosen so that their asymptotic 
expansions in Theorem \ref{formal-to-rigorous} (ii) are valid near the designated 
turning points. In particular,  ${\bf u}_0^\alpha(x,\epsilon)$ is a decaying solution 
along $C^-(\lambda_0)$ and ${\bf u}_0^\beta(x,\epsilon)$ is a decaying solution 
along $C^+(\lambda_0)$. This implies the following lemma.
\begin{lemma}
\label{linear-depend-lemma-of-u}
For any $\lambda_0$ on the asymptotic spectral arcs, there exists a complex 
neighbourhood $U=U(\lambda_0)$ of $\lambda_0$ such that a $\lambda\in U$ 
is an eigenvalue of $\mathfrak{D}_\epsilon$ if and only if 
${\bf u}_0^\alpha(x,\epsilon)$ and ${\bf u}_0^\beta(x,\epsilon)$ are linearly 
dependent.
\end{lemma}

Applying Proposition \ref{dep} near $\alpha$ and $\beta$, we have
\begin{align}
\label{wronsk1}
{\mathcal W}[{\bf u}_1^\alpha,{\bf u}_2^\alpha]{\bf u}_0^\alpha+
{\mathcal W}[{\bf u}_2^\alpha,{\bf u}_0^\alpha]{\bf u}_1^\alpha+
{\mathcal W}[{\bf u}_0^\alpha,{\bf u}_1^\alpha]{\bf u}_2^\alpha=0, \\
\label{wronsk2}
{\mathcal W}[{\bf u}_1^\beta,{\bf u}_2^\beta]{\bf u}_0^\beta+
{\mathcal W}[{\bf u}_2^\beta,{\bf u}_0^\beta]{\bf u}_1^\beta+
{\mathcal W}[{\bf u}_0^\beta,{\bf u}_1^\beta]{\bf u}_2^\beta=0.
\end{align}
On the other hand, there is an obvious relation between ${\bf u}_j^\alpha$ and 
${\bf u}_j^\beta$ for $j=1,2$:
\be
\label{diagrel}
{\bf u}_1^\beta=
e^{z(\beta,\lambda,\alpha)/\epsilon}{\bf u}_1^\alpha,\quad 
{\bf u}_2^\beta=
e^{-z(\beta,\lambda,\alpha)/\epsilon}{\bf u}_2^\alpha
\ee
where $z(\beta,\lambda,\alpha)$ is the action integral between $\alpha$ and $\beta$
\be
\label{actionintegral}
z(\beta,\lambda,\alpha)=
i\int_{\gamma(\alpha\rightsquigarrow \beta)}
\sqrt{-V_0(t,\lambda)}dt.
\ee
It follows that ${\bf u}_0^\alpha(x,\epsilon)$ and ${\bf u}_0^\beta(x,\epsilon)$ 
are linearly dependent if and only if
$$
{\mathcal W}[{\bf u}_2^\alpha,u_0^\alpha]
{\mathcal W}[{\bf u}_0^\beta,{\bf u}_1^\beta]=
{\mathcal W}[{\bf u}_0^\alpha,{\bf u}_1^\alpha]
{\mathcal W}[{\bf u}_2^\beta,{\bf u}_0^\beta]
e^{2z(\beta,\lambda,\alpha)/\epsilon}.
$$
Recall the functions $g_{\pm}(x,\lambda)=\mp[\lambda+\frac{1}{2}S'(x)\pm iA(x)]$.  
For a pair of simple turning points $\alpha$, $\beta$, we define an index 
$\delta (\alpha,\beta)$ by
\be
\label{dab}
\delta(\alpha,\beta)=
\left\{
\begin{array}{cll}
-1 &\text{if}\quad 
g_+(\alpha,\lambda)=g_+(\beta,\lambda)=0 &\text {or}\quad 
g_-(\alpha,\lambda)=g_-(\beta,\lambda)=0,\\
1 &\text{if}\quad 
g_+(\alpha,\lambda)=g_-(\beta,\lambda)=0 &\text {or}\quad
g_-(\alpha,\lambda)=g_+(\beta,\lambda)=0.
\end{array}
\right.
\ee
Hence we arrive at the following result. 
\begin{theorem}
\label{BS}
For any $\lambda_0$ on the asymptotic spectral arcs satisfying condition 
{\bf (H1)}, there exists a complex neighborhood $U$ of $\lambda_0$ such that 
$\lambda\in U$ is an eigenvalue of  $\mathfrak{D}_\epsilon$  if and only if
$$
m(\epsilon)e^{2z(\beta,\lambda,\alpha)/\epsilon}=1,
\quad\text{with }
m(\epsilon)=
\frac{
{\mathcal W}[{\bf u}_0^\alpha,{\bf u}_1^\alpha]{\mathcal W}[{\bf u}_2^\beta,{\bf u}_0^\beta]}
{{\mathcal W}[{\bf u}_2^\alpha,{\bf u}_0^\alpha]{\mathcal W}[{\bf u}_0^\beta,{\bf u}_1^\beta]
}.
$$
As $\epsilon\downarrow0$, $m(\epsilon)$ has the asymptotic behavior
$$
m(\epsilon)=\delta(\alpha_0,\beta_0)+\CO(\epsilon).
$$
\end{theorem}
\begin{proof}
It remains to check the  asymptotic part of the statement.  We already know from 
Proposition \ref{3.567} that as $\epsilon\downarrow0$
\be
\label{asympt-wronsk-1}
{\mathcal W}[{\bf u}_2^\alpha,{\bf u}_0^\alpha]=
\mp 2+\CO(\epsilon),\quad
{\mathcal W}[{\bf u}_0^\alpha,{\bf u}_1^\alpha]=
2i+\CO(\epsilon)
\ee
when $\alpha$ is a zero of $g_\pm$. Similarly, as $\epsilon\downarrow0$ we have
\be
\label{asympt-wronsk-2}
{\mathcal W}[{\bf u}_2^\beta,{\bf u}_0^\beta]=2i+\CO(\epsilon),\quad
{\mathcal W}[{\bf u}_0^\beta,{\bf u}_1^\beta]=\mp 2+\CO(\epsilon)
\ee
when $\beta$ is a zero of $g_\pm$.
These asymptotic formulas immediately give the assertion.
\end{proof}

\begin{remark}(The situation at the endpoints $\lambda_D^{(j)},j=1,2$ of the branches)
It is clear from the proof that the statement of Theorem \ref{BS} is still valid for 
$\lambda_0$ at a limit point of the asymptotic spectral arcs where $x_+(\lambda_0)$ 
and $x_-(\lambda_0)$ coalesce to a double turning point.
\end{remark}

\subsection{What happens near the bifurcation point}
\label{near-the-bf}

Now let us suppose that $\lambda_0$ lies on the asymptotic spectral arcs with 
admissible contour 
$C(\lambda_0)=C^-(\lambda_0)\cup C^0(\lambda_0)\cup C^+(\lambda_0)$ 
and that the Stokes line $C^0(\lambda_0)$ connecting the two simple turning points 
$\alpha_0=\alpha(\lambda_0)$ and $\beta_0=\beta(\lambda_0)$ now also contains 
a third simple turning point between of them,  namely $\gamma_0=\gamma(\lambda_0)$; 
of course,  $\alpha, \beta$ and $\gamma$ are analytic functions defined in a 
neighborhood of $\lambda_0$.

Hence, $C^0(\lambda)$ consists of a Stokes line $l_0^{\alpha}$ connecting 
$\alpha$ and $\gamma$ and a Stokes line $l_0^{\beta}$ connecting $\beta$ and 
$\gamma$: $C^0(\lambda)=l_0^{\alpha}\cup l_0^{\beta}$.  We suppose that
the third Stokes line from $\gamma$ divides the region 
$\omega_{\rm down}^{\alpha,\beta}$ (defined as in section \ref{bs-at-two-turn-pts}) 
into two subregions $\omega_{\rm down}^{\alpha,\gamma}$ and 
$\omega_{\rm down}^{\beta,\gamma}$,  and take a point 
$x_1^\alpha$ in $\omega_{\rm down}^{\alpha,\gamma}$ and $x_1^\beta$ in 
$\omega_{\rm down}^{\beta,\gamma}$ (see Figure \ref{stokes2}). 

\begin{figure}[htbp]
\centering
\includegraphics[bb=-80 0 387 240, width=10cm]{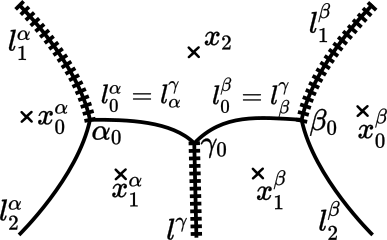}
\caption{Stokes lines near a bifurcation point.}
\label{stokes2}
\end{figure}

Now we let $\lambda$ vary in a small neighborhood of $\lambda_0$. Then the Stokes
lines in $\omega$ change the geometric configuration. We denote by 
$l_j^\alpha(\lambda)$, $l_j^\beta(\lambda)$  for $j=0,1,2$ and 
$l^\gamma(\lambda)$, $l_\alpha^\gamma(\lambda)$, $l_\beta^\gamma(\lambda)$ 
the three Stokes lines emanating from $\alpha(\lambda)$, $\beta(\lambda)$ and 
$\gamma(\lambda)$ respectively such that 
$l_0^{\alpha}(\lambda_0)=l_\alpha^\gamma(\lambda_0)$, 
$l_0^{\beta}(\lambda_0)=l_\beta^\gamma(\lambda_0)$. 
The variation of Stokes geometry as $\lambda$ turns around $\lambda_0$ is 
illustrated in Figure \ref{stokesvariation}. Remark that the Stokes lines stay close to 
those for $\lambda=\lambda_0$ by continuity.

\begin{figure}[htbp]
\centering
\includegraphics[bb=-200 0 1231 857, width=15cm]{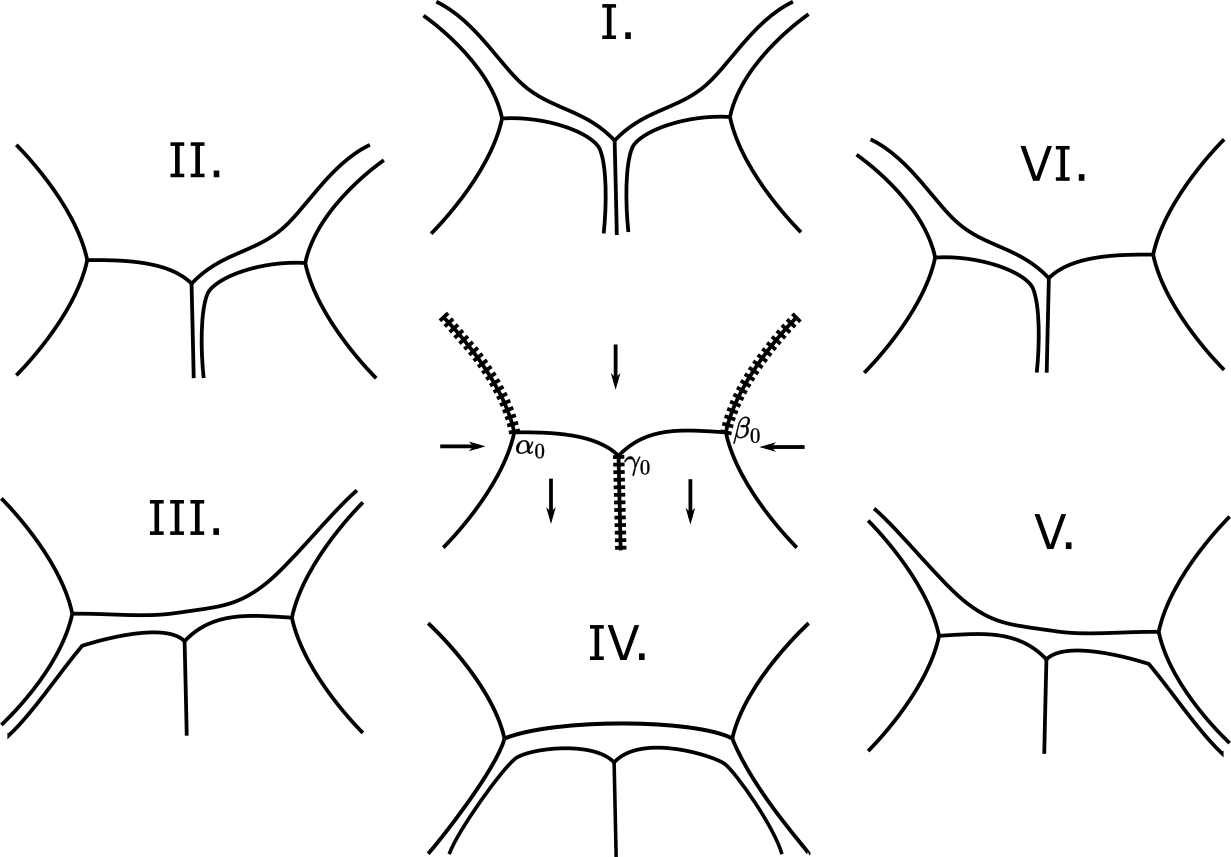}
\caption{Variation of the Stokes geometry around a bifurcation point.}
\label{stokesvariation}
\end{figure}

We put branch cuts along the Stokes lines $l_1^\alpha(\lambda)$, $l_1^\beta(\lambda)$ 
and $l^\gamma(\lambda)$ (we emphasize at this point that our selection for the 
branch cut placed on a Stokes line emerging from $\beta$ in this case is different 
from the previous case where we dealt with only two turning points), and suppose 
that $\re z(x)$ increases towards $\alpha$ in $\omega_{\rm left}^\alpha$.  
We define as before
\begin{align*}
{\bf u}_0^\alpha(x,\epsilon)&={\bf u}^+(x,\lambda,\epsilon,\alpha,x_0^\alpha), \quad
{\bf u}_0^\beta(x,\epsilon)={\bf u}^+(x,\lambda,\epsilon,\beta,x_0^\beta)\\
{\bf u}_1^\alpha(x,\epsilon)&={\bf u}^-(x,\lambda,\epsilon,\alpha,x_1^\alpha), \quad
{\bf u}_1^\beta(x,\epsilon)={\bf u}^-(x,\lambda,\epsilon,\beta,x_1^\beta)\\
{\bf u}_2^\alpha(x,\epsilon)&={\bf u}^+(x,\lambda,\epsilon,\alpha,x_2), \quad
{\bf u}_2^\beta(x,\epsilon)={\bf u}^+(x,\lambda,\epsilon,\beta,x_2)
\end{align*}
\begin{align*}
{\bf u}_1^{\gamma,\alpha}(x,\epsilon)&={\bf u}^-(x,\lambda,\epsilon,\gamma,x_1^\alpha)\\
{\bf u}_1^{\gamma,\beta}(x,\epsilon)&={\bf u}^-(x,\lambda,\epsilon,\gamma,x_1^\beta)\\
{\bf u}_2^\gamma(x,\epsilon)&={\bf u}^+(x,\lambda,\epsilon,\gamma,x_2)
\end{align*}

For each one of the triples $\{{\bf u}_0^\alpha, {\bf u}_1^\alpha, {\bf u}_2^\alpha\}$, 
$\{{\bf u}_2^\gamma, {\bf u}_1^{\gamma, \alpha}, {\bf u}_1^{\gamma,\beta}\}$ and 
$\{{\bf u}_0^\beta, {\bf u}_1^\beta, {\bf u}_2^\beta\}$, we have the following relations
\begin{align}
\label{120a}
{\mathcal W}[{\bf u}_1^\alpha,{\bf u}_2^\alpha]{\bf u}_0^\alpha+
{\mathcal W}[{\bf u}_2^\alpha,{\bf u}_0^\alpha]{\bf u}_1^\alpha+
{\mathcal W}[{\bf u}_0^\alpha,{\bf u}_1^\alpha]{\bf u}_2^\alpha=
0, \\
\label{120c}
{\mathcal W}[{\bf u}_1^{\gamma,\alpha},{\bf u}_1^{\gamma,\beta}]{\bf u}_2^\gamma+
{\mathcal W}[{\bf u}_1^{\gamma,\beta},{\bf u}_2^\gamma]{\bf u}_1^{\gamma,\alpha}+
{\mathcal W}[{\bf u}_2^\gamma,{\bf u}_1^{\gamma,\alpha}]{\bf u}_1^{\gamma,\beta}=
0, \\
\label{120b}
{\mathcal W}[{\bf u}_1^\beta,{\bf u}_2^\beta]{\bf u}_0^\beta+
{\mathcal W}[{\bf u}_2^\beta,{\bf u}_0^\beta]{\bf u}_1^\beta+
{\mathcal W}[{\bf u}_0^\beta,{\bf u}_1^\beta]{\bf u}_2^\beta=
0.
\end{align}
Using the obvious relations
$$
{\bf u}_1^\beta=
e^{z(\beta,\lambda,\gamma)/\epsilon}{\bf u}_1^{\gamma,\beta},\quad
{\bf u}_1^{\gamma,\alpha}=
e^{z(\gamma,\lambda,\alpha)/\epsilon}{\bf u}_1^\alpha
$$
$$
{\bf u}_2^\beta=e^{-z(\beta,\lambda,\gamma)/\epsilon}{\bf u}_2^\gamma,\quad
{\bf u}_2^\gamma=
e^{-z(\gamma,\lambda,\alpha)/\epsilon}{\bf u}_2^\alpha
$$
and 
$$z(\beta,\lambda,\gamma)+z(\gamma,\lambda,\alpha)=z(\beta,\lambda,\alpha)$$ 
we deduce that the sum 
\begin{equation*}
{\mathcal W}[{\bf u}_2^\beta,{\bf u}_0^\beta]{\bf u}_1^\beta+
{\mathcal W}[{\bf u}_0^\beta,{\bf u}_1^\beta]{\bf u}_2^\beta
\end{equation*}
of the second and the third terms of (\ref{120b}) is
\begin{equation*}
{\mathcal W}[{\bf u}_2^\beta,{\bf u}_0^\beta]
e^{z(\beta,\lambda,\gamma)/\epsilon}{\bf u}_1^{\gamma,\beta}+
{\mathcal W}[{\bf u}_0^\beta,{\bf u}_1^\beta]
e^{-z(\beta,\lambda,\gamma)/\epsilon}{\bf u}_2^{\gamma}
\end{equation*}
and this is equal to 
\begin{equation*}
{\mathcal W}[{\bf u}_0^\beta,{\bf u}_1^\beta]
e^{-z(\beta,\lambda,\alpha)/\epsilon}{\bf u}_2^\alpha-
\frac{{\mathcal W}[{\bf u}_2^\beta,{\bf u}_0^\beta]}{{\mathcal W}
[{\bf u}_2^\gamma,{\bf u}_1^{\gamma,\alpha}]}
e^{z(\beta,\lambda,\gamma)/\epsilon}\left 
({\mathcal W}[{\bf u}_1^{\gamma,\alpha},{\bf u}_1^{\gamma,\beta}]{\bf u}_2^\gamma+
{\mathcal W}[{\bf u}_1^{\gamma,\beta},{\bf u}_2^\gamma]{\bf u}_1^{\gamma,\alpha}\right )
\end{equation*}
or 
\begin{multline*}
{\mathcal W}[{\bf u}_0^\beta,{\bf u}_1^\beta]
e^{-z(\beta,\lambda,\alpha)/\epsilon}{\bf u}_2^\alpha\\
\qquad-
\frac{{\mathcal W}[{\bf u}_2^\beta,{\bf u}_0^\beta]}{{\mathcal W}
[{\bf u}_2^\gamma,{\bf u}_1^{\gamma,\alpha}]}e^{z(\beta,\lambda,\gamma)/\epsilon}\left 
({\mathcal W}[{\bf u}_1^{\gamma,\alpha},{\bf u}_1^{\gamma,\beta}]
e^{-z(\gamma,\lambda,\alpha)/\epsilon}{\bf u}_2^\alpha+
{\mathcal W}[{\bf u}_1^{\gamma,\beta},{\bf u}_2^\gamma]
e^{z(\gamma,\lambda,\alpha)/\epsilon}{\bf u}_1^\alpha\right ).
\end{multline*} 
After some calculations, the coefficient of ${\bf u}_2^\alpha$ in this last expression is
$$
{\mathcal W}[{\bf u}_0^\beta,{\bf u}_1^\beta]
e^{-z(\beta,\lambda,\alpha)/\epsilon}-
\frac{{\mathcal W}[{\bf u}_2^\beta,{\bf u}_0^\beta]{\mathcal W}
[{\bf u}_1^{\gamma,\alpha},{\bf u}_1^{\gamma,\beta}]}{{\mathcal W}
[{\bf u}_2^\gamma,{\bf u}_1^{\gamma,\alpha}]}
e^{(z(\beta,\lambda,\gamma)-z(\gamma,\lambda,\alpha))/\epsilon}
$$
while the coefficient of ${\bf u}_1^\alpha$ is
$$
-\frac{{\mathcal W}[{\bf u}_2^\beta,{\bf u}_0^\beta]
{\mathcal W}[{\bf u}_1^{\gamma,\beta},{\bf u}_2^{\gamma}]}{{\mathcal W}
[{\bf u}_2^\gamma,{\bf u}_1^{\gamma,\alpha}]}
e^{z(\beta,\lambda,\alpha)/\epsilon}.
$$
Comparing this with \eqref{120a}, we obtain the following theorem.
\begin{theorem}
Let $\lambda_0$ be a point on the asymptotic spectral arcs such that $C^0(\lambda_0)$  
consists of two Stokes lines; one from $\alpha_0$ to $\gamma_0$ and another one from 
$\beta_0$ to $\gamma_0$ (for a simple turning point $\gamma_0$; cf.  Figure 
\ref{stokes2}). There exists a neighborhood $U$ of $\lambda_0$ such that 
$\lambda\in U$ is an eigenvalue of  $\mathfrak{D}_\epsilon$  if and only if the 
following condition holds
\begin{equation}
\label{QC3}
m_\alpha(\epsilon)e^{2z(\gamma,\lambda,\alpha)/\epsilon}+
m_\beta(\epsilon)e^{2z(\gamma,\lambda,\beta)/\epsilon}=
1.
\end{equation}
Here the functions $m_\alpha(\epsilon)$ and $m_\beta(\epsilon)$ are given by
$$
m_\alpha(\epsilon)=
\frac{{\mathcal W}[{\bf u}_0^\alpha,{\bf u}_1^\alpha]
{\mathcal W}[{\bf u}_1^{\gamma,\beta},{\bf u}_2^{\gamma}]}{{\mathcal W}
[{\bf u}_2^\alpha,{\bf u}_0^\alpha]{\mathcal W}
[{\bf u}_1^{\gamma,\alpha},{\bf u}_1^{\gamma,\beta}]},\quad
m_\beta(\epsilon)=
\frac{{\mathcal W}[{\bf u}_0^\beta,{\bf u}_1^\beta]{\mathcal W}
[{\bf u}_2^\gamma,{\bf u}_1^{\gamma,\alpha}]}{{\mathcal W}
[{\bf u}_2^\beta,{\bf u}_0^\beta]{\mathcal W}
[{\bf u}_1^{\gamma,\alpha},{\bf u}_1^{\gamma,\beta}]}
$$
and as $\epsilon\downarrow0$, they behave like
$$
m_\alpha(\epsilon)=\mp 1+\CO(\epsilon),\quad m_\beta(\epsilon)=\mp 1+\CO(\epsilon)
$$
where the minus/plus sign for $m_\alpha(\epsilon)$ [resp. $m_\beta(\epsilon)$] 
corresponds to the case where $\alpha$ (resp. $\beta$) and $\gamma$ are 
zeros of the same/different $g_+$, $g_-$.
\end{theorem}
\begin{proof}
For the second part,  we have in fact modulo ${\mathcal O}(\epsilon)$,
$$
{\mathcal W}[{\bf u}_0^\alpha,{\bf u}_1^\alpha]=2i,\quad
{\mathcal W}[{\bf u}_0^\beta,{\bf u}_1^\beta]=2i,\quad
{\mathcal W}[{\bf u}_2^\gamma,{\bf u}_1^{\gamma,\alpha}]=2i,\quad
{\mathcal W}[{\bf u}_2^{\gamma},{\bf u}_1^{\gamma,\beta}]=2i,
$$
$$
{\mathcal W}[{\bf u}_2^\alpha,{\bf u}_0^\alpha]=\mp 2,\quad
{\mathcal W}[{\bf u}_2^\beta,{\bf u}_0^\beta]=\pm 2,\quad
{\mathcal W}[{\bf u}_1^{\gamma,\alpha},{\bf u}_1^{\gamma,\beta}]=\pm 2,
$$
where the upper (resp. lower) sign corresponds to the case where the indicated 
turning point is a zero of $g_+$ (resp.  $g_-$).
\end{proof}

Let us investigate the quantization condition \eqref{QC3} for $\lambda$ close to 
(but different from) $\lambda_0$.  First, consider the actions
\begin{align*}
z_\alpha(\lambda)&=z(\gamma(\lambda),\lambda,\alpha(\lambda))\\
z_\beta(\lambda)&=z(\gamma(\lambda),\lambda,\beta(\lambda)).
\end{align*}
The configurations from I to VI of the Stokes lines in Figure \ref{stokesvariation} correspond 
to $\lambda$ satisfying the following conditions on the actions $z_\alpha$, $z_\beta$:
\begin{description}
\item[I] $\re z_\alpha(\lambda)=\re z_\beta(\lambda)<0$,
\item[II] $\re z_\beta(\lambda)<0=\re z_\alpha(\lambda)$,
\item[III] $\re z_\beta(\lambda)=0<\re z_\alpha(\lambda)$,
\item[IV] $\re z_\alpha(\lambda)=\re z_\beta(\lambda)>0$,
\item[V] $\re z_\alpha(\lambda)=0<\re z_\beta(\lambda)$,
\item[VI] $\re z_\alpha(\lambda)<0=\re z_\beta(\lambda)$,
\end{description}

Each condition defines a curve in the $\lambda$-plane starting from $\lambda_0$ 
(see Figure \ref{lines}).  For $\lambda$ not on any these curves, there is no connection 
between any two turning points, and such $\lambda$  cannot be o(1)-close to an 
eigenvalue for small enough $\epsilon$. 

\begin{figure}[htbp]
\centering
\includegraphics[bb=-100 0 674 576, width=10cm]{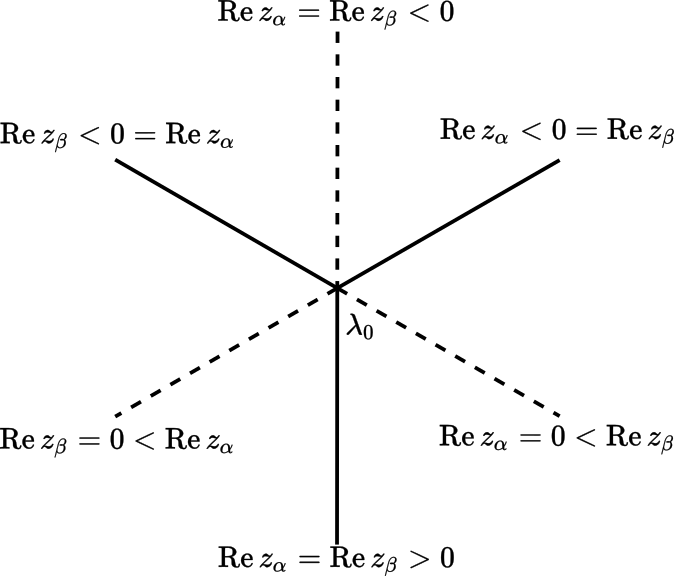}
\caption{Real part of the actions $z_\alpha=z(\gamma,\lambda,\alpha)$ and 
$z_\beta=z(\gamma,\lambda,\beta)$}
\label{lines}
\end{figure}

Consider first conditions I, III or V.  
In fact, the right hand side of the quantization condition (\ref{QC3}) is
exponentially small in case I and exponentially large in cases III and V,
and the quantization condition cannot hold for such $\lambda$'s.

Finally, in the cases II, IV and VI, the quantization condition \eqref{QC3} reduces 
to the usual one of Bohr-Sommerfeld type. We have
\begin{align}
\tilde m_\alpha(\epsilon)e^{2z(\gamma,\lambda,\alpha)/\epsilon}=1
\quad \text{in case II},\\
\tilde m_\beta(\epsilon)e^{2z(\gamma,\lambda,\beta)/\epsilon}=1
\quad \text{in case VI},\\
\tilde m_\gamma(\epsilon)e^{2z(\beta,\lambda,\alpha)/\epsilon}=1
\quad \text{in case IV},
\end{align}
where,  as $\epsilon\downarrow0$,
\begin{align}
\label{ma}
\tilde m_\alpha(\epsilon)=
\frac{m_\alpha(\epsilon)}{1-m_\beta(\epsilon)
e^{2z(\gamma,\lambda,\beta)/\epsilon}}=
\delta(\alpha,\gamma)+\CO(\epsilon), \\
\label{mb}
\tilde m_\beta(\epsilon)=
\frac{m_\beta(\epsilon)}{1-m_\alpha(\epsilon)
e^{2z(\gamma,\lambda,\alpha)/\epsilon}}=
\delta(\beta,\gamma)+\CO(\epsilon), \\
\label{mc}
\tilde m_\gamma(\epsilon)=-
\frac{m_\alpha(\epsilon)}{m_\beta(\epsilon)-
e^{2z(\gamma,\lambda,\beta)/\epsilon}}=
\delta(\alpha,\beta)+\CO(\epsilon).
\end{align}
Notice that the above asymptotics holds a long as the absolute values of the 
actions $z(\gamma,\lambda,\beta)$ in \eqref{ma}, $z(\gamma,\lambda,\alpha)$ 
in \eqref{mb}, $z(\gamma,\lambda,\beta)$ in \eqref{mc} are larger than 
$c\epsilon\log(1/\epsilon)$ for some positive constant $c$, which means 
$|\lambda-\lambda_0|\ge c'\epsilon\log(1/\epsilon)$.

\subsection{The Stokes geometry and quantization in a particular case}
\label{geometry}

In this paragraph, we present the geometric configuration of the very particular  case 
of \cite{mil}
\begin{equation}
\label{our-pot}
A(x)=S(x)=\sech(2x),\quad x\in\mathbb{R}.
\end{equation}
Actually, there is nothing really special about this example; it just happens that it was the 
first case studied numerically by  Bronski in \cite{bron} and subsequently by Miller in 
\cite{mil}. 

We begin with the turning points (see Definition \ref{definition-turn-pt}) of our problem.  
For fixed $\lambda\in\mathbb{C}$,  these are the zeros (in the complex $x$-plane) of 
[cf. (\ref{turn-pt})]
\begin{equation*}
-V_0(x,\lambda)=[\lambda-\tanh(2x)\sech(2x)]^2+\sech^2(2x).
\end{equation*}

In the present case  the potential $V_0(\cdot,\lambda)$ is periodic
with \textit{fundamental period} $i\pi$. For this reason, we define
the \textit{fundamental strip}
\begin{equation}
\label{strip}
\mathcal{S}=\{x\in\mathbb{C}|-\tfrac{\pi}{2}<\im x\leq\tfrac{\pi}{2}\}
\end{equation}
of the complex $x$-plane and restrict our attention only to what happens
in $\mathcal{S}$. For all $\lambda\in\mathbb{C}$, the potential
$V_0(\cdot,\lambda)$ has in $\mathcal{S}$ two fourth order $x$-poles located
at $x=\pm i\pi/4$. On the other hand, the zeros of  $V_0(\cdot,\lambda)$
are $\lambda$-dependent; there are precisely eight in $\mathcal{S}$
(counting multiplicity).

\begin{figure}[htbp]
\centering
\includegraphics[scale=1]{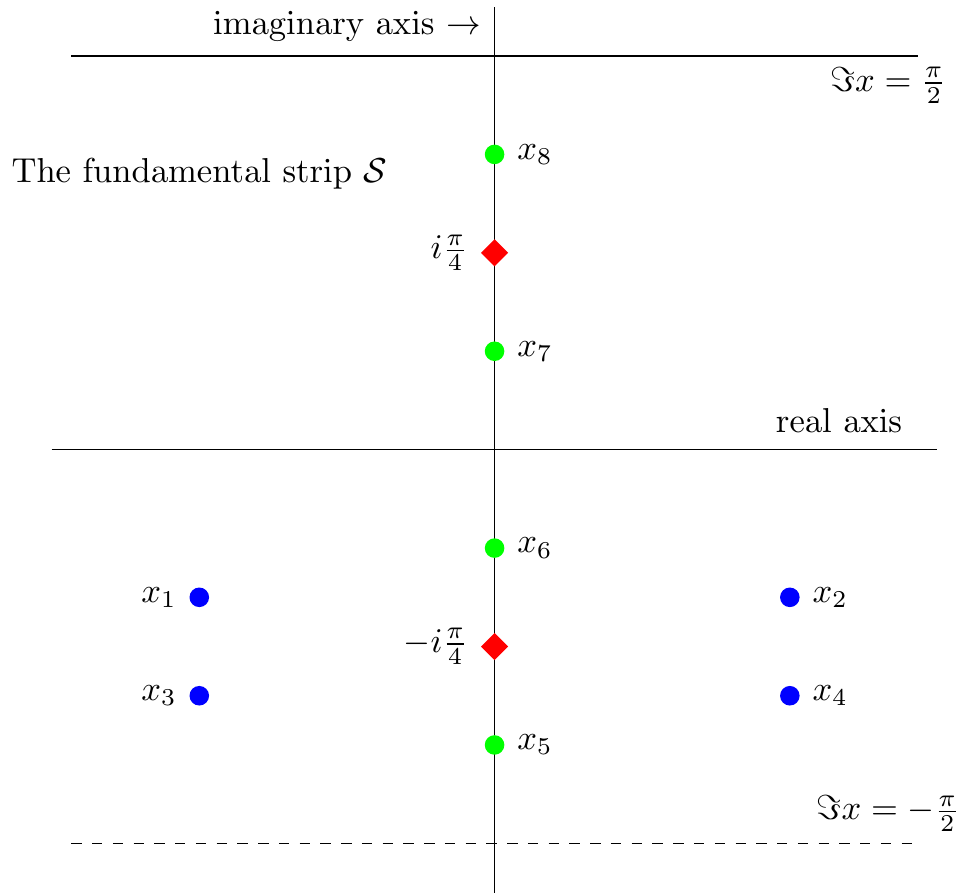}
\caption{A (qualitative) picture showing the turning point position with 
blue \& green circles (in the fundamental strip $\mathcal{S}$ on the 
$x$-complex plane) for a purely imaginary $\lam$ ($=0.2i$). The red 
diamonds denote the poles of the potential.}
\label{tp-for-imag-lam}
\end{figure}

When $\lambda$ is purely imaginary we have eight distinct
(and therefore simple) zeros in $\mathcal{S}$. There are four $x$-zeros
that lie on the imaginary axis, two between the poles and two outside.
The remaining four $x$-zeros (in the fundamental strip)
make up the vertices of a \textit{rectangle}. We label the turning points
in $\mathcal{S}$ as follows (the whole configuration is depicted in
Figure \ref{tp-for-imag-lam}).
\begin{enumerate}
\item[(i)]
The vertices of the rectangle
\begin{itemize}
\item
$x_1(\lambda)$ is
the turning point with negative real part and the greatest imaginary
part.
\item
$x_2(\lambda)=-x_1^*(\lambda)$.
\item
$x_3(\lambda)$ is the turning point with the same real part to $x_1(\lambda)$ and
the smallest imaginary part.
\item
$x_4(\lambda)=-x_3^*(\lambda)$.
\end{itemize}
\item[(ii)]
The four purely imaginary turning points are labeled in order of
increasing imaginary part. More precisely
\begin{itemize}
\item
The turning point between $-i\pi/2$ and $-i\pi/4$ is $x_5(\lambda)$.
\item
Between $-i\pi/4$ and $0$ lies $x_6(\lambda)$.
\item
$x_7(\lambda)$ is to be found between $0$ and
$i\frac{\pi}{4}$.
\item
Finally, the turning point between $i\pi/4$ and $i\pi/2$ is $x_8(\lambda)$.
\end{itemize}
\end{enumerate}
In Figure \ref{tp-for-imag-lam} we show  the general qualitative picture 
which appears in the actual numerics. The ``indexing" process above 
provides us with an unambiguous way for the labeling of the eight turning 
points $x_j(0.2i)$,  $j=1,\dots8$ in $\mathcal{S}$.  

\begin{remark}
\label{which-is-which}
Using this indexing, one can 
show that $x_1$, $x_2$, $x_6$ and $x_8$ are zeros of $g_-$ [see (\ref{gigi})] while 
$x_3$, $x_4$, $x_5$ and $x_7$ are zeros of $g_+$. 
\end{remark}

For $\lam$ not purely imaginary, the nice symmetry of the complex turning 
points  breaks down.  The turning points for non-imaginary values of $\lam$ can be 
uniquely defined by analytically continuing the turning points $x_j(0.2i)$, $j=1,\dots8$
along the L-shaped path (see Figure \ref{analytic-cont-lam}) from $0.2i$ to $\Re\lam+0.2i$ 
(horizontally) and from $\Re\lam+0.2i$ to $\lam$ (vertically).

\begin{figure}[htbp]
\centering
\includegraphics[scale=1]{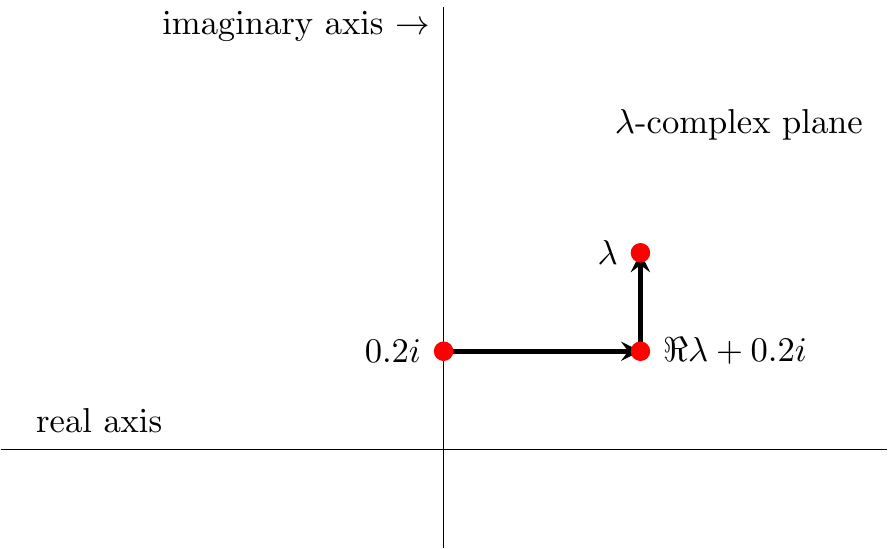}
\caption{The L-shaped analytic continuation in the $\lam$-complex plane
of the turning points $x_j(0.2i)$, $j=1,\dots8$ for a non-purely imaginary
$\lam$.}
\label{analytic-cont-lam}
\end{figure}

For some $\lam\in\C$, $V_0(\cdot,\lam)$ has double zeros which can be found by 
simply solving (with respect to $x$) simultaneously the equations 
$V_0(x,\lam)=\tfrac{d}{dx}V_0(x,\lam)=0$. It is  useful to compute these double zeros
since it turns out that they  are related  to the end-points of the (semiclassical) 
asymptotic spectrum. One easily sees that these double zeros  are present 
only for four values of $\lam$, namely
\be\label{lamda-double}
\lam_D=
i\sigma\sqrt{\frac{1}{2}+\frac{1+i\tau\sqrt{7}}{8}}
\Bigg(1-\frac{1+i\tau\sqrt{7}}{4}\Bigg)
\quad\text{where}\quad\sigma,\tau=\pm1
\ee
and are given by the solutions of the transcendental equations
\footnote{Each one of these four equations provides us with a double
turning point in $\mathcal{S}$.}
\be\nn
\tanh(2x)=\frac{1+i\tau\sqrt{7}}{4i\sigma}
\quad\text{where}\quad\sigma,\tau=\pm1
\ee
We adopt the notation 
$\lam_D^{(j)}$, $j=1,\dots,4$, by requiring that  the point $\lam_D^{(j)}$ is located in 
the $j^{\rm th}$ quadrant in the complex $\lam$-plane. 
Observe that the set of the four $\lam_D$'s in (\ref{lamda-double}) can be written as
$(\lam,-\lam,\lam^*,-\lam^*)$. These four $\lam_D$'s turn out to be
-as we shall eventually discover numerically- the endpoints of the four branches of a
$\mathlarger{\mathlarger{\yud}}$-shaped asymptotic spectrum in the complex $\lam$-plane, 
symmetric to both of the coordinate axes. The \textsf{Y}-shaped set consisting of the points 
in the upper half-plane $\mathbb{H}^+$ is the union of three arcs which intersect 
at a point $i\mu_{\otimes}$ (the numerical experiments in \cite{mil} indicate that 
$\mu_{\otimes}\approx0.28$).   

We  now proceed to investigate for which $\lam\in\C$ there exists an admissible 
contour (as defined in paragraph \S\ref{evs}). For these values of $\lam$ the WKB 
analysis for locating the eigenvalues will  be possible. For each $j,k=1,\dots,8$ 
with $j\neq k$, we now define the complex-valued \textit{action integrals}
\be\label{action-int}
I_{jk}(\lam)=
z(x_k(\lambda),\lambda,x_j(\lambda))=
i\int_{C^0(x_j(\lambda)\rightsquigarrow x_k(\lambda))}
\sqrt{-V_0(t,\lam)}dt.
\ee
We construct well-defined 
branches for these functions by defining signs of the square root for
$\lam=0.2i$ and applying analytic continuation along the same L-shaped
path as was used to define the turning points $x_j(\lam)$, $j=1,\dots,8$
(again refer to Figure \ref{analytic-cont-lam}). For reasons that will be
explained later, we restrict our attention only to $I_{12}$, $I_{16}$ and $I_{26}$.

\begin{figure}[htbp]
\begin{center}
\includegraphics[width=90mm]{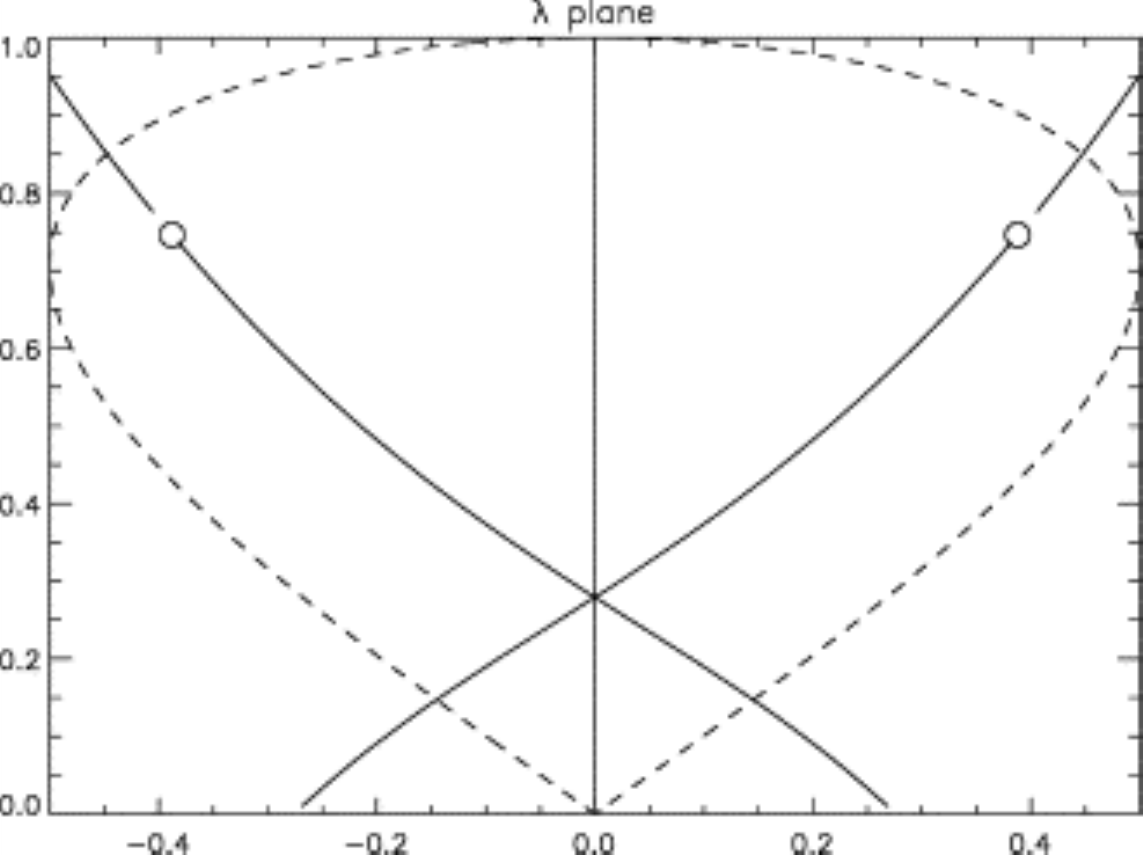}
\caption{
Candidates for the asymptotic spectral arcs on the $\lambda$-upper-half-plane.  
The dashed curve in this figure (which is Fig. 3 in \cite{mil}) is the so called 
\textit{real-turning-point curve} 
(it is the set
$\{\,
\lambda\in\mathbb{C}\mid\exists\hspace{2pt}r\in\mathbb{R}
\hspace{2pt}:\hspace{2pt}V_0(r,\lambda)=0
\,\}$);
it plays no role in our analysis.}
\label{curves-in-lambda-upper-half-plane}
\end{center}
\end{figure}

We confine ourselves to the upper $\lam$-half-plane, since the spectrum of
the Dirac operator in (\ref{dirac}) is symmetric with respect to the real
axis. The three conditions
$\re[I_{12}(\lam)]=0$, $\re[I_{16}(\lam)]=0$ and $\re[I_{26}(\lam)]=0$
[recall (\ref{lambda-relation})] yield three curves in the upper
$\lam$-half-plane (see Figure \ref{curves-in-lambda-upper-half-plane}).  
We label them as follows
\begin{align*}
\Lambda_{12} & =\{\,\lam\in\mathbb{H}^+\mid\re[I_{12}(\lam)]=0\,\},\\
\Lambda_{16} & =\{\,\lam\in\mathbb{H}^+\mid\re[I_{16}(\lam)]=0\,\},\\
\Lambda_{26} & =\{\,\lam\in\mathbb{H}^+\mid\re[I_{26}(\lam)]=0\,\}.
\end{align*}
Actually, the curve $\Lambda_{12}$ coincides with the upper
imaginary semi-axis while the curves $\Lambda_{16}$ and
$\Lambda_{26}$ have negative and positive slopes respectively.
There is  only one point of intersection in the upper
$\lam$-half-plane, which lies on the imaginary axis and is denoted by
$\lam_{\otimes}=i\mu_{\otimes}$. It is easily observed that $\lam_D^{(1)}$
belongs to the curve $\Lambda_{26}$ while $\lam_D^{(2)}=-\lam_D^{(1)\hspace{2pt}*}$ 
belongs to the curve $\Lambda_{16}$.

\begin{figure}[htbp]
\centering
\includegraphics[scale=1]{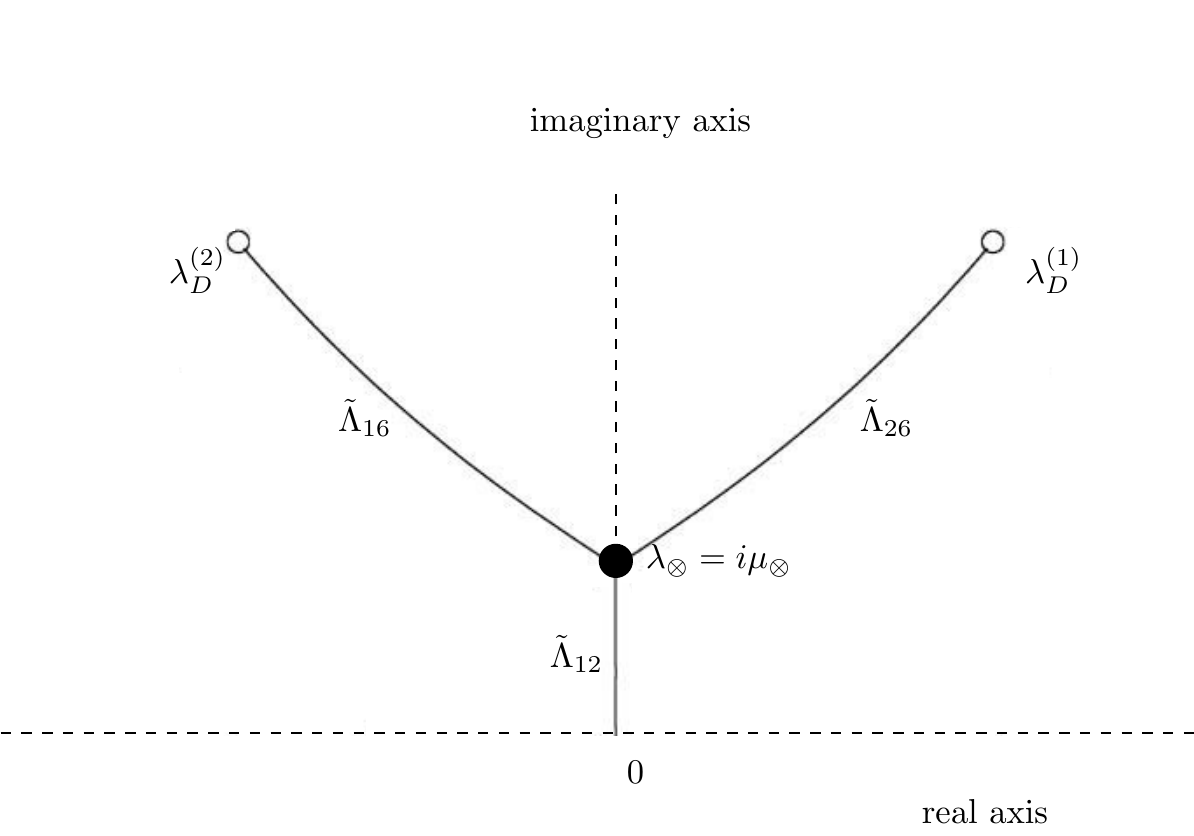}
\caption{The asymptotic spectral arcs on the $\lam$-plane.}
\label{cut-the-L}
\end{figure}

Our understanding is that \cite{mil} only considers $I_{12}$, $I_{16}$ and $I_{26}$ because  
the curves described by the rest of the equations $\Re[I_{jk}(\lam)]=0$ (as well as curves in 
the $\lam$-plane arising from equations corresponding to pairs of turning points in different 
period strips) play no new role in the WKB analysis of the particular potentials (\ref{our-pot}).  
More precisely, even though Stokes curves connecting the respective turning points exist, 
there are no progressive paths connecting to $\pm\infty$. For similar reasons one
eventually  eliminates portions of the curves $\Lambda_{12}$, $\Lambda_{16}$ and
$\Lambda_{26}$. Specifically, we only consider (see Figure \ref{cut-the-L})
\begin{itemize}
\item
$\tilde{\Lambda}_{12}$, the part of ${\Lambda}_{12}$ that lies between $0$ and 
$\lam_{\otimes}$, endpoints included
\item
$\tilde{\Lambda}_{26}$, which   lies between $\lam_{\otimes}$ and
$\lam_D^{(1)}$, endpoints included
\item
$\tilde{\Lambda}_{16}$, which lies between
$\lam_{\otimes}$ and $\lam_D^{(2)}$,  endpoints included.
\end{itemize}
These are the asymptotic spectral arcs defined in Definition \ref{asympt-spec-arc}, 
in our specific case (\ref{our-pot}).  
The union of these three arcs defines a \textsf{Y}-shaped set in the upper $\lam$-half-plane 
which is the part of the asymptotic spectrum that lies on the upper-half-plane. It has been 
numerically  observed in \cite{mil} that the semiclassical spectrum observed by Bronski 
in \cite{bron} coincides with this asymptotic spectrum computed here!

We must point out however, that we have a different argument why there are no eigenvalues
away from this  union of asymptotic spectral arcs. The numerics of \cite{mil} show that for
$\lambda$ away from this  union of asymptotic spectral arcs there is always a progressive path from
$-\infty$ to $+\infty$ which coincides with the real axis for large $|x|$.
From standard WKB theory these $\lambda$'s cannot be eigenvalues 
\footnote{See also Remark \ref{no-evs-remark-kamvi} at end of this section.}.

\begin{figure}[htbp]
\begin{center}
\includegraphics[width=60mm]{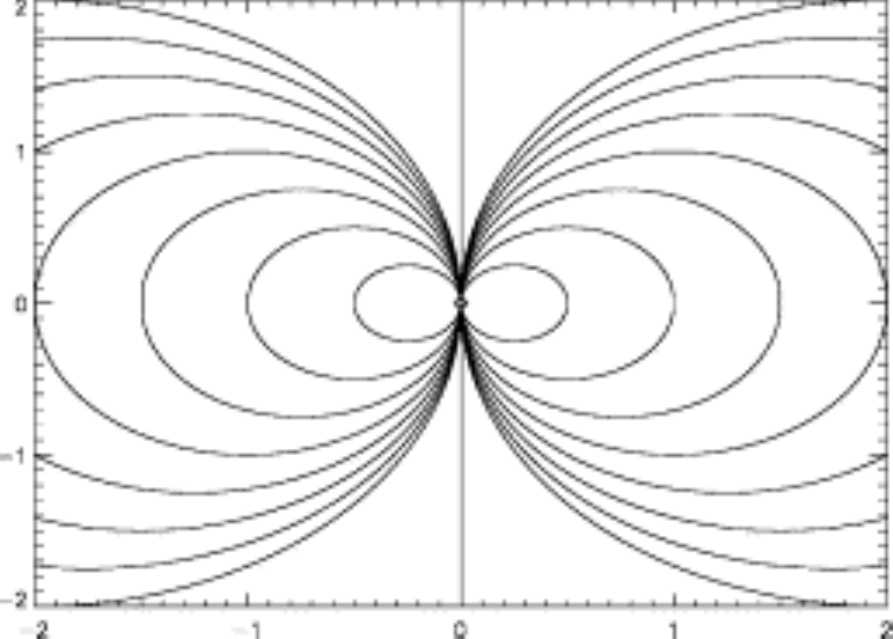}
\end{center}
\caption{The behavior of Stokes lines near a pole in $x$-plane.  
This is Fig. 5 in \cite{mil}.}
\label{x-behavior-near-a-pole}
\end{figure}

We now examine carefully the geometry of the solutions of the differential 
equation (\ref{de-real-path}) in the complex $x$-plane. From each simple turning point 
$x_j(\lam)$ three orbits of (\ref{de-real-path})  are emerging at angles of $2\pi/3$ 
(on the other hand  there is an infinite number of orbits meeting at each forth-order pole; 
see Figure \ref{x-behavior-near-a-pole}). These are Stokes curves of (\ref{de-real-path}). 

\begin{figure}[htbp]
\begin{center}
\includegraphics[width=110mm]{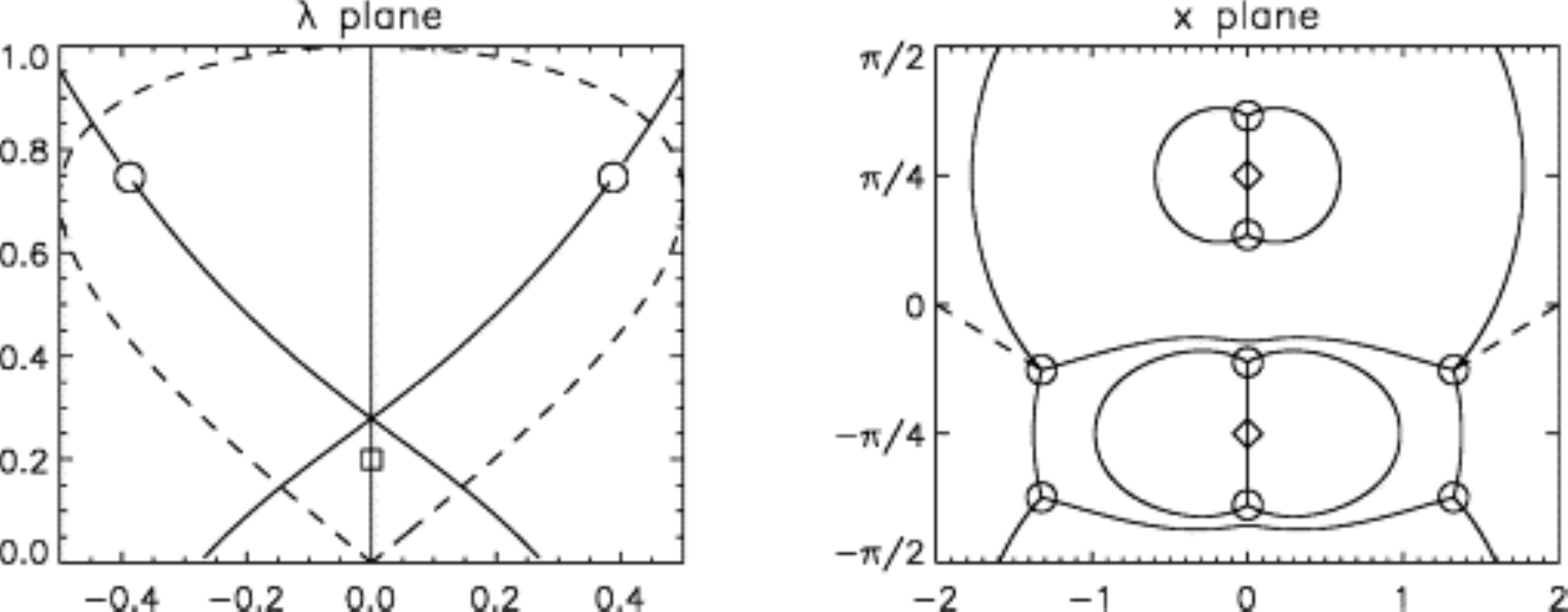}
\end{center}
\caption{On the left, the $\lambda$-plane with a small square indicating point  $\lambda=0.2i\in\tilde \Lambda_{12}$, below the bifurcation point.  On the right, for that  particular of $\lambda$, the fundamental strip $\mathcal{S}$ in the $x$-plane, showing clearly the Stokes lines that emanate from all turning points. This is Fig. 4 in \cite{mil}.}
\label{lambda-02i-and-x-planes}
\end{figure}

For $\lam=0.2i\in \tilde \Lambda_{12}$ (located below $\lam_\otimes$ on the imaginary axis),  
the corresponding situation of the $x$-plane is shown in Figure \ref{lambda-02i-and-x-planes}.  
We observe that there is a Stokes line $C^0$ connecting $x_1(0.2i)$ and $x_2(0.2i)$.  
Furthermore, there exist progressive paths $C^-$, $C^+$ joining $x_1(0.2i)$ to $-\infty$ and  
$x_2(0.2i)$ to $+\infty$ respectively 
\footnote{
Note here that while the Stokes line is uniquely defined, there is some freedom in the 
choice of the two progressive paths.}.

\begin{figure}[htbp]
\begin{center}
\includegraphics[width=110mm]{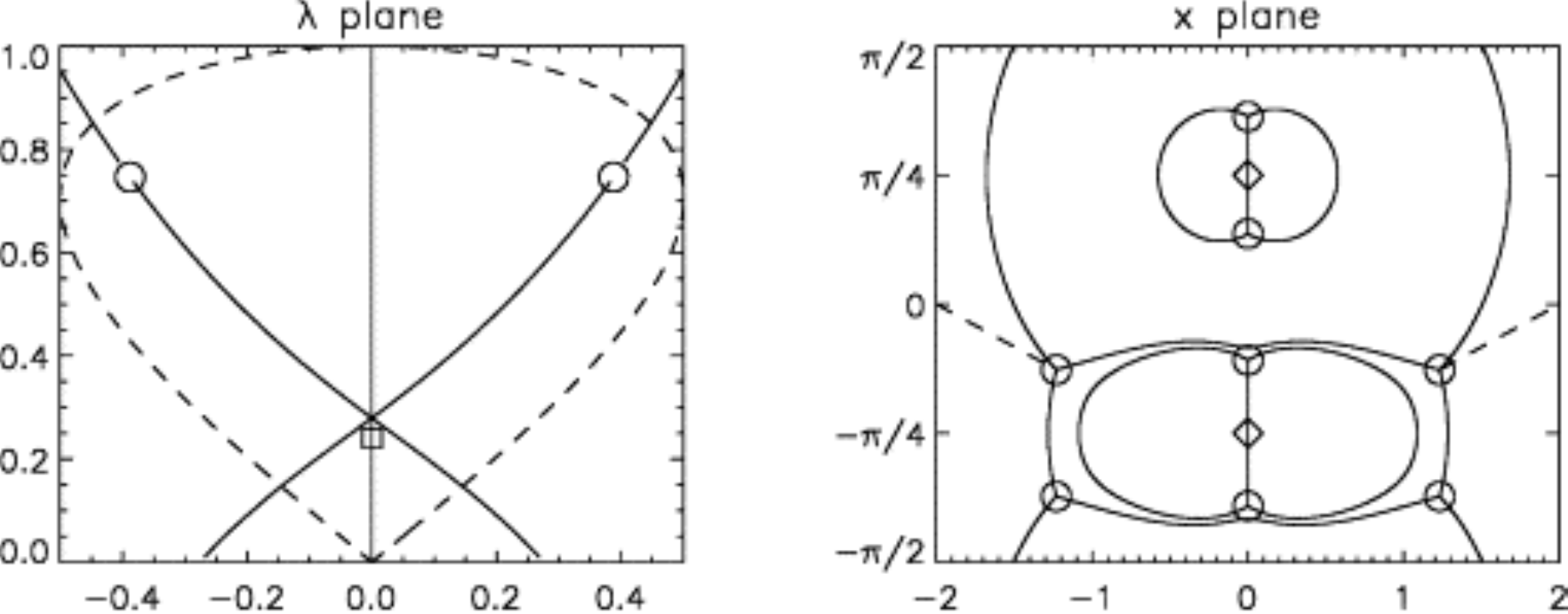}
\end{center}
\caption{On the left, the $\lambda$-plane, with a small square indicating a point  just below close  the bifurcation point $\lam_\otimes\approx0.28i$.  
On the right, for that particular value of $\lam$, the strip $\mathcal{S}$ of the $x$-plane and the Stokes lines in $\mathcal{S}$ that emanate from all turning points. Compare with case IV of Figure \ref{stokesvariation}.
This is Fig. 6 in \cite{mil}.}
\label{lambda-close-bifur-and-x-planes}
\end{figure}

For a value of $\lam$ on the imaginary axis coming very close to  the triple
intersection point $\lam_\otimes$ (but still imaginary and just below $\lam_{\otimes}$),  
the situation on the $x$-plane still looks similar to that we already examined 
(cf.  Figure \ref{lambda-close-bifur-and-x-planes}). The complex turning point $x_6(\lam)$ 
is moving up very close to the Stokes line that continues to connect $x_1(\lam)$ and 
$x_2(\lam)$. This proximity becomes more evident  as $\lam$ approaches the triple intersection 
point $\lam_\otimes$. The progressive paths emerging from $x_1(\lam)$, $x_2(\lam)$ continue 
to exist.

\begin{figure}[htbp]
\begin{center}
\includegraphics[width=110mm]{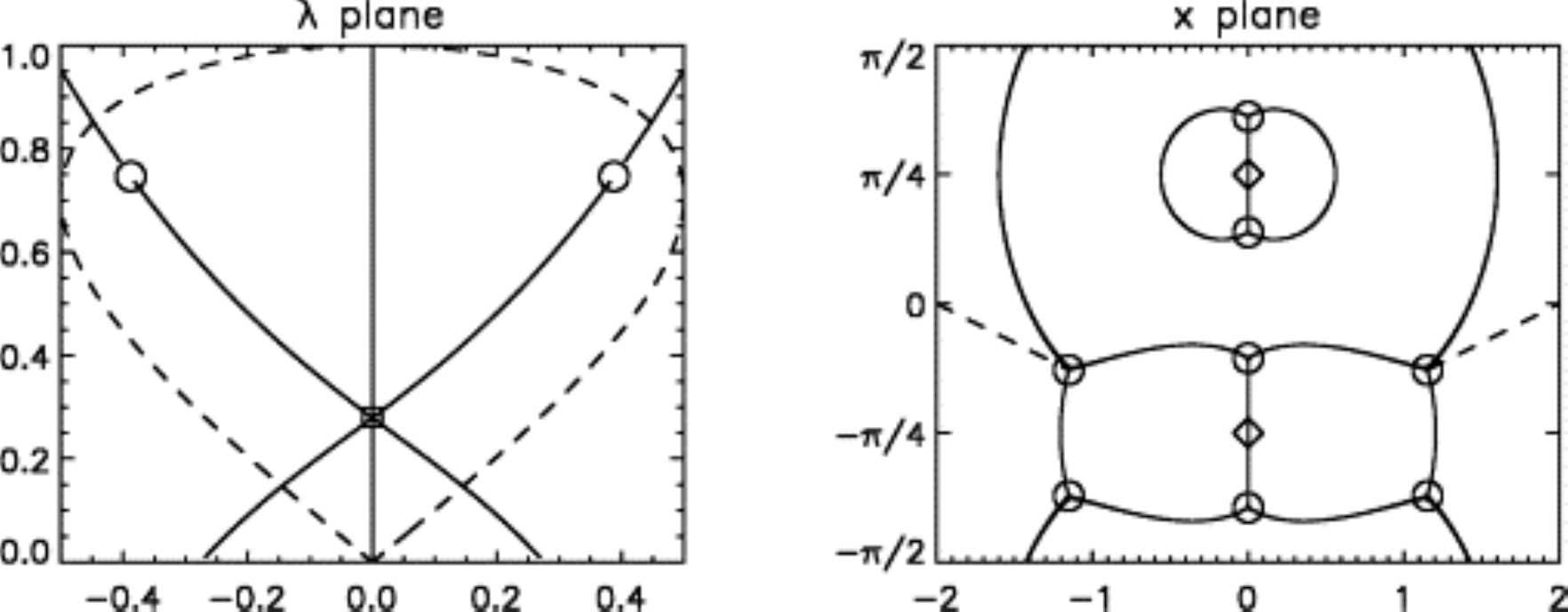}
\end{center}
\caption{On the left, the $\lambda$-plane in the extreme case when 
$\lambda$ is right at the bifurcation point $\lam_\otimes$.  On the right,
the $x$-plane with the Stokes lines in $\mathcal{S}$. Compare with 
Figure \ref{stokes2} or the central picture of Figure \ref{stokesvariation}. 
This is Fig. 7 in \cite{mil}.}
\label{lambda-bifur-and-x-planes}
\end{figure}

When $\lam=\lam_\otimes$ (see Figure \ref{lambda-bifur-and-x-planes}) 
there is no Stokes line connecting $x_1(\lam_\otimes)$ and $x_2(\lam_\otimes)$ 
without passing  through another turning point.  
Parts of the two paths previously connecting $x_6(\lam_\otimes)$ to 
$x_5(\lam_\otimes)$ have now partly  coalesced with the path  connecting 
$x_1(\lam_\otimes)$ and $x_2(\lam_\otimes)$.  Hence
$x_1(\lam_\otimes)$ and $x_2(\lam_\otimes)$ are now connected to 
$x_6(\lam_\otimes)$ via Stokes lines.  But the remaining Stokes line from 
$x_6(\lam_\otimes)$ continues to pass through the pole at $-i \pi/4$.  
The three conditions
$\re[I_{12}(\lam_\otimes)]=0$, $\re[I_{26}(\lam_\otimes)]=0$, $\re[I_{16}(\lam_\otimes)]=0$
are satisfied simultaneously and the corresponding Stokes lines meet at  
$x_6(\lam_\otimes)$ and at a $2\pi/3$ angle from each other. For this special value of 
$\lam$, there is a connected sequence of two Stokes lines and the progressive paths 
continue to exist. 

\begin{figure}[htbp]
\centering
\includegraphics[scale=.3]{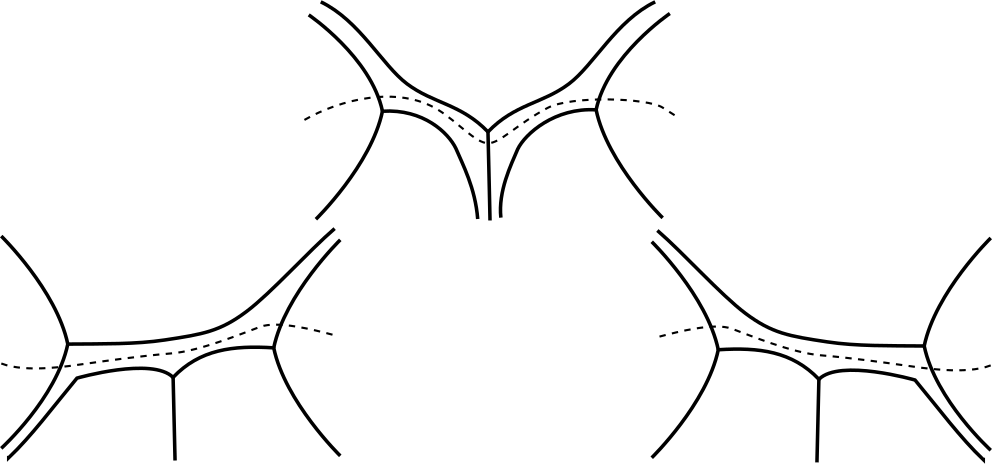}
\caption{Progressive paths (dashed curves) in the $x$-plane when 
$\lambda$ is slightly off the asymptotic spectrum; cf. Figure 5, cases I, III, V.}
\label{progro-paths}
\end{figure}

\begin{figure}[htbp]
\begin{center}
\includegraphics[width=120mm]{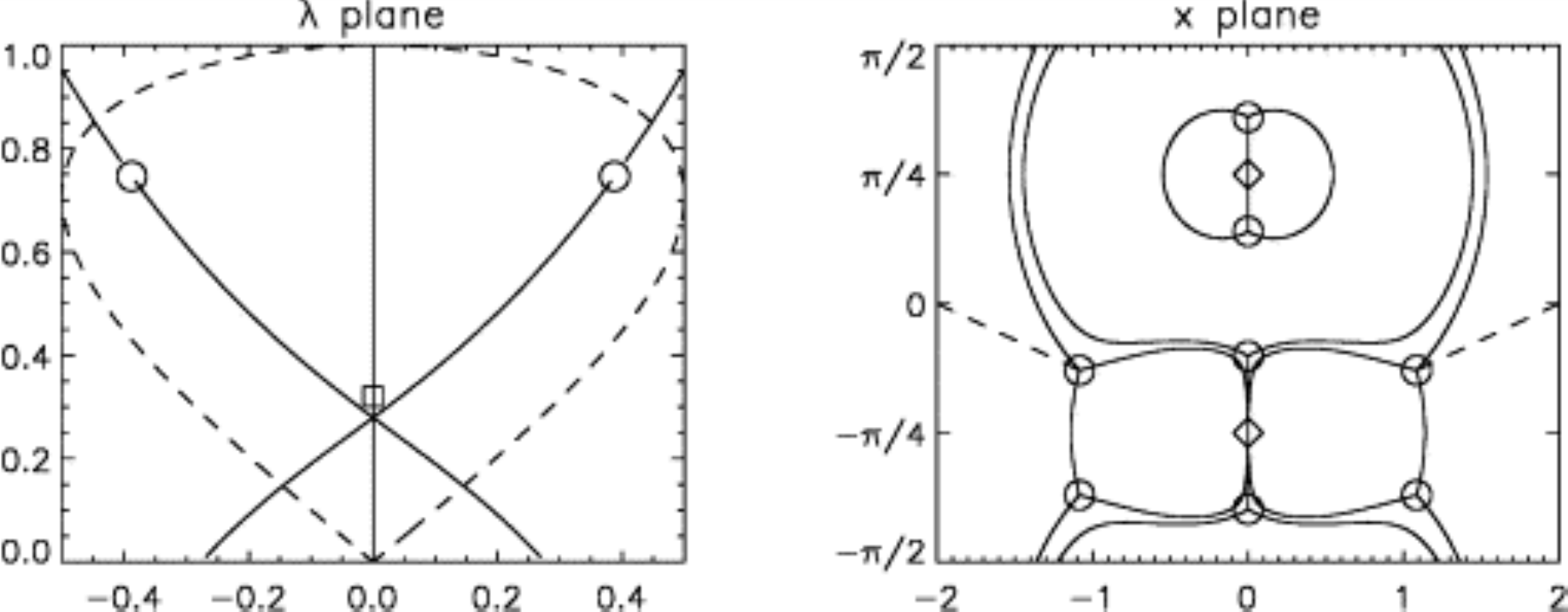}
\end{center}
\caption{The $\lambda$-plane on the left,  for a specific $\lambda\in\Lambda_{12}$ 
just above the bifurcation point $\lam_\otimes$.  On the right,  the $x$-plane with the 
Stokes lines in $\mathcal{S}$. Compare with case I of Figure \ref{stokesvariation}. 
Note that one can draw a progressive path from $-\infty$ to $+\infty$ which passes 
above $x_1$, just below $x_6$ and above $x_2$ (see Figure \ref{progro-paths}, top). 
Hence there is no eigenvalue semiclassically. This is Fig. 8 in \cite{mil}.}
\label{lambda-above-bifur-and-x-planes}
\end{figure}

Having completed the description of what happens when $\lam\in\tilde\Lambda_{12}$, there
are five more remaining directions in the complex $\lam$-plane all meeting
at the triple intersection point $\lam_\otimes$ (cf.  Figure \ref{cut-the-L}).

\begin{figure}[htbp]
\begin{center}
\includegraphics[width=120mm]{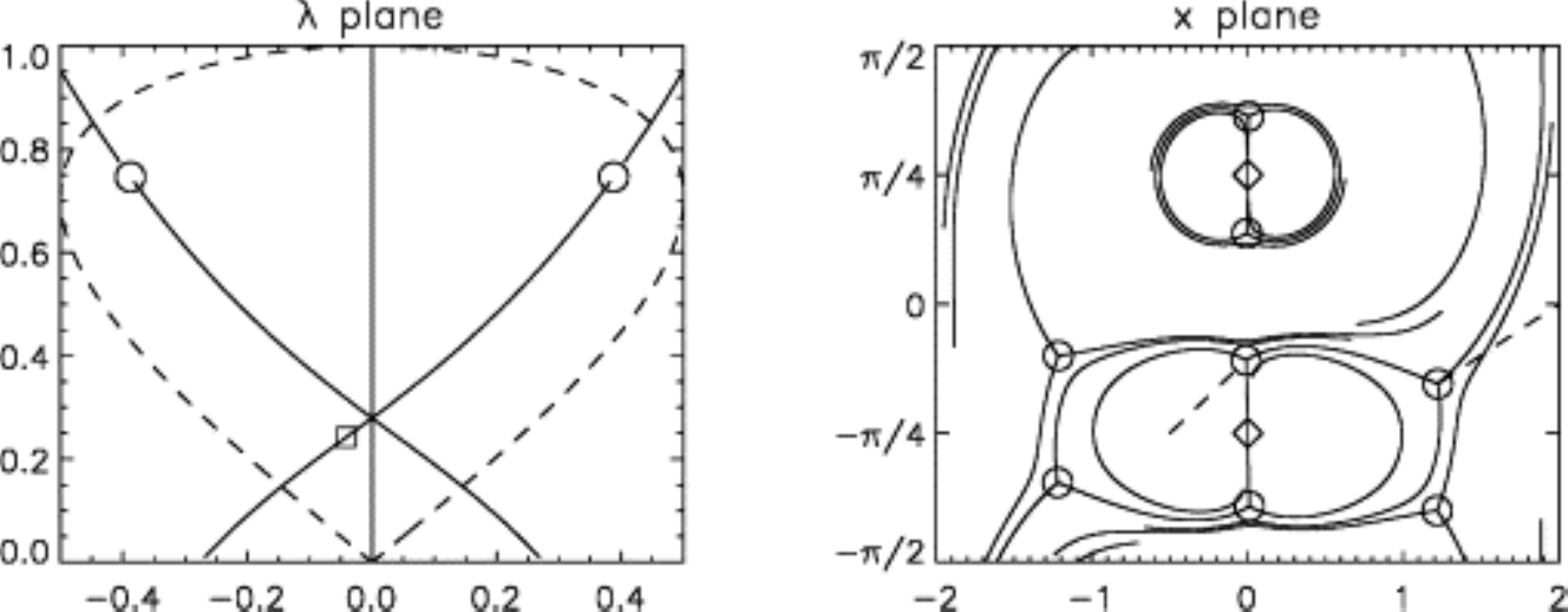}
\end{center}
\caption{The $\lambda$-plane on the left,  for a particular $\lambda\in\Lambda_{26}$ 
just below $\lam_\otimes\approx0.28i$.  On the right,  the $x$-plane with 
the Stokes lines in $\mathcal{S}$. Compare with case III of Figure \ref{stokesvariation}. 
Note that one can draw a progressive path from $-\infty$ to $+\infty$ which passes 
below $x_1$, just above $x_6$ and above $x_2$ (see Figure \ref{progro-paths}, bottom left). 
Hence there is no eigenvalue semiclassically. This is Fig. 9 in \cite{mil}.}
\label{no-wkb-1}
\end{figure}

Now suppose  we continue along the imaginary
axis above $\lam_\otimes$ (see Figure \ref{lambda-above-bifur-and-x-planes}). 
There is no more admissible connection between the  turning points 
$x_1(\lam)$, $x_2(\lam)$. There are paths emanating from $x_1(\lam)$, $x_2(\lam)$ 
(see Definition \ref{x-appropriate-admissible-curve}) meeting at the  double pole of 
$\sqrt{-V_0(\cdot,\lam)}$ so the holomorphic deformation 
is impossible, requiring us to restrict the 
WKB analysis described so far only to the section $\tilde \Lambda_{12}$. 
Note the existence of a progressive path from $-\infty$ to $+\infty$ which passes 
above $x_1$, just below $x_6$ and above $x_2$ (see dashed curve at the top of 
Figure \ref{progro-paths}). Hence there is no eigenvalue semiclassically.

\begin{figure}[htbp]
\begin{center}
\includegraphics[width=120mm]{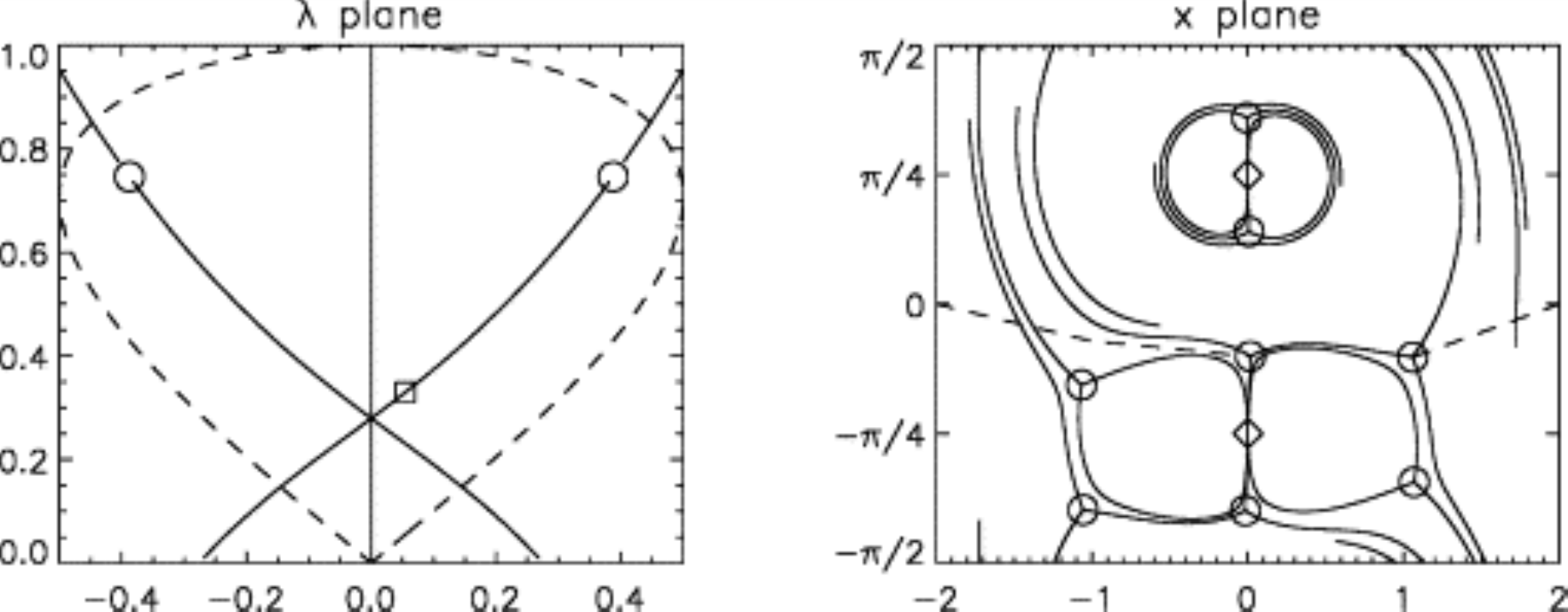}
\end{center}
\caption{The $\lambda$-plane on the left,  for a particular $\lambda\in\tilde{\Lambda}_{26}$ 
just above $\lam_\otimes\approx0.28i$.  On the right,  the $x$-plane with 
the Stokes lines in $\mathcal{S}$. Compare with case VI of Figure \ref{stokesvariation}. 
This is Fig. 10 in \cite{mil}.}
\label{no-wkb-2}
\end{figure}

Same observations for different pairs of turning points show that the WKB analysis 
only works for $I_{12}$, $I_{16}$ and $I_{26}$ (and indeed only for the sections 
$\tilde \Lambda_{12}$, $\tilde \Lambda_{26}$ and $\tilde \Lambda_{16}$).  One can 
examine the rest of the cases by checking Figures \ref{no-wkb-1} through \ref{no-wkb-5}.  

\begin{figure}[htbp]
\begin{center}
\includegraphics[width=120mm]{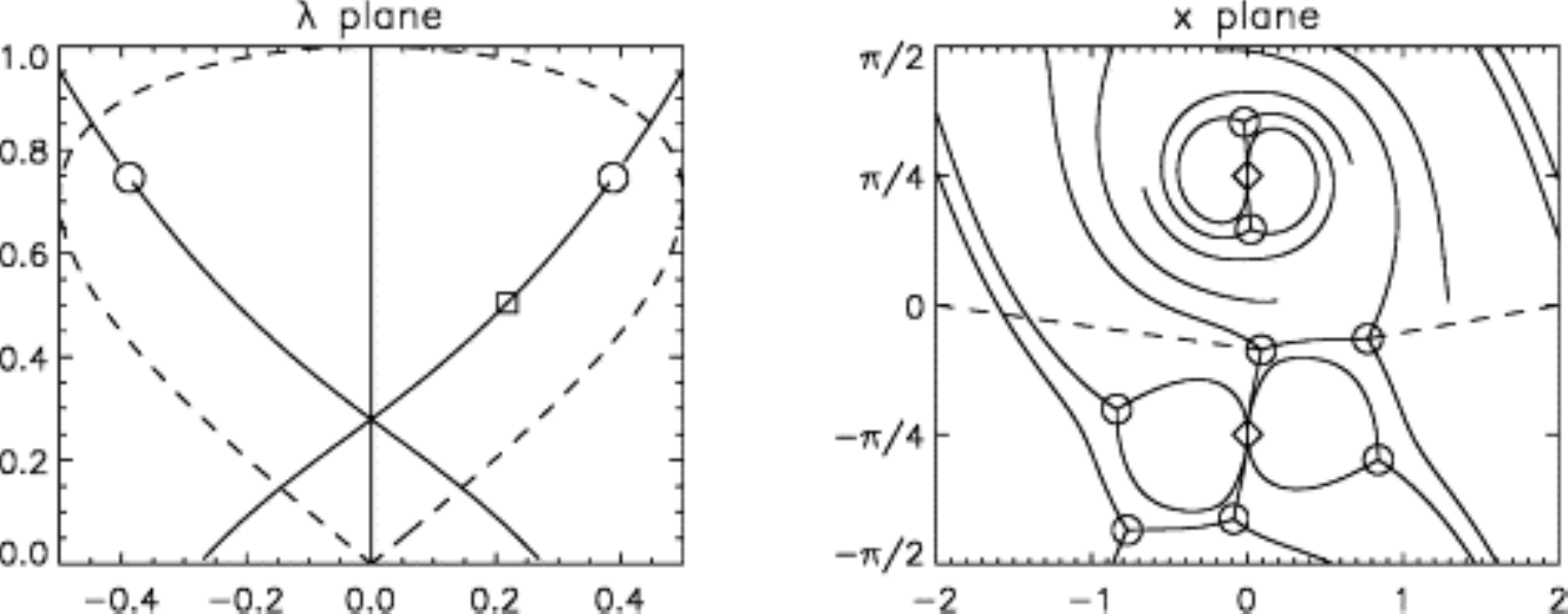}
\end{center}
\caption{The $\lambda$-plane on the left,  for a particular $\lambda\in\tilde{\Lambda}_{26}$ 
with $\Im\lambda$ greater than that shown in Figure \ref{no-wkb-2} but less than 
$\Im\lambda_D^{(1)}$.  On the right,  the $x$-plane with the Stokes lines 
in $\mathcal{S}$.  This is Fig. 11 in \cite{mil}.}
\label{no-wkb-3}
\end{figure}

Finally we give the quantization condition of eigenvalues along 
$\tilde \Lambda_{12}$, $\tilde \Lambda_{26}$ and $\tilde \Lambda_{16}$ away from the  
bifurcation point $\lam_\otimes$, and in a neighborhood of $\lam_\otimes$ using the 
results of the previous subsections.

\begin{figure}[htbp]
\begin{center}
\includegraphics[width=120mm]{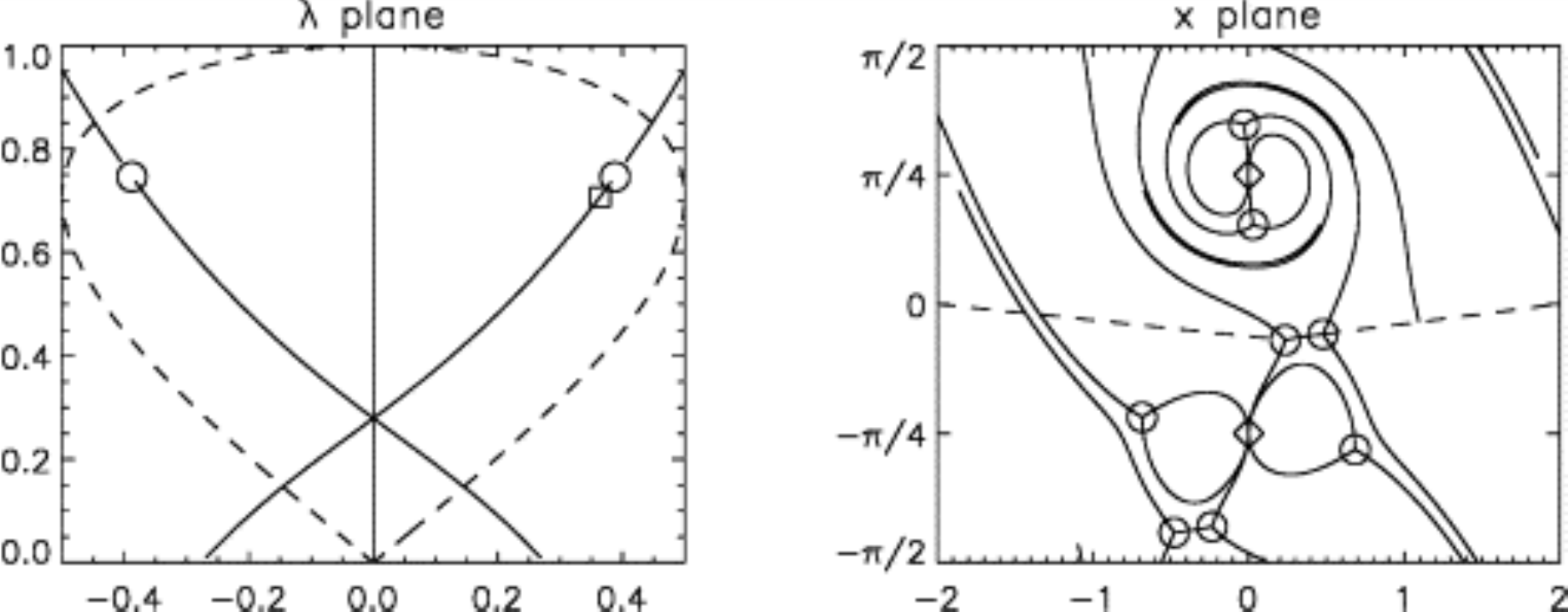}
\end{center}
\caption{On the left, the $\lambda$-plane for a specific $\lambda\in\tilde{\Lambda}_{26}$ 
just below $\lambda_D^{(1)}$.  On the right,  the $x$-plane with the Stokes 
lines in $\mathcal{S}$.  This is Fig. 12 in \cite{mil}.}
\label{no-wkb-4}
\end{figure}

\begin{theorem}
Let $(j,k)$ be $(1,2), (2,6)$ or $(1,6)$. For any $\epsilon$-independent $\delta>0$ and 
$\lambda_0\neq0$ on $\tilde \Lambda_{jk}$, fixed and independent of $\epsilon$, satisfying 
$|\lambda_0-\lambda_\otimes|>\delta$, there exist a complex neighborhood $U$ of $\lambda_0$ 
and a function $m_{jk}(\epsilon)$ with a uniform asymptotic behavior 
$m_{jk}(\epsilon)=-1+\CO(\epsilon)$ in $U$, such that $\lambda\in U$ is an eigenvalue of 
$\mathfrak{D}_\epsilon$  if and only if
$$
m_{jk}(\epsilon)e^{2z(\beta,\lambda,\alpha)/\epsilon}=1.
$$
\end{theorem}

\begin{figure}[htbp]
\begin{center}
\includegraphics[width=120mm]{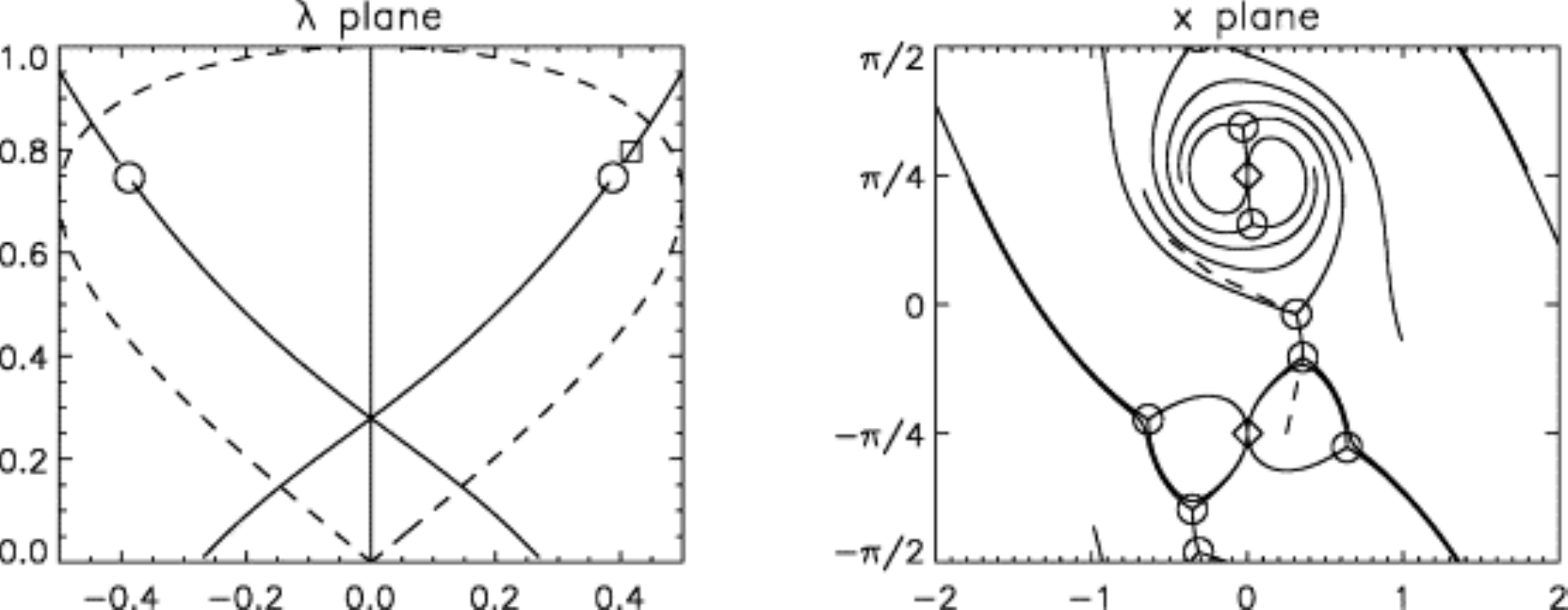}
\end{center}
\caption{On the left, the $\lambda$-plane for a specific $\lambda\in\Lambda_{26}$ 
just above $\lambda_D^{(1)}$.  On the right,  the $x$-plane with the Stokes lines in 
$\mathcal{S}$. Note that one can draw a progressive path that is a deformation of the 
real line in such a way that it crosses the Stokes line from $x_2$ to $x_6$ transversally. 
Hence there is no eigenvalue semiclassically. This is Fig. 13 in \cite{mil}.}
\label{no-wkb-5}
\end{figure}

\begin{theorem}
There exists a neighborhood $U$ of $\lambda_\otimes$ and  functions $m_\alpha(\epsilon)$, 
$m_\beta(\epsilon)$ with uniform asymptotic behaviors $m_\alpha(\epsilon)=-1+\CO(\epsilon)$, 
$m_\beta(\epsilon)=-1+\CO(\epsilon)$ in $U$, such that 
$\lambda\in U$ is an eigenvalue of  $\mathfrak{D}_\epsilon$  if and only if
\begin{equation}
\label{qc3}
m_\alpha(\epsilon)e^{2z(\gamma,\lambda,\alpha)/\epsilon}+
m_\beta(\epsilon)e^{2z(\gamma,\lambda,\beta)/\epsilon}=
1.
\end{equation}
\end{theorem}

\begin{remark}
\label{no-evs-remark-kamvi}
We claim that no eigenvalue can appear  away  from the asymptotic spectral arcs. 
Although  we do not present  a rigorous proof of this fact, it 
follows from the numerical investigation of the geometry of turning points and Stokes lines
for different values of $\lambda$.
For example, Figure \ref{lambda-above-bifur-and-x-planes}, Figure \ref{no-wkb-1} and 
Figure \ref{no-wkb-5} show that for  $\lambda$ slightly off the asymptotic spectral arcs,
a progressive path from $x=-\infty$ to $x=+\infty$ does exist,
which is in fact real near $x=\pm \infty$.  It follows easily from standard WKB theory 
that such $\lambda$ cannot be eigenvalues. Indeed, in view of part (ii) of Theorem 
\ref{formal-to-rigorous}, and the fact that $S$ is real near $x=\pm \infty$, 
the important factor in formula (\ref{wkb-sols}) for ${\bf u}^+$ that determines the 
behavior at $x=\pm\infty$, is the exponential $\exp(z/\epsilon)$. If $\Re z$ is monotonic 
along a path from $-\infty$ to $+\infty$ then it is impossible that the solution decays 
at both ends.
\end{remark}

\begin{remark}
\label{finite-gap-remark}
The question arises whether there is a possibility of a general proof of the fact 
that no eigenvalue can appear  away  from the asymptotic spectral arcs. A consequence 
of this would be a general proof that the limiting spectrum (the minimal set such that 
every actual eigenvalue lies $\epsilon$-close to this set) would be what we have called 
the asymptotic spectrum (which we defined in terms of the existence of an admissible path). 
In turn that would prove that the limiting spectrum is a union of analytic arcs.

A possible strategy is suggested by an analogous problem appearing in the semiclassical 
analysis of the inverse scattering problem in \cite{kmm} where a crucial ingredient of the 
proof is the existence of appropriate \say{steepest descent paths}. The existence proof 
in \cite{kr} requires the study of an associated max-min variational problem.

In our own problem here, we can define a generalized admissible contour as a contour 
that is either admissible in the sense of Definition \ref{general-admissible-curve} 
(extended to allow for double turning points as degenerate cases) or a progressive 
path from $x=-\infty$ to $x=+\infty$. The inequalities defining progressiveness and the 
equations defining Stokes lines can be interpreted as variational \say{Euler-Lagrange} 
conditions for a max-min problem. We expect that a proof along the lines of \cite{kr} 
will show the existence of the desired contour which is either a progressive path from 
$x=-\infty$ to $x=+\infty$ (hence no eigenvalue) or an admissible contour in the sense 
of Definition \ref{general-admissible-curve} (hence  an asymptotic spectral arc).
\end{remark}

\section{Eigenvalues That Lie Near Zero} 
\label{evs-near-zero}

\noindent
\underline{Note}: In this section  $A(x)=S(x)=\sech(2x)$ for simplicity. But the theory applies very generally.

\subsection{Preparations}
\label{preparations-near-0}

We now turn our attention to the semiclassical behavior of eigenvalues (and their 
corresponding norming constants) of the Zakharov-Shabat operator that 
lie near zero. We will focus on  a neighborhood $\mathscr{D}$ of the real axis (to be appropriately specified 
later on) and start with the problem
\be\nn
\mathfrak{D}_\epsilon\mathbf{u}(x)=\lambda\mathbf{u}(x),\quad
x\in\mathscr{D}
\ee
where
\begin{itemize}
\item
the Dirac (or Zakharov-Shabat) operator is defined by
\begin{equation}\nn
\mathfrak{D}_\epsilon=
\begin{bmatrix}
-\frac{\epsilon}{i}\frac{d}{dx} & -iA(x)\exp\{iA(x)/\epsilon\}\\
-iA(x)\exp\{-iA(x)/\epsilon\} & \frac{\epsilon}{i}\frac{d}{dx}
\end{bmatrix}
\end{equation}
\item
$\epsilon>0$ is the semiclassical parameter,
\item
$\mathbf{u}=[u_{1}\hspace{3pt}u_{2}]^T$ is a function from $\mathscr{D}$ to $\mathbb{C}^2$ and 
\item
$\lambda\in\C$ represents the spectral parameter. 
\end{itemize}

First,  we apply the change of variables [cf. equation (4) in \cite{mil}]
\be
\label{w-plus-minus}
w_{\pm}=
\frac{u_{2}\exp\{iA/(2\epsilon)\}\pm u_{1}\exp\{-iA/(2\epsilon)\}}
{\sqrt{A\pm i(\lambda+A'/2)}}.
\ee
Then for the lower sign (minus), the equation reads
\begin{multline}\nn
w_{-}''=
\Bigg[
-\epsilon^{-2}\Big\{A(x)^2+[\lambda+\tfrac{A'(x)}{2}]^2\Big\}\\
+
\frac{3}{4}\bigg\{ \frac{A'(x)-i\frac{A''(x)}{2}}{A(x)-i[\lam+\tfrac{A'(x)}{2}]} \bigg\}^2
-
\frac{1}{2}\frac{A''(x)-i\frac{A'''(x)}{2}}{A(x)-i[\lam+\tfrac{A'(x)}{2}]}
\Bigg]
w_{-}
\end{multline}
where we dropped the dependence of $w_{-}$ on $(x,\lambda,\epsilon)$ for notational 
simplicity and let prime denote differentiation with respect to $x$. Hence, dropping the subscript, 
and considering only the upper half $\lambda$-plane due to symmetry, we are led to the following 
proposition.
\begin{proposition}
Finding the discrete spectrum of $\mathfrak{D}_{\epsilon}$ is equivalent to finding the 
values of  $\lambda\in\mathbb{H}^+$ for which equation 
\be
\label{eq-for-spec-olver}
\frac{d^2w}{dx^2}=[\epsilon^{-2}V_0(x,\lambda)+g(x,\lambda)]w
\ee
has an $L^{2}(\mathscr{D})$ solution.  In this equation, the leading order potential 
$V_0(x,\lambda)$ satisfies
\be
\label{olver-f}
V_0(x,\lambda)=
-\Big\{A(x)^2+[\lambda+\tfrac{A'(x)}{2}]^2\Big\}
\ee
while the correction $g(x,\lambda)$ is given by 
\be
\label{olver-g}
g(x,\lambda)=
\frac{3}{4}\bigg\{ \frac{A'(x)-i\frac{A''(x)}{2}}{A(x)-i[\lam+\tfrac{A'(x)}{2}]} \bigg\}^2
-
\frac{1}{2}\frac{A''(x)-i\frac{A'''(x)}{2}}{A(x)-i[\lam+\tfrac{A'(x)}{2}]}.
\ee
\end{proposition} 
 
In this section, we are only interested in the behavior of the eigenvalues that lie near zero. 
So, let us first fix some set $\mathscr{O}$ on the spectral plane that satisfies the following 
condition.

\vspace{0.5cm}

\begin{minipage}{0.1\textwidth}
{\bf (H2):}
\end{minipage}
\hspace{0.01\textwidth}
\begin{minipage}{0.8\textwidth}
The set $\mathscr{O}$ is an open and simply connected set, proper subset of 
${\mathcal R}_0\cap\mathbb{H}^+$ [recall the formula for the mumerical range (\ref{r0}) 
and  of course here $S(x)=A(x)=\sech(2x)$] 
and located  
in such a way that its intersection with $\tilde{\Lambda}_{12}$ 
(see Figure \ref{cut-the-L}) is a single interval 
and  has no common points with either $\tilde{\Lambda}_{16}$ 
or $\tilde{\Lambda}_{26}$.
\end{minipage}

\vspace{0.5cm}
\noindent
From now on, we shall only consider eigenvalues $\lambda$ that lie in a set $\mathscr{O}$ 
satisfying {\bf (H2)}.

\subsection{Solutions in a neighborhood of only one simple turning point}
\label{sols-1-turn-pt}

In this subsection we shall confine the study of the  differential equation (\ref{eq-for-spec-olver}) 
to some domain in which it has only one simple turning point. Let us be more precise. 
\begin{assumption}
\label{assume-one-tp}
For $\epsilon>0$, consider the differential equation
\be\nn
\frac{d^2w}{dx^2}=[\epsilon^{-2}V_0(x,\lambda)+g(x,\lambda)]w,
\quad
(x,\lambda)\in\mathscr{D}\times\mathscr{O}
\ee
where $V_0(x,\lambda)$ is given by (\ref{olver-f}) and $g(x,\lambda)$ by (\ref{olver-g}) 
for $A(x)=\sech(2x)$. We assume that the following hold true.
\begin{itemize}
\item
We denote by $c=c(\lambda)$, one of the following (simple) turning points: $x_1$, $x_2$, 
$x_6$ or $x_8$ (cf. \S\ref{geometry}) which are (simple) zeros of $g_-$ [see (\ref{gigi}), 
recalling that $V_0=g_-g_+$]. 
\item
The set $\mathscr{D}\subset\C$ (in the $x$-plane) is an open and simply connected neighborhood 
of the real axis so that $c$ is an interior point and the only transition point (of the equation) in it.
\item
The set $\mathscr{O}$ satisfies condition $\mathbf{(H2})$.
\end{itemize}
\end{assumption}
\begin{remark}
Clearly, we can always choose $\mathscr{D}$ and $\mathscr{O}$ so that the assumption 
above holds. 
\end{remark}
\begin{remark}
From the assumptions above, it follows that
\begin{itemize}
\item
the function $(x-c)^{-1}V_0(x,\lambda)$ is holomorphic and non-vanishing throughout 
$\mathscr{D}$ (including $c$).
\item
the function $g(x,\lambda)$ is holomorphic in $\mathscr{D}$. 
\end{itemize}
\end{remark}

Having placed the assumptions on our equation, we continue as follows. We introduce a new 
variable $\zeta=\zeta(\cdot,\lambda):\mathscr{D}\rightarrow\C$ defined by
\be
\label{zeta-olver}
\zeta=
\bigg(
\int_{\gamma(c\rightsquigarrow x)}V_0^{1/2}(t,\lambda)dt
\bigg)
^{2/3}.
\ee
In order to choose the branch, we proceed as follows. We expand 
$(t-c)^{-1}V_0(t,\lambda)$ in a Taylor-series in the neighborhood of $c$, 
namely
\be\nn
(t-c)^{-1}V_0(t,\lambda)=f_{0}+f_{1}(t-c)+f_{2}(t-c)^{2}+\dots
\ee
for a sequence $\{f_n\}_{n\in\N_0}$ of complex numbers (depending on $\lambda$) 
where $f_0\neq0$. Substituting for $V_0$ in (\ref{zeta-olver}) and integrating term 
by term, we find that $\zeta$ has a Taylor-series expansion (in the neighborhood of $c$) 
that begins with
\be
\label{zeta-taylor-expan}
\zeta=
\Big(\frac{2}{3}\Big)^{2/3}f_0^{1/3}(x-c)
\Big[
1+\frac{1}{5}\frac{f_{1}}{f_0}(x-c)+\dots
\Big]
\ee
We select any branch of the coefficient $f_0^{1/3}$ that is convenient to us, and this fixes the 
relation between $\zeta$ and $x$ in the neighborhood of $c$; elsewhere, $\zeta$ is defined by 
continuity.

Next, we define $\Delta=\zeta(\mathscr{D},\lambda)$. \textit{We suppose that $\mathscr{D}$ 
is restricted in such a way that the mapping $\zeta(\cdot,\lambda):\mathscr{D}\rightarrow\Delta$ 
is one-to-one [it is already onto; of course this implies that $c$ is the only point in 
$\mathscr{D}$ at which $\zeta(\cdot,\lambda)$ vanishes]}. Hence, it follows that 
$\zeta(\cdot,\lambda)\in\mathscr{H}(\mathscr{D})$ and therefore the 
$\zeta(\cdot,\lambda)$-mapping from $\mathscr{D}$ on to $\Delta$ is conformal. 
Also, we define the following.
\begin{definition}
Let $x(\cdot,\lambda)$ be the inverse of $\zeta(\cdot,\lambda)$ and $S_j$, $j=0,1,2$ be the 
sectors defined by 
\be\nn
(2j-1)\frac{\pi}{3}\leq\ph\zeta\leq(2j+1)\frac{\pi}{3}.
\ee
The sets 
\be\nn
\mathscr{D}_j=x(\Delta\cap S_j,\lambda),\quad\text{for}\quad j=0,1,2
\ee
shall be called the \textbf{principal regions} associated with the transition point $c$.
\end{definition}

Furthermore, we define 
\be
\label{relation-V0-fhat}
\hat{f}(x,\lambda)\equiv\frac{4}{9}\frac{V_0(x,\lambda)}{\zeta(x,\lambda)}=
\Big[\frac{d\zeta}{dx}(x,\lambda)\Big]^2
\ee
where it is clear that $\hat{f}(\cdot,\lambda)$ is holomorphic and non-vanishing in $\mathscr{D}$.
Subsequently, we introduce the notion of a balancing function. 
\begin{definition}
Consider a function $\Omega(t)$ which is any conveniently chosen positive function 
of the complex variable $t$ that is continuous and satisfies the asymptotics
\be\nn
\Omega(t)=\asympt(t^{1/2})\quad\text{uniformly in the neighborhood of}\quad t=\infty.
\ee
Such a function shall be called a \textbf{balancing function}, and for $j,k=0,1,2$ with 
$j\neq k$ we set 
\be
\label{ro-j-k}
\rho_{j,k}=\sup_{t\in S_j\cup S_k}\{\Omega(t)\msf_{j,k}^2(t)\}
\ee
where the $\msf_{j,k}$ are the auxiliary functions found in appendix \ref{mod_bessel} 
(they are related to the modified Bessel functions).
\end{definition}
\begin{remark}
\label{finiteness-of-ro}
Using the bounds we have for $\msf_{j,k}(t)$, the fact that $\msf_{j,k}(t)$ is 
continuous in $S_j\cup S_k$ and the asymptotics for the balancing function, it is easy 
to realize that the quantities $\rho_{j,k}$ are always finite.
\end{remark}

We fix a way to measure the error of the approximation  by setting 
an error-control function.
\begin{definition}
We define an \textbf{error-control function} $H(x,\lam,\epsilon)$ of equation 
(\ref{eq-for-spec-olver}) to be any primitive of
\be\nn
\frac{1}{\Omega(\epsilon^{-2/3}\zeta)}
\bigg\{
\frac{1}{\hat{f}^{1/4}(x,\lam)}
\frac{d^2}{d x^2}
\Big[\frac{1}{\hat{f}^{1/4}(x,\lam)}\Big]
-\frac{g(x,\lam)}{\hat{f}^{1/2}(x,\lam)}
\bigg\}
\ee
where for $\hat{f}^{1/4}$ and $\hat{f}^{1/2}$ we choose any continuous branches 
in $\mathscr{D}$ such that $\hat{f}^{1/2}=(\hat{f}^{1/4})^2$. Using 
(\ref{relation-V0-fhat}) and $V_0$ instead of $\hat{f}$  the above becomes
\be\nn
\frac{3}{2}
\frac{\zeta^{1/2}}{\Omega(\epsilon^{-2/3}\zeta)}
\bigg\{
\frac{1}{V_0^{1/4}(x,\lam)}
\frac{d^2}{d x^2}
\Big[\frac{1}{V_0^{1/4}(x,\lam)}\Big]
-\frac{g(x,\lam)}{V_0^{1/2}(x,\lam)}
\bigg\}-
\frac{5}{16}
\frac{1}{\zeta^2\Omega(\epsilon^{-2/3}\zeta)}.
\ee
\end{definition}

Finally, before stating the theorem about the solutions of (\ref{eq-for-spec-olver}), let us fix 
some complementary notation.  
\begin{definition}
\label{progr-path-olver-1}
Take two points $x_j\in\clos(\mathscr{D}_j)$ and $y_k\in\clos(\mathscr{D}_k)$ 
where $j,k\in\{0,1,2\}$. A path $\mathscr{P}:\gamma(x_j\rightsquigarrow y_k)$ 
[meaning a Jordan arc comprising of a finite chain of $R_2$ arcs; where an $R_2$ 
arc is an arc $\sigma=\sigma(s)$ ($s$ being the arc parameter) such that 
$\sigma''(s)$ is continuous and $\sigma'(s)$ does not vanish], shall be called a 
\textbf{progressive path} joining $x_j,y_k$ if it satisfies the following ``monotonicity 
condition": as $x$ travels along $\mathscr{P}$  from $x_j$ to $y_k$, the function
\be
\label{progro-mono}
\re\int_{\gamma(x_j\rightsquigarrow x)}V_0^{1/2}(t,\lambda)dt
\ee
is non-decreasing
\footnote{This definition is somewhat different from that of the previous section where
we impose that the same expression is strictly increasing for a path to be progressive.}. 
We choose the branch of the integral in (\ref{progro-mono}) so that the 
whole function is continuous, staying non-positive in $\mathscr{D}_j$ while non-negative in 
$\mathscr{D}_k$. Furthermore, we define 
\be\nn
\mathsf{H}_k(x_j)=
\{\, x\in\mathscr{D}_j\cup\mathscr{D}_k
\mid
\exists\hspace{3pt}\text{a progressive path}\hspace{3pt}
\gamma(x_j\rightsquigarrow x)\subset\mathscr{D}_j\cup\mathscr{D}_k\,\}.
\ee
\end{definition}
\begin{definition}
\label{variation-along-progro}
Let $\mathscr{P}:\gamma(a\rightsquigarrow b)$ be a progressive path between the points 
$a$ and $b$. We define the \textbf{variation} of the error-control function $H$ along 
$\mathscr{P}$ to be
\begin{multline*}
\mathcal{V}_{\mathscr{P}}[H](\lam,\epsilon)
\equiv
\mathcal{V}_{a,b}[H](\lam,\epsilon)\\
=
\int_{\gamma(a\rightsquigarrow b)}
\bigg|
\frac{1}{\Omega(\epsilon^{-2/3}\zeta)}
\bigg\{
\frac{1}{\hat{f}^{1/4}(t,\lam)}
\frac{d^2}{d t^2}
\Big[\frac{1}{\hat{f}^{1/4}(t,\lam)}\Big]
-\frac{g(t,\lam)}{\hat{f}^{1/2}(t,\lam)}
\bigg\}
\bigg|
dt.
\end{multline*}
\end{definition}
\noindent
All these, lead us to the next important theorem about the solutions to equation 
(\ref{eq-for-spec-olver}).
\begin{theorem}
Consider Assumption \ref{assume-one-tp}, 
let $a,b\in\C$ be arbitrary and choose a reference point $a_j\in\clos(\mathscr{D}_j)$, $j=0,1,2$ 
(the reference point could be $\infty$). Also, for $k\in\{0,1,2\}$ such that $k\neq j$, consider the 
set $\mathsf{H}_k(a_j)$ of points in $\mathscr{D}_j\cup\mathscr{D}_k$ that can be joined to 
$a_j$ with a progressive path $\mathscr{P}$ in $\mathscr{D}_j\cup\mathscr{D}_k$ [with the 
understanding that if $\zeta(a_j)$ is at infinity, then we suppose that it is the point at infinity on a 
path $P$ in $S_j$ and require $\zeta(\mathscr{P})$ to coincide with $P$ in a neighborhood of 
$\zeta(a_j)$].  Then for every $\epsilon>0$, equation (\ref{eq-for-spec-olver}), namely
\be\nn
\frac{d^2w}{dx^2}=[\epsilon^{-2}V_0(x,\lambda)+g(x,\lambda)]w
\ee
with $V_0,g$ given by (\ref{olver-f}) and (\ref{olver-g}) respectively, has a solution 
$w(x,\lambda,\epsilon)$ (depending of course on $a,b$ and $a_j$) that is holomorphic in 
$\mathscr{D}\setminus\{c\}$ and continuous at $c$. Furthermore, when $x\in\mathsf{H}_k(a_j)$ 
we have
\be\nn
w(x,\lambda,\epsilon)=
\hat{f}^{-1/4}(x,\lambda)
\big[
aU_j(\epsilon^{-2/3}\zeta)+bU_k(\epsilon^{-2/3}\zeta)+R(\zeta,\lambda,\epsilon)
\big]
\ee
where $U_j$, $U_k$ are the MBF defined in appendix \ref{mod_bessel}. 
The remainder $R$ satisfies the following precise estimates
\begin{multline}
\label{error-bounds-olver}
\frac{|R(\zeta,\lambda,\epsilon)|}{\msf_{j,k}(\epsilon^{-2/3}\zeta)},
\frac{\Big|\frac{\partial R}{\partial\z}(\z,\lambda,\epsilon)\Big|}
{\epsilon^{-2/3}\nsf_{j,k}(\epsilon^{-2/3}\zeta)},
\frac{\Big|\frac{\partial[\zeta^{1/4}R(\z,\lambda,\epsilon)]}{\partial\z}\Big|}
{\epsilon^{-2/3}|\zeta|^{1/4}\hat{\nsf}_{j,k}(\epsilon^{-2/3}\zeta)}
\\
\leq
\frac{\sigma_{j,k}}{\ro_{j,k}}\esf_{j,k}(\epsilon^{-2/3}\zeta)
\bigg(
\exp
\Big\{
\frac{\ro_{j,k}}{3|\lambda_{j,k}|}\epsilon^{2/3}\mathcal{V}_{a_j,x}[H](x,\lambda,\epsilon)
\Big\}
-1
\bigg)
\end{multline}
where the variation of $H$ is being evaluated along $\mathscr{P}$ 
(see Definition \ref{variation-along-progro}), the $\ro_{j,k}$ are given 
by (\ref{ro-j-k}) and the quantities $\sigma_{j,k}, \lambda_{j,k}$ satisfy 
\be
\label{sigma-j-k}
\sigma_{j,k}=
\sup_{p\in\zeta(\mathscr{P})}
\bigg\{
\frac{\Om(\epsilon^{-2/3}p)\msf_{j,k}(\epsilon^{-2/3}p)}{\esf_{j,k}(\epsilon^{-2/3}p)}
\big|aU_j(\epsilon^{-2/3}p)+bU_k(\epsilon^{-2/3}p)\big|
\bigg\}
\ee
and
\be
\label{lambda-j-k}
\lambda_{j,k}=\frac{\sin\Big[(k-j)\frac{\pi}{3}\Big]}{\sin\frac{\pi}{3}}
\ee
respectively. When $x\in\mathscr{D}_k$, the bounds (\ref{error-bounds-olver}) for the error, 
apply only to the branch of $w$ obtained by analytic continuation from a neighborhood of $c$ 
in $\mathscr{D}_j$ by rotation through an angle $2(k-j)\frac{\pi}{3}$.
\end{theorem}
\begin{proof}
For the proof see \S\S 3.3-3.6 in Olver's \cite{olver1978} with $m=3$. The statement of 
the theorem can be found in \S 3.3 as Theorem 3.1.
\end{proof}
\begin{remark}
For the theorem to be meaningful, it is necessary that the right-hand side of 
(\ref{error-bounds-olver}) is finite. Recall that $\ro_{j,k}$ are finite (cf. Remark \ref{finiteness-of-ro}) 
and that in our case $1/|\lambda_{j,k}|$ are finite as well [see (\ref{lambda-j-k}) and recall that 
$(k-j)\neq0\3$]. But 
$\mathcal{V}_{\mathscr{P}}[H]$ has to converge (the variation of $H$ is calculated along 
$\mathscr{P}$) and also the $\sigma_{j,k}$ have to be finite. It is seen that the quantity  in brackets 
in (\ref{sigma-j-k}) is finite at all finite points $p\in S_j\cup S_k$ and bounded as 
$p\rightarrow\infty$ in $S_k$. As $p\rightarrow\infty$ in an internal part of $S_j$ 
(see appendix \ref{mod_bessel}) however, the quantity  in brackets is unbounded, unless 
$b=0$. Accordingly, if $\zeta(a_j)=\infty$ in an internal part of $S_j$, then the theorem may 
be applied only in the case $b=0$.
\end{remark}
\begin{remark}
In order to deal with the behavior near the remaining (simple) turning points, namely 
$x_3$, $x_4$, $x_5$ and $x_7$, instead of the differential equation for $w_-$, we consider 
the differential equation for $w_+$ [see (\ref{w-plus-minus})] and apply exactly the same 
strategy like the one just presented.
\end{remark}

\subsection{Asymptotic estimates for the error terms}
\label{asymptt-estim-for-error}

Let us now consider a circumstance in which the error-term $R$ is sufficiently small for the 
theorem to supply meaningful approximations when $\epsilon$ is small. But let us first define 
the function
\be
\label{F-func}
F(x,\lambda)=
\int
\bigg\{
\frac{1}{V_0^{1/4}(x,\lam)}
\frac{d^2}{d x^2}
\Big[\frac{1}{V_0^{1/4}(x,\lam)}\Big]
-\frac{g(x,\lam)}{V_0^{1/2}(x,\lam)}
\bigg\}
dx
\ee
where for this integral we shall not use paths that intersect $c$, while for 
$V_0^{1/4}, V_0^{1/2}$ we adopt any branches, provided they are continuous 
and the latter is the square of the former. With this in mind, we have the following result. 
\begin{theorem}
Let Assumption \ref{assume-one-tp} be satisfied (observe that $V_0,g$ are independent of 
$\epsilon$), consider $a_j,a_k$ to be arbitrary points in $\clos(\mathscr{D}_j)$ and 
$\clos(\mathscr{D}_k)$ respectively (including the point at infinity) and let $\mathscr{P}$ 
be a progressive path joining $a_j$ and $a_k$. Then as $\epsilon\downarrow0$, we have 
\be\nn
\mathcal{V}_{\mathscr{P}}[H](\lambda,\epsilon)=\asympt(\epsilon^{1/3}).
\ee
\end{theorem}
\begin{proof}
Referring to Theorem 4.1(ii) when $m=3$ in \cite{olver1978}, one only needs to check that 
\begin{itemize}
\item
the function $1/\zeta(x)$ [refer to (\ref{zeta-olver})] and 
\item
the function $F(x,\lambda)$ in (\ref{F-func})
\end{itemize}
are of bounded variation as $x\rightarrow a_j$ or $a_k$ along $\mathscr{P}$.
\end{proof}

\subsection{The connection theorem}
\label{olver-connection}

In this section, we will state a theorem that concerns the connection of solutions of 
our equation between two principal regions near a turning point. But first, we reformulate 
the geometric concepts introduced in \S\ref{sols-1-turn-pt} using a new variable $\xi_c$, 
in order to avoid referring explicitly to the variable $\zeta$. We begin with a definition 
(see Definition \ref{stokes-curves-gg}).

\begin{definition}
The curves in the complex $x$-plane having the equation
\be\nn
\re\int_{\gamma(c\rightsquigarrow x)}V_0^{1/2}(t,\lambda)dt=0
\ee
are called the \textbf{Stokes curves} associated with the turning point $c$
\footnote{
Olver actually calls them \textbf{anti-Stokes curves} or \textbf{principal curves}.}. 
Either branch may be used for the square root of $V_0$, provided continuity is maintained. 
\end{definition}

Since $c$ is a simple turning point of equation (\ref{eq-for-spec-olver}), there are three distinct 
Stokes curves intersecting at $c$ at an angle of $2\pi/3$.  All 
of them occupy the same Riemann sheet. Using the conformality of the mapping one can prove that: 
(i) a Stokes curve can either terminate at $c$ or at a boundary point of $\mathscr{D}$, and (ii) 
that no Stokes curve intersects itself or any other Stokes curve on the same Riemann sheet 
(except c).

Evidently, the Stokes  curves are boundaries of  principal regions (each principal region 
includes its bounding Stokes curves). From here on, \textit{we suppose that in each principal 
region there is a one-to-one relation between $x$ and any continuous branch of the integral 
$\int V_0^{1/2}(t,\lambda)dt$} (which is equivalent to the fact that the mapping from 
$\mathscr{D}$ to $\Delta$ is one-to-one).

We start by fixing an arbitrarily chosen principal region and name it $\mathscr{D}_0$. 
We then define $\mathscr{D}_1$ to be the successive principal domain encountered as we 
pass around $c$ in a counterclockwise direction and similarly set $\mathscr{D}_{2}$ be the 
successive principal region (to $\mathscr{D}_0$) encountered as we pass around $c$ in 
a clockwise direction.  Whatever choice we make for $\mathscr{D}_0$, the labeling 
can be made consistent with that chosen before, by appropriately choosing the branch 
of $f_0^{1/3}$ in (\ref{zeta-taylor-expan}).

\begin{definition}
Take $j\in\{0,1,2\}$. We shall refer to 
$\mathscr{D}_j\cap\mathscr{D}_{(j+1)\3}$ as the 
\textbf{left boundary} of $\mathscr{D}_j$. Also we will call 
$\mathscr{D}_j\cap\mathscr{D}_{(j-1)\3}$, 
the \textbf{right boundary} of $\mathscr{D}_j$. 
\end{definition}

Next, in each principal region $\mathscr{D}_j$, $j\in\{0,1,2\}$ associated with $c$ 
we define 
\be\label{xi-c-func}
\xi_c(x,\lambda)=\int_{\gamma(c\rightsquigarrow x)}V_0^{1/2}(t,\lambda)dt
\ee
taking the branch that is continuous with $\re\xi_c(x,\lambda)\geq0$. This determines a 
function of $x$ that is continuous in $\mathscr{D}$ except on the Stokes  curves. 
Fix a $j\in\{0,1,2\}$. Clearly, for $x$ on the left boundary of $\mathscr{D}_j$ we have 
$\xi_c(x,\lambda)\in i\R^+$, while on the right boundary of $\mathscr{D}_j$ we have 
$\xi_c(x,\lambda)\in i\R^-$. Since the left boundary of $\mathscr{D}_j$ is also the 
right boundary of $\mathscr{D}_{(j+1)\3}$ and the right boundary of 
$\mathscr{D}_j$ is also the left boundary of $\mathscr{D}_{(j-1)\3}$, 
we realize that $\xi_c(x,\lambda)$ changes sign when it crosses a principal curve.
Finally we redefine the notion of the progressive path (cf. Definition \ref{progr-path-olver-1}).
\begin{definition}
Any Jordan arc $\mathscr{P}$ comprising of a finite chain of $R_2$ arcs and having the 
property that $\re\xi_c(\cdot,\lambda)$ is monotonic on the intersection of $\mathscr{P}$ 
with any principal region, shall (also) be called a \textbf{progressive path} 
associated with $c$.
\end{definition}

Before stating the main theorem of this section, let us formulate an important lemma first.
\begin{lemma}
\label{unique-solution-lemma}
Consider the differential equation (\ref{eq-for-spec-olver}), taking into account 
Assumption \ref{assume-one-tp}. Also, let $j\in\{0,1,2\}$, pick a boundary point or 
point at infinity $a_j\in\clos(\mathscr{D}_j)$ and consider a progressive path 
$\mathscr{P}$ in $\mathscr{D}_j$ passing from $a_j$. 
Then equation (\ref{eq-for-spec-olver}) has a unique solution 
$w(x,\lambda,\epsilon)\in\mathscr{H}(\mathscr{D})$ and satisfies
\begin{align*}
V_0^{1/4}(x,\lambda)w(x,\lambda,\epsilon)
&\sim
e^{-\xi_c(x,\lambda)/\epsilon}
\quad\text{as}\quad x\rightarrow a_j\quad\text{along}\quad\mathscr{P}\\
\frac{\partial[V_0^{1/4}(x,\lambda)w(x,\lambda,\epsilon)]}{\partial x}
&\sim
-\frac{1}{\epsilon}
V_0^{1/2}(x,\lambda)e^{-\xi_c(x,\lambda)/\epsilon}
\quad\text{as}\quad x\rightarrow a_j\quad\text{along}\quad\mathscr{P}.
\end{align*}
\end{lemma}
\begin{proof}
Using Theorem 11.1 in chapter 6 of Olver's book \cite{olver1997} one needs to realize that 
\begin{itemize}
\item
$\re\xi_c(x,\lambda)\rightarrow+\infty$ as $x\rightarrow a_j$ along $\mathscr{P}$ and
\item
$\mathcal{V}_{a_j,x}[F]$ converges as $x\rightarrow a_j$ along $\mathscr{P}$.
\end{itemize}
It follows that such a solution exists, is holomorphic and additionally satisfies the desired 
asymptotics. Since $\re\xi_c(x,\lambda)\rightarrow+\infty$ as $x\rightarrow a_j$ along 
$\mathscr{P}$, this solution is recessive and hence unique.
\end{proof}

We are now ready for current section's main theorem.
\begin{theorem}
\label{connection-of-sols}
Consider the differential equation (\ref{eq-for-spec-olver}), taking into account 
Assumption \ref{assume-one-tp}. Also, let $j, k\in\{0,1,2\}$ with $k\neq j$, pick two 
boundary points or points at infinity, namely $a_j\in\clos(\mathscr{D}_j)$, 
$a_k\in\clos(\mathscr{D}_k)$ (different from $c$) and consider a progressive path 
$\mathscr{P}$ in $\mathscr{D}_j\cup\mathscr{D}_k$ joining them. Then the unique solution 
$w(x,\lambda,\epsilon)$ of equation (\ref{eq-for-spec-olver}) provided by the previous lemma, 
satifies the following asymptotics as $x\rightarrow a_k$ along $\mathscr{P}$
\begin{align*}
V_0^{1/4}(x,\lambda)w(x,\lambda,\epsilon)
&\sim
i^{k-j-1}(\lambda_{j,k}+\kappa_0)
e^{\xi_c(x,\lambda)/\epsilon}\\
\frac{\partial[V_0^{1/4}(x,\lambda)w(x,\lambda,\epsilon)]}{\partial x}
&\sim
\frac{i^{k-j-1}}{\epsilon}(\lambda_{j,k}+\kappa_0)
V_0^{1/2}(x,\lambda)e^{\xi_c(x,\lambda)/\epsilon}
\end{align*}
where $\kappa_0$ is independent of $x$ and subject to the bound
\be\nn
|\kappa_0|\leq(1+\lambda_{j,k}^2)^{1/2}
\bigg(
\exp
\Big\{
\frac{\ro_{j,k}}{3|\lambda_{j,k}|}\epsilon^{2/3}\mathcal{V}_{\mathscr{P}}[H]
\Big\}
-1
\bigg).
\ee
\end{theorem}
\begin{proof}
The proof of this theorem follows from Theorem 5.1 in \cite{olver1978} merely by checking 
that
\begin{itemize}
\item
$\re\xi_c(x,\lambda)\rightarrow+\infty$ as $x\rightarrow a_j$ along $\mathscr{P}$,
\item
$\mathcal{V}_{a_j,x}[F]$ converges as $x\rightarrow a_j$ along $\mathscr{P}$ and 
\item
$\re\xi_c(x,\lambda)\rightarrow+\infty$ as $x\rightarrow a_k$ along $\mathscr{P}$.
\end{itemize}
\end{proof}

\subsection{The asymptotic form of the connection theorem}
\label{connection-asymptotics}

In this section, we  construct asymptotic estimates for the error-terms in 
Theorem \ref{connection-of-sols} for small $\epsilon>0$ by applying the results from 
\S\ref{asymptt-estim-for-error}. Let us first fix some notation that will be used in this and 
the following sections. We have
\begin{itemize}
\item
We use the notation ``$\doteq$" to state that besides the validity of the 
equation shown, the corresponding equation obtained by formally differentiating with 
respect to $x$, ignoring the differentiation of the $\asympt$-terms present, also stands 
true.
\item
When an $\asympt$-term appears in an equation it is understood to hold uniformly for 
all the values of $x$ associated with that equation.
\item
With $\chi(\epsilon)$,  we denote any positive function of $\epsilon$ such that 
as $\epsilon\downarrow0$ it satisfies 
\be
\label{chi-epsilon}
\chi(\epsilon)\rightarrow0\quad\text{and}\quad 
1/\chi(\epsilon)=\asympt(\epsilon^{-1}).
\ee
\item
Lastly, let $a, b$ be two points on a path $\mathscr{P}$. We write $[a,b]_{\mathscr{P}}$, 
$(a,b)_{\mathscr{P}}$, $[a,b)_{\mathscr{P}}$ and $(a,b]_{\mathscr{P}}$ to denote the 
part of $\mathscr{P}$ that lies between $a$ and $b$, with endpoints included 
or not, accordingly.
\end{itemize}
With these in mind, we formulate the following theorem.
\begin{theorem}
\label{connection-theorem-of-olver}
For $\epsilon>0$, consider the differential equation 
\be
\label{diff-eq-olver-spec}
\frac{d^2w}{dx^2}=[\epsilon^{-2}V_0(x,\lambda)+g(x,\lambda)]w,
\quad
(x,\lambda)\in\mathscr{D}\times\mathscr{O}
\ee
and  assume Assumption \ref{assume-one-tp}. Also, let $j, k\in\{0,1,2\}$ with $k\neq j$ 
and consider a progressive path $\mathscr{P}$ in $\mathscr{D}_j\cup\mathscr{D}_k$.
Pick points $a_j, b_j, b_k, a_k$, in that order on $\mathscr{P}$, neither of which depends 
on $\epsilon$ nor coincides with $c$, and such that $a_j\in\clos(\mathscr{D}_j)$, 
$a_k\in\clos(\mathscr{D}_k)$, $b_j\in\mathscr{D}_j$ and $b_k\in\mathscr{D}_k$ 
[thus $a_j, a_k$ (but not $b_j, b_k$) may be boundary points of $\mathscr{D}$, including 
the point at infinity]. Finally, let the function $w(x,\lambda,\epsilon)$ denote the solution of 
the differential equation (\ref{diff-eq-olver-spec}) (as provided by 
Lemma \ref{unique-solution-lemma}).Then, on $[b_k,a_k)_{\mathscr{P}}$, the 
analytic continuation of $w(x,\lambda,\epsilon)$ [obtained by passing from $\mathscr{D}_j$ 
to $\mathscr{D}_k$ in the same sense as the sign of $(k-j)$], is given by
\begin{align*}
\text{case A}:\quad
w(x,\lambda,\epsilon)
&=
w_{\rm I}(x,\lambda,\epsilon)
\quad\text{or}
\\
\text{case B}:\quad
w(x,\lambda,\epsilon)
&=
w_{\rm I,1}(x,\lambda,\epsilon)+w_{\rm L}(x,\lambda,\epsilon)
\quad\text{or}
\\
\text{case C}:\quad
w(x,\lambda,\epsilon)
&=
w_{\rm I,-1}(x,\lambda,\epsilon)+w_{\rm R}(x,\lambda,\epsilon)
\end{align*}
depending on whether $b_k$ is an interior point of $\mathscr{D}_k$ (\text{case A}),  or 
$[b_k,a_k)_{\mathscr{P}}$ lies on the left boundary of $\mathscr{D}_k$ (\text{case B}), or 
 $[b_k,a_k)_{\mathscr{P}}$ lies on the right boundary of $\mathscr{D}_k$ (\text{case C}). 
Here, $w_{\rm I}(x,\lambda,\epsilon)$, $w_{\rm I,\pm1}(x,\lambda,\epsilon)$, 
$w_{\rm L}(x,\lambda,\epsilon)$ and $w_{\rm R}(x,\lambda,\epsilon)$ are solutions of the 
differential equation (\ref{diff-eq-olver-spec}) so that on $[b_k,a_k)_{\mathscr{P}}$ they 
have the following asymptotic forms as $\epsilon\downarrow0$
\begin{equation*}
\begin{rcases}
        V_0^{1/4}(x,\lambda)w_{\rm I}(x,\lambda,\epsilon)\\
        V_0^{1/4}(x,\lambda)w_{\rm I,1}(x,\lambda,\epsilon)\\
        V_0^{1/4}(x,\lambda)w_{\rm I,-1}(x,\lambda,\epsilon)
\end{rcases}
\doteq
i^{k-j-1}
[\lambda_{j,k}+\asympt(\hat{\chi})]
e^{\xi_c(x,\lambda)/\epsilon}
\end{equation*}
\begin{equation*}
V_0^{1/4}(x,\lambda)w_{\rm L}(x,\lambda,\epsilon)
\doteq
i^{k-j}
[\lambda_{j,k+1}+\asympt(\hat{\chi})]
e^{-\xi_c(x,\lambda)/\epsilon}
\end{equation*}
\begin{equation*}
V_0^{1/4}(x,\lambda)w_{\rm R}(x,\lambda,\epsilon)
\doteq
-i^{k-j}
[\lambda_{j,k-1}+\asympt(\hat{\chi})]
e^{-\xi_c(x,\lambda)/\epsilon}.
\end{equation*} 
In these relations, $\lambda_{j,k}$ are given by (\ref{lambda-j-k}), $V_0^{1/4}(x,\lambda)$ denotes 
the branch obtained from that used in Lemma \ref{unique-solution-lemma} [equivalently 
see (\ref{sol-asympt-at-initial-pt}) below] by analytic continuation 
in the same manner as for $w(x,\lambda,\epsilon)$ and the function 
$\hat{\chi}=\hat{\chi}(\epsilon)$ is defined by
\be\label{chi-hat-epsilon}
\hat{\chi}(\epsilon)=\max\{\epsilon, \chi(\epsilon)\}.
\ee
\end{theorem}
\begin{proof}
For the proof, the reader can refer to Theorem 6.1 in \cite{olver1978} with $m=3$ and 
the function $g$ analytic at $c$. It suffices to see that the following are true. 
\begin{itemize}
\item[(i)]
The function $1/\xi_c(x,\lambda)$ [see (\ref{xi-c-func})]  is of bounded variation on 
$(a_j,b_j]_{\mathscr{P}}$ and $[b_k,a_k)_{\mathscr{P}}$
\item[(ii)]
The function $F(x,\lambda)$ given by (\ref{F-func}) is of bounded variation on 
$(a_j,b_j]_{\mathscr{P}}$ and $[b_k,a_k)_{\mathscr{P}}$ and
\item[(iii)]
The function $w(x,\lambda,\epsilon)$ satisfies the following.
\begin{itemize}
\item
For $x\in(a_j,b_j]_{\mathscr{P}}$
\be\label{sol-asympt-at-initial-pt}
V_0^{1/4}(x,\lambda)w(x,\lambda,\epsilon)\doteq
[1+\asympt(\chi)]
e^{-\xi_c(x,\lambda)/\epsilon}
\quad\text{as}\quad\epsilon\downarrow0
\ee
\item
As $x\rightarrow a_j$ along $\mathscr{P}$, the functions
\begin{align*}
e^{\xi_c(x,\lambda)/\epsilon}&
V_0^{1/4}(x,\lambda)
w(x,\lambda,\epsilon), \\
V_0^{-1/2}(x,\lambda)
e^{\xi_c(x,\lambda)/\epsilon}&
\frac{\partial[V_0^{1/4}(x,\lambda)w(x,\lambda,\epsilon)]}{\partial x}
\end{align*}
tend to non-zero finite limits.
\end{itemize}
\end{itemize}
\end{proof}
\begin{remark}
Even though $w_{\rm I}(x,\lambda,\epsilon)$ and $w_{\rm I,\pm1}(x,\lambda,\epsilon)$ have 
the same asymptotics as $\epsilon\downarrow0$ on $[b_k,a_k)_{\mathscr{P}}$, it should be 
emphasized that they do represent distinct solutions of equation (\ref{diff-eq-olver-spec}). 
\end{remark}

\subsection{Application of the connection theorem}
\label{application-of-connection}

In this section, we apply Theorem \ref{connection-theorem-of-olver} to prove
a Bohr-Sommerfeld type quantization condition for the eigenvalues that lie 
near zero. Once again, we begin by some geometric formulations that will pave the way 
to the final result.

Let now $\mathscr{D}$ denote a simply connected open neighborhood of the real axis that 
contains only two simple turning points of our differential equation. Let us be more precise 
and make the following assumption.
\begin{assumption}
\label{assume-two-tp}
Consider the differential equation
\be\nn
\frac{d^2w}{dx^2}=[\epsilon^{-2}V_0(x,\lambda)+g(x,\lambda)]w,
\quad
(x,\lambda)\in\mathscr{D}\times\mathscr{O}
\ee
where 
\begin{itemize}
\item
The functions $V_0(x,\lambda)$ and $g(x,\lambda)$ are given by
\be\nn
V_0(x,\lambda)=
-\Big\{A(x)^2+[\lambda+\tfrac{A'(x)}{2}]^2\Big\}
\ee
\be\nn
g(x,\lambda)=
\frac{3}{4}\bigg\{ \frac{A'(x)-i\frac{A''(x)}{2}}{A(x)-i[\lam+\tfrac{A'(x)}{2}]} \bigg\}^2
-
\frac{1}{2}\frac{A''(x)-i\frac{A'''(x)}{2}}{A(x)-i[\lam+\tfrac{A'(x)}{2}]}
\ee
for $A(x)=\sech(2x)$.
\item
The set $\mathscr{O}$ is an open and simply connected set 
(of the $\lambda$-plane), proper subset of ${\mathcal R}_0\cap\mathbb{H}^+$ 
[see (\ref{r0}) with $S(x)=A(x)=\sech(2x)$]  located  in such a way that its intersection 
with $\tilde{\Lambda}_{12}$ is a single interval and moreover it does not have common 
points with either $\tilde{\Lambda}_{16}$ or $\tilde{\Lambda}_{26}$. Also, let $c\equiv x_1$ 
and $\hat{c}\equiv x_2$ be the two distinct simple zeros of $g_-$ as introduced in 
\S\ref{geometry} [see Remark \ref{which-is-which}; recall that $V_0=g_-g_+$ where 
$g_\pm$ are given by (\ref{gigi})].
\item
The set $\mathscr{D}\subset\C$ is an open 
and simply connected domain (of the $x$-plane) so that it contains the real axis and 
that the only transition points of the differential equation in it, are $c$ and $\hat{c}$. 
\end{itemize}
\end{assumption}
\begin{remark}
Observe that the assumption above implies that
\begin{itemize}
\item[(i)]
the function $(x-c)^{-1}(x-\hat{c})^{-1}V_0(x,\lambda)$ is holomorphic 
and non-vanishing throughout $\mathscr{D}$ (including $c, \hat{c}$) and
\item[(ii)]
the function $g(x,\lambda)$ is holomorphic in $\mathscr{D}$. 
\end{itemize}
\end{remark}

Let the WKB approximation of a solution of our differential equation be given 
at a point in $\clos(\mathscr{D})$ other than one of the turning points. Then the WKB 
approximation of the same solution at any other point (except at a turning point) can 
be found by at most two applications of Theorem \ref{connection-theorem-of-olver}.

As was done previously, we associate with the two transition points the functions 
\begin{align*}
\xi_c(x,\lambda)
&=
\int_{\gamma(c\rightsquigarrow x)}V_0^{1/2}(t,\lambda)dt\\
\xi_{\hat{c}}(x,\lambda)
&=
\int_{\sigma(\hat{c}\rightsquigarrow x)}V_0^{1/2}(t,\lambda)dt.
\end{align*}
For the former, namely $\xi_c(x,\lambda)$, we take a branch of the integral so that 
$\re\xi_c(x,\lambda)\geq0$ and so that it makes $\xi_c(x,\lambda)$ continuous 
except at $\re\xi_c(x,\lambda)=0$. The last relation defines a set of three Stokes 
curves emanating from $c$. We refer to any Stokes curve emanating from the turning 
point $c$ as a \textit{$c$-Stokes curve}
\footnote{
Olver calls them \textit{$c$-principal curves}.}. 
We argue similarly for $\xi_{\hat{c}}(x,\lambda)$.

On the $\xi$-plane,
where  
\be\nn
\xi(x,\lambda)=
\int V_0^{1/2}(x,\lambda)dx
\ee
in which the integration constant is arbitrary and the branch is chosen in a continuous 
way, we have the following. The $c$-Stokes curves and the $\hat{c}$-Stokes curves are 
mapped to straight lines parallel to the imaginary axis. Conformal mapping theory shows 
that the $c$-Stokes curves intersect on the same Riemann sheet only at $c$ and a similar 
assertion holds for $\hat{c}$. Also, $c$-Stokes curves and $\hat{c}$-Stokes curves do not 
intersect on the same sheet. However, in our case $c$ and $\hat{c}$ are linked together 
by a common Stokes curve: the interval $[\xi(c), \xi(\hat{c})]$ (on the $\xi$-plane) is 
parallel to the imaginary axis.

\begin{figure}
\centering
\includegraphics[bb=-100 0 710 240, width=14cm]{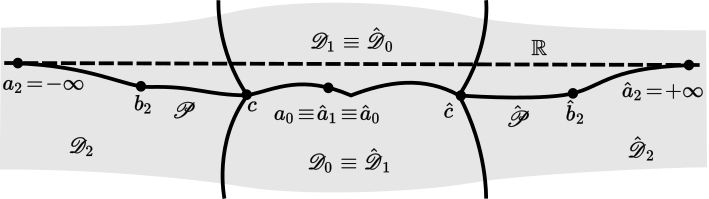}
\caption{The configuration on the $x$-plane for the case of near-zero eigenvalues.}
\label{stokes4}
\end{figure}

Now, we assemble all the results of this paragraph and formulate the theorem 
about the eigenvalues that lie near zero.
\begin{theorem}
\label{theorem-evs-near-zero-ala-olver}
Consider Assumption \ref{assume-two-tp}. Then for any $\lambda_0\in\mathscr{O}$, there 
exists a complex neighborhood $U$ of $\lambda_0$ such that $\lambda\in U$ is an 
eigenvalue of $\mathfrak{D}_{\epsilon}$ if and only if
\be\nn
n(\lambda,\epsilon)e^{2\xi_c(\hat{c})/\epsilon}=1
\ee
where as $\epsilon\downarrow0$, the function $n(\lambda,\epsilon)$ satisfies 
\be\nn
n(\lambda,\epsilon)=-1+\asympt(\epsilon).
\ee
\end{theorem}
\begin{proof}
Figure \ref{stokes4} shows the geometric configuration in the complex $x$-plane for our case, 
along with the labeling of the corresponding principal regions. In particular, there are 
two regions, namely $\mathscr{D}_0=\hat{\mathscr{D}}_1$ and 
$\mathscr{D}_1=\hat{\mathscr{D}}_0$, that have both $c$ and $\hat{c}$ on their 
boundaries, each being a principal region with respect to either turning point. 

We would like to connect the WKB form of a solution on the segment 
$(-\infty,b_2]_{\mathscr{P}}$ of a progressive path $\mathscr{P}\in\mathscr{D}_2$, 
to the the WKB form of the same solution on the segment 
$[\hat{b}_2,+\infty)_{\hat{\mathscr{P}}}$ of a progressive path 
$\hat{\mathscr{P}}\in\hat{\mathscr{D}}_2$. The essential observation for one to 
notice is the relation between $\xi_c$ and $\xi_{\hat{c}}$ in the intermediate area. 
We have
\be\nn
\xi_c(x)-\xi_{\hat{c}}(x)=\xi_c(\hat{c})=-\xi_{\hat{c}}(c),
\quad
x\in\mathscr{D}_0\cup\mathscr{D}_1.
\ee

We take an arbitrary point $a_0$ on the common Stokes curve linking $c$ and 
$\hat{c}$. Then, $a_0$ can be joined to $b_2$ by an extension of $\mathscr{P}$ 
that passes through $c$ and is progressive; now $[c,a_0]_{\mathscr{P}}$ coincides 
with the common Stokes curve. By applying Theorem \ref{connection-theorem-of-olver} 
(with $k=0$ and $j=2$) we find that the WKB form of $w(x,\lambda,\epsilon)$ at $a_0$ 
is given by
\be\nn
w(x,\lambda,\epsilon)=
w_{\rm I,1}^{(c)}(x,\lambda,\epsilon)+
w_{\rm L}^{(c)}(x,\lambda,\epsilon)
\ee 
where as $\epsilon\downarrow0$ we have the asymptotics
\begin{align*}
V_0^{1/4}(x,\lambda)w_{\rm I,1}^{(c)}(x,\lambda,\epsilon)
&\doteq
i^{-3}
[\lambda_{2,0}+\asympt(\hat{\chi})]
e^{\xi_c(x,\lambda)/\epsilon}
\\
V_0^{1/4}(x,\lambda)w_{\rm L}^{(c)}(x,\lambda,\epsilon)
&\doteq
i^{-2}
[\lambda_{2,1}+\asympt(\hat{\chi})]
e^{-\xi_c(x,\lambda)/\epsilon}.
\end{align*}
where $\hat{\chi}(\epsilon)$ is given by (\ref{chi-hat-epsilon}) and (\ref{chi-epsilon}).

To prepare for passage through $\hat{c}$, we observe that 
$\mathscr{D}_0=\hat{\mathscr{D}}_1$. Consequently, we relabel $a_0$ now naming 
it $\hat{a}_1$ and extend $\hat{\mathscr{P}}$, still progressive, to pass through $\hat{c}$ 
and continue along the common Stokes  curve until $\hat{a}_1$ is reached. In the WKB 
form of $w(x,\lambda,\epsilon)$ at $\hat{a}_1$ found by passage through $c$, there 
are two terms: $w_{\rm I,1}^{(c)}(x,\lambda,\epsilon)$ and 
$w_{\rm L}^{(c)}(x,\lambda,\epsilon)$. Let us find how these two behave when they 
pass through $\hat{c}$. 

The contribution from $w_{\rm L}^{(c)}(x,\lambda,\epsilon)$ to the WKB form on 
$[\hat{b}_2,+\infty)_{\hat{\mathscr{P}}}$ is obtained by replacing 
$e^{-\xi_c(x,\lambda)/\epsilon}$ by 
$e^{-\xi_c(\hat{c},\lambda)/\epsilon}e^{-\xi_{\hat{c}}(x,\lambda)/\epsilon}$ 
and applying Theorem \ref{connection-theorem-of-olver} (with $c$ replaced by 
$\hat{c}$ and $k=2, j=1$). We find
\be\nn
w_{\rm L}^{(c)}(x,\lambda,\epsilon)=w_{\rm I}^{(\hat{c})}(x,\lambda,\epsilon)
\ee 
where as $\epsilon\downarrow0$ we have
\be\nn
V_0^{1/4}(x,\lambda)w_{\rm I}^{(\hat{c})}(x,\lambda,\epsilon)
\doteq
i^{0}
[\lambda_{1,2}+\asympt(\hat{\chi})]
e^{-\xi_c(\hat{c},\lambda)/\epsilon}e^{\xi_{\hat{c}}(x,\lambda)/\epsilon}.
\ee

To handle the contribution from $w_{\rm I,1}^{(c)}(x,\lambda,\epsilon)$, the key step is to 
regard $a_0$ as a member of $\hat{\mathscr{D}}_0$ and relabel it as $\hat{a}_0$. Since 
this entails crossing a Stokes curve, $\xi_{\hat{c}}(x,\lambda)$ is replaced by 
$-\xi_{\hat{c}}(x,\lambda)$. Thus, $e^{\xi_c(x,\lambda)/\epsilon}$ becomes 
$e^{\xi_c(\hat{c},\lambda)/\epsilon}e^{-\xi_{\hat{c}}(x,\lambda)/\epsilon}$, 
where $\hat{c}$ is regarded as a member of $\mathscr{D}_0$ in calculating 
$\xi_c(\hat{c},\lambda)\in i\R^+$ and $x\in\hat{\mathscr{D}}_0$. Now, by an application 
of Theorem \ref{connection-theorem-of-olver} (with $c$ replaced by $\hat{c}$ and 
$k=2, j=0$), we find
\be\nn
w_{\rm I,1}^{(c)}(x,\lambda,\epsilon)=w_{\rm I}^{(\hat{c})}(x,\lambda,\epsilon)
\ee 
where as $\epsilon\downarrow0$, the following asymptotic form is satisfied
\be\nn
V_0^{1/4}(x,\lambda)w_{\rm I}^{(\hat{c})}(x,\lambda,\epsilon)
\doteq
-i^2
[\lambda_{0,2}+\asympt(\hat{\chi})]
e^{\xi_c(\hat{c},\lambda)/\epsilon}e^{\xi_{\hat{c}}(x,\lambda)/\epsilon}.
\ee

Now, combining the last two results, we are capable of obtaining the behavior of 
$w(x\lambda,\epsilon)$ on $[\hat{b}_2,+\infty)_{\hat{\mathscr{P}}}$. As 
$\epsilon\downarrow0$ we find that for $x\in[\hat{b}_2,+\infty)_{\hat{\mathscr{P}}}$ 
we have
\be\nn
V_0^{1/4}(x,\lambda)w(x,\lambda,\epsilon)
\doteq
\Big\{
[1+\asympt(\hat{\chi})]e^{-\xi_c(\hat{c},\lambda)/\epsilon}+
[1+\asympt(\hat{\chi})]e^{\xi_c(\hat{c},\lambda)/\epsilon}
\Big\}
e^{\xi_{\hat{c}}(x,\lambda)/\epsilon}.
\ee
But $\lambda$ is an eigenvalue of $\mathfrak{D}_\epsilon$ if and only if 
$w(x,\lambda,\epsilon)$ is decaying as $x\rightarrow+\infty$ which is equivalent 
with the fact that the quantity inside the braces in the last result is zero.  This easily leads to
\be\nn
n(\lambda,\epsilon)e^{2\xi_c(\hat{c})/\epsilon}=1
\ee
for some function $n(\lambda,\epsilon)$ so that 
we have 
\be\nn
n(\lambda,\epsilon)=-1+\asympt(\hat{\chi})
\quad\text{as}\quad
\epsilon\downarrow0.
\ee
This completes our proof.
\end{proof}

\section{Norming Constants}
\label{norming-constants}

Here, we present the results for the semiclassical behavior of the norming constants 
that correspond to the eigenvalues (see sections \ref{evs} and \ref{evs-near-zero}) of our 
Dirac operator. These norming constants, are the proportionality constants that relate 
the two eigenfunctions corresponding to a particular eigenvalue of our operator. We 
have the following result.

\begin{corollary}
\label{norm-const-corol}
Consider a $\lambda_0$ on the asymptotic spectral arcs satisfying the condition {\bf (H1)} 
of \S\ref{bs-at-two-turn-pts} and for some complex neighborhood $U$ of $\lambda_0$ 
take $\lambda\in U$ to be an eigenvalue of $\mathfrak{D}_\epsilon$. Then the norming 
constant corresponding to $\lambda$ satisfies the following asymptotics as 
$\epsilon\downarrow0$
\be\nn
\pm1+\asympt(\epsilon).
\ee
\end{corollary}
\begin{proof}
Suppose that ${\bf u}_0^\alpha$ and ${\bf u}_0^\beta$ are the eigenfunctions corresponding 
to the eigenvalue $\lambda$. These were defined explicitly in \S\ref{bs-at-two-turn-pts} as 
solutions of the problem (\ref{ev-problem}). From Lemma \ref{linear-depend-lemma-of-u}, 
and the formulae (\ref{wronsk1}), (\ref{wronsk2}) (or equivalently Theorem \ref{BS}) we find that
\be\nn
\frac{{\bf u}_0^\beta}{{\bf u}_0^\alpha}=
\frac{
{\mathcal W}[{\bf u}_2^\beta,{\bf u}_0^\beta]
}
{
{\mathcal W}[{\bf u}_2^\alpha,{\bf u}_0^\alpha]
}
e^{z(\beta,\lambda,\alpha)/\epsilon}=
\frac{
{\mathcal W}[{\bf u}_0^\beta,{\bf u}_1^\beta]
}
{
{\mathcal W}[{\bf u}_0^\alpha,{\bf u}_1^\alpha]
}
e^{-z(\beta,\lambda,\alpha)/\epsilon}.
\ee
This gives rise to the fact
\be\nn
\bigg(\frac{{\bf u}_0^\beta}{{\bf u}_0^\alpha}\bigg)^2=
\frac{
{\mathcal W}[{\bf u}_0^\beta,{\bf u}_1^\beta]{\mathcal W}[{\bf u}_2^\beta,{\bf u}_0^\beta]
}
{
{\mathcal W}[{\bf u}_0^\alpha,{\bf u}_1^\alpha]{\mathcal W}[{\bf u}_2^\alpha,{\bf u}_0^\alpha]
}.
\ee 
Now using the asymptotic formulas (\ref{asympt-wronsk-1}) and (\ref{asympt-wronsk-2}), we find 
that $({\bf u}_0^\beta/{\bf u}_0^\alpha)^2$ behaves like 
$-\delta(\alpha_0,\beta_0)+\asympt(\epsilon)$ 
as $\epsilon\downarrow0$. Hence as $\epsilon\downarrow0$, 
the norming constant satisfies
\be\nn
\frac{{\bf u}_0^\beta}{{\bf u}_0^\alpha}=
\pm1+\asympt(\epsilon)
\ee
since in all of our cases the turning points in our analysis are zeros of $g_-$ which consequently 
force $\delta$ to be $-1$ [cf. (\ref{dab})].
\end{proof}

\begin{remark}
The signs of the norming constants change consecutively from one eigenvalue to the next 
one as we move along the asymptotic spectral arcs on the spectral plane.
\end{remark}
\begin{remark}
For the norming constants of the eigenvalues that lie near zero, 
Theorem \ref{theorem-evs-near-zero-ala-olver} shows that we obtain the 
same result.
\end{remark}

\section{Reflection Coefficient}
\label{reflection}

In this section we will consider the reflection coefficient for our Dirac 
operator 
\be\nn
\mathfrak{D}_\epsilon=
\begin{bmatrix}
-\frac{\epsilon}{i}\frac{d}{dx} & -iA(x)\exp\{iS(x)/\epsilon\}\\
-iA(x)\exp\{-iS(x)/\epsilon\} & \frac{\epsilon}{i}\frac{d}{dx}
\end{bmatrix}.
\ee
where again $A(x)=\sech(2x)$ and also the phase function $S(x)=\sech(2x), x\in\R$. Recall that the  
continuous spectrum of our Dirac operator is the whole real line. So in this section we are 
considering $\lambda\in\R$. 

\subsection{Reflection away from zero.}
\label{away-0-refle}

We begin with the case where this $\lambda\in\R$ is \textit{independent} of $\epsilon$ and 
consider a $\delta>0$ so that $|\lambda|\geq\delta$. Under the (different) change of variables
\be\nn
W_{\pm}=
\frac{u_{2}\exp\{iA/(2\epsilon)\}\pm u_{1}\exp\{-iA/(2\epsilon)\}}
{\sqrt{A\pm i(\lambda+A'/2)}}
\ee
equation (\ref{ev-problem}) is transformed to the following equation 
(actually, by applying the transformation above, we get two independent equations; 
we only consider the case for the lower index and set $W=W_-$).

\be\label{final-scattter-eq}
\frac{d^2W}{dx^2}=[-\epsilon^{-2}f(x,\lam)+g(x,\lam)]y
\ee
where $f$ and $g$ are given by the following formulae
\be\label{f-tilde}
f(x,\lam)=A^2(x)+[\lambda+\tfrac{A'(x)}{2}]^2
\ee
and
\begin{align}
\nn
g(x,\lam)
&=
\frac{1}{2}\Big[\big(\log\{A(x)-i[\lambda+\tfrac{A'(x)}{2}]\}\big)'\Big]^2
-\frac{ \Big(\{A(x)-i[\lambda+\tfrac{A'(x)}{2}] \}^{1/2}\Big)''}
{ \{A(x)-i[\lambda+\tfrac{A'(x)}{2}] \}^{1/2} }
\\
\label{g-tilde}
&=
\frac{3}{4}\Bigg\{ \frac{A'(x)-i\frac{A''(x)}{2}}{A(x)-i[\lam+\tfrac{A'(x)}{2}]} \Bigg\}^2
-
\frac{1}{2}\frac{A''(x)-i\frac{A'''(x)}{2}}{A(x)-i[\lam+\tfrac{A'(x)}{2}]}
\end{align}
\noindent
We have the following definitions.
\begin{definition}
We define an \textbf{error-control function} $H(x,\lam)$ for equation 
(\ref{final-scattter-eq}), to be a primitive of
\be\label{error-control-scattering}
\frac{1}{f^{1/4}(x,\lam)}
\frac{d^2}{dx^2}
\Big[
\frac{1}{f^{1/4}(x,\lam)}
\Big]
-\frac{g(x,\lam)}{f^{1/2}(x,\lam)}.
\ee 
\end{definition}
\begin{definition}
We define the \textbf{variation} of $H$ in the interval 
$(x_1,x_2)\subseteq(0,+\infty)$ to be given by 
\be\label{var-control-scatt}
\mathcal{V}_{x_1,x_2}[H](\lam)=
\int_{x_1}^{x_2}
\bigg|
\frac{1}{f^{1/4}(t,\lam)}
\frac{d^2}{dx^2}
\Big[
\frac{1}{f^{1/4}(t,\lam)}
\Big]
-\frac{g(t,\lam)}{f^{1/2}(t,\lam)}
\bigg|
dt.
\ee
\end{definition}

Observe that (\ref{f-tilde}) implies $f>0$ in $\R$. Consequently, 
equation (\ref{final-scattter-eq}) has no turning points. Furthermore, 
notice that 
\begin{itemize}
\item
$g$ is complex-valued
\item
$f$ is twice continuously differentiable with respect to $x$ and
\item
$g$ is continuous.
\end{itemize}
These properties allow one (cf. Theorem $2.2$ of \S $2.4$ 
from chapter 6 of \cite{olver1997} along with the remarks from \S 5.1 
of the same chapter) to arrive at the following theorem.
\begin{theorem}
Take an arbitrary interval $(x_1, x_2)\subseteq\R$.  Then, equation 
(\ref{final-scattter-eq}) has in the above interval, two twice continuously 
differentiable solutions $w_\pm$ with
\be\nn
w_\pm(x,\epsilon)=
f^{-1/4}(x,\lam)
\exp\Big\{\pm\frac{i}{\epsilon}\int f^{1/2}(t,\lam)dt\Big\}
[1+r_\pm(x,\epsilon)].
\ee
The remainders $r_\pm$ satisfy
\be\label{error-term-scattering}
|r_\pm(x,\epsilon)|\hspace{5pt},\hspace{5pt}
\epsilon f^{-1/2}(x,\epsilon)
\Big|\frac{\partial r_\pm}{\partial x}(x,\epsilon)\Big|
\leq
\exp\{\epsilon\mathcal{V}_{\kappa,x}[H](\lam)\}-1
\ee
where $\kappa$ is an arbitrary (finite or infinite) point in the closure of $(x_1, x_2)$, 
provided that $\mathcal{V}_{\kappa,x}[H](\lambda)<+\infty$. 
\end{theorem}

\begin{remark}
Since $g$ is not real, we cannot expect the solutions $w_\pm$ to be complex 
conjugates.
\end{remark}
\begin{remark}
It follows that
\begin{itemize}
\item
$r_\pm(x,\epsilon)\rightarrow0$ as $x\to\kappa$ and
\item
$\epsilon f^{-1/2}(x,\epsilon)
\frac{\partial r_\pm}{\partial x}(x,\epsilon)\rightarrow0$ as 
$x\to\kappa$.
\end{itemize}
\end{remark}

Using (\ref{error-control-scattering}),  notice that $H$ is independent
of $\epsilon$ whence the right-hand side of (\ref{error-term-scattering}) is 
$\asympt(\epsilon)$ as $\epsilon\downarrow0$ and fixed $x$. But 
$\mathcal{V}_{x_1,x_2}[H](\lam)<+\infty$ which implies that 
this $\asympt$-term is uniform with respect to $x$ since
$\mathcal{V}_{\kappa,x}[H](\lam)\leq
\mathcal{V}_{x_1,x_2}[H](\lam)$.
Hence 
\be\nn
w_\pm(x,\epsilon)\sim
f^{-1/4}(x,\lam)
\exp\Big\{\pm\frac{i}{\epsilon}\int f^{1/2}(t,\lam)dt\Big\}
\quad\text{as}\quad
\epsilon\downarrow0
\ee
uniformly in $(x_1, x_2)$.

Next we define the \textit{Jost solutions}.  The Jost solutions are defined as 
the components of the bases $\{J_-^l, J_+^l\}$ and $\{J_-^r, J_+^r\}$ of the 
two-dimensional linear space of solutions of equation (\ref{final-scattter-eq}), 
which satisfy the asymptotic conditions
\begin{align*}
J_\pm^l(x,\lam) &\sim \exp\big\{\pm i\lambda x/\epsilon\big\}
\quad\text{as}\quad x\to-\infty\\
J_\pm^r(x,\lam) &\sim \exp\big\{\pm i\lambda x/\epsilon\big\}
\quad\text{as}\quad x\to+\infty.
\end{align*}
From scattering theory, we know that the \textit{reflection coefficient} 
$\mathcal{R}(\lam,\epsilon)$ for the waves incident on the potential from 
the right, can be expressed in terms of wronskians of the Jost 
solutions. More presicely, we have
\be
\label{r}
\mathcal{R}(\lam,\epsilon)=\frac{W[J_-^l, J_-^r]}{W[J_+^r, J_-^l]}.
\ee

The next step is to construct the Jost solutions as 
\textit{WKB solutions}. For this, we define the following four 
WKB solutions
\begin{align*}
\bar{w}_\pm^l(x,\epsilon)=
f^{-1/4}(x,\lam)
\exp\Big\{\pm\frac{i}{\epsilon}\Big(\lam x+\int_{-\infty}^x 
[f^{1/2}(t,\lam)-\lam]dt\Big)\Big\}
[1+\bar{r}_\pm^l(x,\epsilon)]\\
\bar{w}_\pm^r(x,\epsilon)=
f^{-1/4}(x,\lam)
\exp\Big\{\pm\frac{i}{\epsilon}\Big(\lam x+\int_{+\infty}^x 
[f^{1/2}(t,\lam)-\lam]dt\Big)\Big\}
[1+\bar{r}_\pm^r(x,\epsilon)]
\end{align*} 
which we are going to modify slightly in a while. If we take the 
limits as $x\to\pm\infty$ of the above, we instantly notice the
following relations between $\bar{w}_\pm^l, \bar{w}_\pm^l$ and 
the Jost solutions $J_\pm^l, J_\pm^r$; we have
\begin{align*}
J_\pm^l &=\lam^{1/2}\bar{w}_\pm^l\\
J_\pm^r &=\lam^{1/2}\bar{w}_\pm^r.
\end{align*}
Let now $w_\pm^l, w_\pm^r$ be four WKB solutions satisfying
\begin{align}
\label{wkb-left}
w_\pm^l(x,\epsilon)=
f^{-1/4}(x,\lam)
\exp\Big\{\pm\frac{i}{\epsilon}\int_{0}^x 
f^{1/2}(t,\lam)dt\Big\}
[1+r_\pm^l(x,\epsilon)]
\\
\label{wkb-right}
w_\pm^r(x,\epsilon)=
f^{-1/4}(x,\lam)
\exp\Big\{\pm\frac{i}{\epsilon}\int_{0}^x 
f^{1/2}(t,\lam)dt\Big\}
[1+r_\pm^r(x,\epsilon)].
\end{align} 
Once again, the connnection between $w_\pm^l, w_\pm^r$ and 
$\bar{w}_\pm^l, \bar{w}_\pm^r$ is evident.  Indeed, as $\epsilon\downarrow0$, we have
\begin{align*}
\bar{w}_\pm^l &=
\exp\Big\{\pm\frac{i}{\epsilon}\int_{-\infty}^0 
[f^{1/2}(t,\lam)-\lam]dt\Big\}
w_\pm^l\big(1+\text{o}(1)\big)\\
\bar{w}_\pm^r &= 
\exp\Big\{\mp\frac{i}{\epsilon}\int_0^{+\infty} 
[f^{1/2}(t,\lam)-\lam]dt\Big\}
w_\pm^r\big(1+\text{o}(1)\big).
\end{align*}
Subsequently, as $\epsilon\downarrow0$, for the Jost solutions we have
\begin{align}
\label{j-w-left}
J_\pm^l &=
\lam^{1/2}
\exp\Big\{\pm\frac{i}{\epsilon}\int_{-\infty}^0 
[f^{1/2}(t,\lam)-\lam]dt\Big\}
w_\pm^l\big(1+\text{o}(1)\big)\\
\label{right-j-w}
J_\pm^r &=
\lam^{1/2}
\exp\Big\{\mp\frac{i}{\epsilon}\int_0^{+\infty} 
[f^{1/2}(t,\lam)-\lam]dt\Big\}
w_\pm^r\big(1+\text{o}(1)\big).
\end{align} 

Recall that the properties of $A$ show that the
function $t\mapsto f^{1/2}(t,\lam)-\lam$ is in $L^1(\R)$.
Furthermore, we have
\be\label{scatter-norm}
\int_{-\infty}^0[f^{1/2}(t,\lam)-\lam]dt=
\int_0^{+\infty}[f^{1/2}(t,\lam)-\lam]dt=
\frac{1}{2}
\|f^{1/2}(\cdot,\lam)-\lam\|_{L^1(\R)}
\ee
and we define
\be
\label{sigma-norm}
\sigma(\lambda)=\|f^{1/2}(\cdot,\lam)-\lam\|_{L^1(\R)}.
\ee
Substituting (\ref{j-w-left}), (\ref{right-j-w}), (\ref{scatter-norm}) 
and (\ref{sigma-norm}) in (\ref{r}) we have
\be
\label{refle}
\mathcal{R}(\lam,\epsilon)=
e^{i\sigma(\lambda)/\epsilon}
\frac{W[w_-^l, w_-^r]}{W[w_+^r, w_-^l]}.
\ee

Finally,  using (\ref{wkb-left}), (\ref{wkb-right}) and 
(\ref{error-term-scattering}) we find that
\begin{itemize}
\item
$W[w_-^l, w_-^r]=\asympt(1)$ as $\epsilon\downarrow0$ and
\item
$W[w_+^r, w_-^l]=-\frac{2i}{\epsilon}+\asympt(1)$ as $\epsilon\downarrow0$.
\end{itemize}
Substituting these last results in (\ref{refle}) we get that
\be\nn
\mathcal{R}(\lam,\epsilon)=
e^{i\sigma(\lambda)/\epsilon}
\asympt(\epsilon)
\quad\text{as}\quad
\epsilon\downarrow0.
\ee
\noindent
Hence we have just proved the following theorem.
\begin{theorem}
Consider an arbitrary $\delta>0$.  Then the reflection coefficient 
of equation (\ref{final-scattter-eq}) as defined by (\ref{r}),  satisfies
\be
\mathcal{R}(\lam,\epsilon)=
\asympt(\epsilon)
\quad\text{as}\quad
\epsilon\downarrow0
\ee
uniformly for $|\lambda|\geq\delta$.  
\end{theorem}

\subsection{Reflection near zero.}
\label{near-0-refle}

Now we turn to the case where $\lambda$ depends on $\epsilon$ 
[$\lam=\lam(\epsilon)$] and particularly we let $\lambda$ approach $0$ 
like $\epsilon^b$ for an $\epsilon$-independent positive constant $b$. 
Using (\ref{f-tilde}) and (\ref{g-tilde}), we see that 
(\ref{error-control-scattering}) can be written as
\begin{align}
\label{formula-for-variation}
\nn
\frac{1}{f^{1/4}(x,\lam)}
&
\frac{d^2}{dx^2}
\Big[
\frac{1}{f^{1/4}(x,\lam)}
\Big]
-\frac{g(x,\lam)}{f^{1/2}(x,\lam)}=\\
&\nn
\hspace{10pt}
\frac{5}{16}\{A(x)^2+[\lambda+\tfrac{A'(x)}{2}]^2\}^{-5/2}
\{2A(x)A'(x)+[\lam+\tfrac{A'(x)}{2}]A''(x)\}^2\\
&\nn
-
\frac{1}{4}\{A(x)^2+[\lambda+\tfrac{A'(x)}{2}]^2\}^{-3/2}\\
&\nn
\hspace{2.15cm}
\cdot
\{2[A(x)A''(x)+A'(x)^2]+[\lam+\tfrac{A'(x)}{2}]A'''(x)+\tfrac{A''(x)^2}{2}\}\\
&\nn
-
\frac{3}{4}\{A(x)^2+[\lambda+\tfrac{A'(x)}{2}]^2\}^{-1/2}
\Bigg\{ \frac{A'(x)-i\frac{A''(x)}{2}}{A(x)-i[\lam+\tfrac{A'(x)}{2}]} \Bigg\}^2\\
&
+
\frac{1}{2}\{A(x)^2+[\lambda+\tfrac{A'(x)}{2}]^2\}^{-1/2}
\frac{A''(x)-i\frac{A'''(x)}{2}}{A(x)-i[\lam+\tfrac{A'(x)}{2}]}.
\end{align}
We can  easily see that each  of the  terms in the sum (\ref{formula-for-variation}) is less than
\begin{equation*}
C\lambda^{-5/2}e^{-\gamma |x|}
\end{equation*}
where $C,\gamma>0$, since $A(x)=\sech(2x)$ and $A(x), A'(x), A''(x), A'''(x) \sim\pm e^{-2|x|}$ 
as $x\rightarrow \pm\infty$. 
Recalling that $\lambda(\epsilon)\in[\epsilon^b,+\infty)$ where $b>0$ is independent of 
$\epsilon$,  and (\ref{error-term-scattering}) we see that the variation in 
(\ref{var-control-scatt}) behaves like 
\be\nn 
\mathcal{V}_{0,+\infty}[H](\lam(\epsilon))=
\asympt\Big(\epsilon^{-5b/2}\Big)
\quad\text{as}\quad\epsilon\downarrow0. 
\ee
Hence using (\ref{wkb-left}),  (\ref{wkb-right}) we get
\begin{itemize}
\item
$W[w_-^l, w_-^r]=\asympt(\epsilon^{-5b/2})$ 
as $\epsilon\downarrow0$ and
\item
$W[w_+^r, w_-^l]=
-\frac{2i}{\epsilon}+\asympt(\epsilon^{-1-5b/2})$ 
as $\epsilon\downarrow0$.
\end{itemize}
Substituting these last results in (\ref{refle}),  we finally obtain that
\be
\mathcal{R}(\lam(\epsilon),\epsilon)=
e^{i\sigma(\lambda(\epsilon))/\epsilon}
\asympt(\epsilon^{1-5b/2})
\quad\text{as}\quad
\epsilon\downarrow0.
\ee

So, we have  showed the following.
\begin{theorem}
\label{spectrum+scatter}
Consider $0<b<\tfrac{2}{5}$ (independent of $\epsilon$). Then the reflection 
coefficient of equation (\ref{final-scattter-eq}) as defined by (\ref{r}),  satisfies
\be
\label{final-r}
\mathcal{R}(\lam(\epsilon),\epsilon)=
\asympt(\epsilon^{1-5b/2})
\quad\text{as}\quad
\epsilon\downarrow0
\ee
uniformly for $\lambda(\epsilon)$ in any closed interval of $[\epsilon^b,+\infty)$.  
\end{theorem}

\begin{remark}
A similar (and better) estimate can be achieved by using the exact WKB method. See Corollary 2.8 
in \cite{fujii+kamvi}.
\end{remark}

\section{Conclusion: Application to the semiclassical focusing NLS}
\label{application-nls}

In a sense what we have been able to prove in the previous sections  is that the behavior of the solution of 
(\ref{ivp-nls}) in the regime $\epsilon \downarrow 0$ can be approximated by a 
``soliton ensemble", i.e a purely soliton solution. This solution is uniquely defined 
by the eigenvalues of the associated Dirac operator and their corresponding norming 
constants: the eigenvalues are exactly the WKB approximations of the previous sections 
while the norming constants alternate between -1 and 1. The role of the continuous part of 
the spectrum is negligible when epsilon is small, since the reflection coefficient is 
small away from zero
\footnote{We refer to \cite{h+k2} for the actual details justifying the approximation.}. 
The remaining question is the semiclassical behavior of the soliton ensemble itself.

This is a generalization of the problem considered in \cite{kmm} and \cite{kr}, the difference 
being that here the relevant WKB-eigenvalues lie in the \textsf{Y}-shaped set of Figure \ref{cut-the-L} 
(together with its reflection with respect to the real axis), while in \cite{kmm} the spectrum 
is purely imaginary (the four non-imaginary ``branches" are not there). Even though the proofs 
of \cite{kmm} and \cite{kr} are long and complicated they can be readily transferred from one 
case to another. 

We simply recall that the starting point of the strategy in \cite{kmm} consists of considering
a fairly arbitrary contour around the imaginary discrete spectrum and transforming the meromorphic 
Riemann-Hilbert problem (with no jump contour) given by the inverse scattering theory to a  
holomorphic Riemann-Hilbert problem with jump given on that fairly arbitrary contour. $Eventually$ 
this contour has to be deformed to an optimal contour on which the asymptotic analysis of the
holomorphic Riemann-Hilbert problem will be feasible.

The limiting density of eigenvalues is a crucial quantity in the analysis. 
In the special case $A(x)=S(x)=\sech(2x)$, for example,
\begin{equation}
\rho(\lam)=
\frac{\lam}{\pi}\int_{x_1(\lam)}^{x_2(\lam)}
\frac{dx}{(A(x)^2+(\lam + \frac{1}{2} S'(x))^2)^{1/2}}
\end{equation}
for $\lam \in \tilde \Lambda_{12}$, while
\begin{equation}
\rho(\lam)=
\frac{\lam}{\pi}\int_{x_1(\lam)}^{x_6(\lam)}
\frac{dx}{(A(x)^2+ (\lam + \frac{1}{2} S'(x))^2)^{1/2}}
\end{equation}
for $\lam \in \tilde \Lambda_{16}$ and
\begin{equation}
\rho(\lam)=
\frac{\lam}{\pi}\int_{x_2(\lam)}^{x_6(\lam)}
\frac{dx}{(A(x)^2+(\lam + \frac{1}{2} S'(x))^2)^{1/2}}
\end{equation}
for $\lam \in \tilde \Lambda_{26}$.
By symmetry, the last two integrals are equal and it is not hard to see that all three integrals 
are analytic in the upper half-plane, the top integral one being equal to the sum of the two others\footnote{Just write every 
integral as a contour integral on a closed curve wrapping around each $\tilde \Lambda_{ij}$ and then deform to 
a curve independent of $\lam$ to  get local analyticity. Then, as $\lam$ varies,
 start deforming  those closed curves along the moving cuts (actually Stokes curves moving with $\lam$)
to achieve global analyticity. The geometry of turning points is known numerically and this is crucial.
One can deform such that the number of turning points along each deformation remains two and cover the whole upper half-plane.
The only possibly problematic point is the bifurcation point $\lam_\otimes$ where a third turning point appears. 
But the linear relation between the three integrals above shows that
analyticity holds at that point as well.}. 
So $\rho$ is not continuous at the bifurcation point 
$\lam_\otimes$ if defined as a function on the \textsf{Y}-shaped set of Figure \ref{cut-the-L}
but it is a sum of two functions $\rho_1$ and $\rho_2$ such that $\rho_1$ 
is continuous in $\tilde \Lambda_{12} \cup \tilde \Lambda_{16}$ and $\rho_2$ is continuous 
in $\tilde \Lambda_{12} \cup \tilde \Lambda_{26}$. In fact, each $\rho_j$ admits (the same) analytic
extension, which is of course $\frac{1}{2} \rho$. We also note that because of the nice behavior of 
$A$ and $S$ at infinity, $\rho$ is regular at $0$. 

It follows eventually that the strategy of \cite{kmm} and \cite{kr} goes through by performing 
the first step of \cite{kmm} for a sum of two measures, one supported on $\tilde \Lambda_{12} \cup \tilde \Lambda_{16}$
and the other on $\tilde \Lambda_{12} \cup \tilde \Lambda_{26}$ each with an analytic density
with an entire extension. 
Even though the \textsf{Y}-shaped set of Figure \ref{cut-the-L} is more complicated than 
the imaginary interval appearing as an asymptotic spectrum in the zero phase case,
the symmetries involved (with respect to conjugation and taking negatives) ensure the
properties of the density $\rho$ needed in the analysis of  \cite{kr}. We leave  the details 
to the interested reader 
\footnote{Note that the imaginary part of $\rho$ is positive on the real line. This is the main  
condition in \cite{kr} for the admissibility of the contour maximizing the equilibrium energy.}
.

\appendix

\section{Modified Bessel Functions}
\label{mod_bessel}

\subsection{Primary solutions}

Let us in this section of the appendix describe some properties of Modified Bessel 
Functions (or MBFs) which are used in certain parts of this survey. For the proofs and a 
more detailed analysis on the material presented in this section, we refer the interested 
reader to \cite{olver-et-al} and \S2 of \cite{olver1978}. 

We start by investigating solutions of the following equation

\be
\label{mbf-basic-equation}
\frac{d^2w(t)}{dt^2}=\frac{9}{4}tw(t).
\ee
When $t$ is real we shall adopt the function
\be
\label{primary-sols-mbf}
U(t)=\Big(\frac{2t}{\pi}\Big)^{1/2}K_{1/3}(t^{3/2})
\ee
as a standard solution, where $K_{1/3}$ denotes the MBF of order $1/3$ in the usual 
notation. We continue to employ this solution in the more general case of $t$ being 
a complex number. The functions on the RHS of  (\ref{primary-sols-mbf}) are understood 
to assume their principal values when $\ph t=0$ and be defined by continuity for other 
values of $\ph t$.

\subsection{Secondary solutions}

Let us fix some notation and define some \textbf{sectors} on the complex plane. So, for 
$j\in\{0,1,2\}$ we set 
\be\nn
S_j=
\{\, 
x\in\C
\mid
(2j-1)\tfrac{\pi}{3}\leq\ph x\leq(2j+1)\tfrac{\pi}{3}
\,\}.
\ee
If there exists a number $0<\delta<1$ we say that the set
\be\nn
{\rm int}_\delta(S_j)=
\{\, 
x\in\C
\mid
(2j-\delta)\tfrac{\pi}{3}\leq\ph x\leq(2j+\delta)\tfrac{\pi}{3}
\,\}
\ee
is an \textbf{internal part} of the sector $S_j$.

We now introduce the secondary solutions for equation (\ref{mbf-basic-equation}). Here, we 
treat $t$ as a complex number. These are defined by
\be\nn
U_j(t)=U(te^{-i2\pi j/3}),\quad j\in\{0,1,2\}.
\ee
It can be easily observed that $U_j$ is a solution of equation (\ref{mbf-basic-equation}) 
that is decaying in the sector $S_j$, as $t\rightarrow\infty$.

For $j,k\in\{0,1,2\}$, the wronskian of the pair $\{U_j,U_k\}$ is given by 
\begin{align*}
W[U_j,U_k]
&=
3ie^{-i(j+k)\pi/3}\lambda_{j,k}\quad\text{where}\\
\lambda_{j,k}
&=
\frac{\sin[(k-j)\pi/3]}{\sin\pi/3}
\end{align*}
From this, we conclude that $U_j$, $U_k$ are linearly independent when 
$j,k\in\{0,1,2\}$ with $j\neq k$.

\subsection{Auxiliary functions}

\underline{Note}: throughout this section, we consider $j,k\in\{0,1,2\}$ with $j\neq k$ 
(so that $\lambda_{j,k}\neq0$).

First of all, it is well known that as $t\rightarrow\infty$, the solution $U_{j}$ is decaying
in $S_j$ and dominant in $S_k$ and by symmetry, $U_k$ is dominant in $S_j$ and 
decaying in $S_k$. This in turn, implies that $\{U_j, U_k\}$ comprises a numerically 
satisfactory pair of solutions in the closed region $S_j\cup S_k$ except possibly in a 
neighborhood of $0$. Hence, the pair above constitutes an appropriate solution basis 
in $S_j\cup S_k$ and in order to majorize these solutions, we shall introduce auxiliary 
weight, modulus and phase functions.

We start by defining the function $E(t)$ by the formula
\be\nn
E(t)=|\exp\{(-1)^j t^{3/2}\}|
\quad\text{for}\quad
t\in S_j
\ee 
where the branch of $t^{3/2}$ is $|t|^{3/2}\exp\{i3(\ph t)/2\}$.On the boundaries of 
$S_j$, $E(t)=1$ and so $E$ is continuous everywhere with $E(t)\geq1$. As weight 
functions in $S_j\cup S_k$ we assign the functions $\esf_{j,k}(t)$ defined by
\be\nn
\esf_{j,k}(t)=
\begin{cases}
1/E(t),\quad t\in S_j\\
E(t),\quad t\in S_k
\end{cases}
\ee 
It is a continuous function with $\esf_{j,k}(0)=1$. Also it satisfies the following symmetry
\be\nn
\esf_{k,j}(t)=1/\esf_{j,k}(t)
\quad\text{for}\quad
t\in S_j\cup S_k.
\ee

For the modulus and phase functions in $S_j\cup S_k$ we define
\begin{align*}
|U_j(t)|
&=
\esf_{j,k}(t)\msf_{j,k}(t)\cos\thsf_{j,k}(t)\\
|U_j'(t)|
&=
\esf_{j,k}(t)\nsf_{j,k}(t)\cos\omsf_{j,k}(t)
\end{align*}
and
\begin{align*}
|U_k(t)|
&=
\frac{1}{\esf_{j,k}(t)}\msf_{j,k}(t)\sin\thsf_{j,k}(t)\\
|U_k'(t)|
&=
\frac{1}{\esf_{j,k}(t)}\nsf_{j,k}(t)\sin\omsf_{j,k}(t)
\end{align*}
From these, we arrive at
\begin{align*}
\msf_{j,k}(t)
&=
\bigg[
\frac{|U_j(t)|^2}{\esf_{j,k}(t)^2}+
\esf_{j,k}(t)^2|U_k(t)|^2
\bigg]^{1/2}
\\
\nsf_{j,k}(t)
&=
\bigg[
\frac{|U_j'(t)|^2}{\esf_{j,k}(t)^2}+
\esf_{j,k}(t)^2|U_k'(t)|^2
\bigg]^{1/2}
\\
\thsf_{j,k}(t)
&=
\arctan
\bigg[
\esf_{j,k}(t)^2
\bigg|\frac{U_k(t)}{U_j(t)}\bigg|
\bigg]
\\
\omsf_{j,k}(t)
&=
\arctan
\bigg[
\esf_{j,k}(t)^2
\bigg|\frac{U_k'(t)}{U_j'(t)}\bigg|
\bigg]
\end{align*} 
The above formulae show that each one of $\msf_{j,k}(t), \nsf_{j,k}(t), \thsf_{j,k}(t)$ 
and $\omsf_{j,k}(t)$ is continuous in $S_j\cup S_k$. Also we have the following 
symmetry relations
\begin{align*}
\msf_{j,k}(t)
&=
\msf_{k,j}(t)
\\
\nsf_{j,k}(t)
&=
\nsf_{k,j}(t)
\\
\thsf_{j,k}(t)
+&
\thsf_{k,j}(t)
=\frac{\pi}{2}
\\
\omsf_{j,k}(t)
+&
\omsf_{k,j}(t)
=\frac{\pi}{2}
\end{align*} 

In  internal parts of $S_j$, $S_k$, for  large $|t|$ the asymptotic behavior of the modulus
functions is given by
\begin{align*}
\msf_{j,k}(t)
&\sim
(1+\lambda_{j,k}^2)^{1/2}|t|^{-1/4}
\\
\nsf_{j,k}(t)
&\sim
\frac{3}{2}(1+\lambda_{j,k}^2)^{1/2}|t|^{1/4}.
\end{align*}
But in the full domain $S_j\cup S_k$, the modulus and phase functions fluctuate 
as $t\rightarrow\infty$. For $\msf_{j,k}(t)$ we have the following bound
\be\nn
\msf_{j,k}(t)
\leq
C_{j,k}|t|^{-1/4}[1+\Theta(t^{3/2})]
\quad\text{for}\quad
t\in S_j\cup S_k
\ee
where $\Theta(t)$ and $C_{j,k}$ are given by 
\begin{align}
\label{big-theta}
\Theta(t)
&=
\exp\{5\pi/(72|t|)\}-1
\\
\label{big-c-jk}
C_{j,k}
&=
\max\{
[1+(|\lambda_{j,k}|+|\lambda_{j,k+1}|)^2]^{1/2}, 
[1+(|\lambda_{j,k}|+|\lambda_{j,k-1}|)^2]^{1/2}
\}.
\end{align}

Unfortunately, no such simple bound is available for $\nsf_{j,k}(t)$ and for this we shall 
introduce another pair of modulus and phase functions that will be convenient to work 
with. This is defined by
\begin{align*}
\bigg|\frac{d[t^{1/4}U_j(t)]}{dt}\bigg|
&=
|t|^{1/4}
\esf_{j,k}(t)
\hat{\nsf}_{j,k}(t)
\cos\hat{\omsf}_{j,k}(t)\\
\bigg|\frac{d[t^{1/4}U_k(t)]}{dt}\bigg|
&=
|t|^{1/4}
\frac{1}{\esf_{j,k}(t)}
\hat{\nsf}_{j,k}(t)
\sin\hat{\omsf}_{j,k}(t)
\end{align*}
Immediately, we find that 
\begin{align*}
\hat{\nsf}_{j,k}(t)
&=
|t|^{-1/4}
\bigg\{
\frac{1}{\esf_{j,k}(t)^2}
\bigg|\frac{d[t^{1/4}U_j(t)]}{dt}\bigg|^2+
\esf_{j,k}(t)^2
\bigg|\frac{d[t^{1/4}U_k(t)]}{dt}\bigg|^2
\bigg\}^{1/2}
\\
\hat{\omsf}_{j,k}(t)
&=
\arctan
\bigg\{
\esf_{j,k}(t)^2
\bigg|\frac{d[t^{1/4}U_k(t)]}{dt}\bigg|\Big/
\bigg|\frac{d[t^{1/4}U_j(t)]}{dt}\bigg|
\bigg\}.
\end{align*} 
Symmetry now tell us that 
\begin{align*}
\hat{\nsf}_{j,k}(t)
&=
\hat{\nsf}_{k,j}(t)
\\
\hat{\omsf}_{j,k}(t)
+&
\hat{\omsf}_{k,j}(t)
=\frac{\pi}{2}
\end{align*} 
while as $t\rightarrow\infty$ in internal parts of $S_j$ or $S_k$ we have
\be\nn
\hat{\nsf}_{j,k}(t)
\sim
\frac{3}{2}(1+\lambda_{j,k}^2)^{1/2}|t|^{1/4}.
\ee
Lastly,  $\hat{\nsf}_{j,k}(t)$ satisfies the following bound
\be\nn
\hat{\nsf}_{j,k}(t)
\leq
\frac{3}{2}
C_{j,k}|t|^{1/4}[1+\Theta(t^{3/2})]
\quad\text{for}\quad
t\in S_j\cup S_k.
\ee
where $\Theta(t)$ and $C_{j,k}$ are given by (\ref{big-theta}) and (\ref{big-c-jk}) 
respectively.

\section*{Data Availability}

Data sharing is not applicable to this article as no new data were
created or analyzed in this study.

\section*{Acknowledgements} 

We would like to thank \href{http://users.math.uoc.gr/~nefrem/}{Nikos Efremidis} for 
Figures \ref{efrem-zero-phase} and \ref{efrem-non-zero-phase}  in the Introduction. 
Figures \ref{curves-in-lambda-upper-half-plane} through \ref{no-wkb-5} have been 
taken from \cite{mil}. Also
\begin{itemize}
\item
The research of the first author was supported by a 
\href{https://www.jsps.go.jp/english/}{JSPS} grant-in-aid No. 21K03303. 
\item
The second author acknowledges the support of the Institute of Applied 
and Computational Mathematics of the Foundation of Research and 
Technology - Hellas (\href{https://www.forth.gr/}{FORTH}), via grant
MIS 5002358 and the support of the \href{https://en.uoc.gr}{University of Crete} 
via grant 10753. He also expresses his sincere gratitude to the Independent Power 
Transmission Operator(\href{https://www.admie.gr/en}{IPTO}) for a scholarship 
through the \href{http://www.sse.uoc.gr/en/}{School of Sciences and Engineering} 
of the University of Crete. Finally he acknowledges the support of the 
\href{https://www.fulbright.gr/en/}{Fulbright Institute in Greece} through a 
Visiting Researcher Fellowship at the 
\href{https://www.math.princeton.edu}{Department of Mathematics, Princeton University}.
\item
The third  author acknowledges the generous support of 
\href{http://en.ritsumei.ac.jp}{Ritsumeikan University} during a visit in $2018$ and 
the support of the \href{https://gsri.gov.gr}{Greek Secretariat of Research and Technology} 
under ARISTEIA II Grant No.\ 3964.
\end{itemize}

\end{document}